\documentclass[10pt]{article}
\usepackage{amsmath,amssymb,amsthm}
\newtheorem{proposition}{Proposition}
\newtheorem{corollary}{Corollary}
\usepackage{longtable}
\usepackage{booktabs,xcolor}
\usepackage{tikz}
\usepackage{pgfplots} \pgfplotsset{compat=1.18} 
\usepackage{fontawesome}
\usetikzlibrary{shapes.geometric, arrows}
\tikzstyle{block} = [rectangle, rounded corners, minimum width=3cm, minimum height=1cm, text centered, draw=black, fill=blue!20]
\tikzstyle{arrow} = [thick,->,>=stealth]
\usepackage{multirow}
\usepackage[margin=1in]{geometry}
\usepackage{xcolor}
\usetikzlibrary{shapes,arrows.meta, positioning,calc}
\usepackage{caption}
\usepackage{subcaption}
\tikzstyle{block2} = [
  rectangle, 
  rounded corners, 
  minimum width=2.8cm, 
  minimum height=1.1cm,
  text centered, 
  draw=black, 
  fill=blue!10,
  text width=2.8cm,
  align=center
]
\tikzstyle{data2} = [
  trapezium, 
  trapezium left angle=70, 
  trapezium right angle=110, 
  minimum height=1.1cm, 
  text centered, 
  draw=black, 
  fill=green!20,
  text width=2.5cm,
  align=center
]
\begin{document}

\title{ Breaking the Barriers of Text-Hungry and  Audio-Deficient  AI}

\author{ Hamidou Tembine, Issa Bamia, Massa NDong,  Bakary Coulibaly \\ Oumar Issiaka Traor\'e,  Moussa Traor\'e,  Moussa Sanogo,  \\ Mamadou Eric Sangar\'e, Salif Kant\'e, Daryl Noupa Yongueng,  \\  Hafiz Tiomoko Ali,  Malik Tiomoko, Fréjus Laleye  \\ Boualem Djehiche,   Wesmanegda Elisee Dipama,  Idris Baba Saje, \\ Hammid Mohammed  Ibrahim,  Moumini Sanogo, Marie Coursel Nininahazwe,\\   Abdul-Latif  Siita,   Haine Mhlongo,  Teddy Nelvy Dieu Merci Kouka, \\ Mariam Serine Jeridi, Mutiyamuogo  Parfait Mupenge, Lekoueiry Dehah, \\ Abdoul Aziz Bio Sidi D Bouko, Wilfried Franceslas Zokoue,  Odette Richette Sambila,\\ Alina RS Mbango, Mady Diagouraga, Oumarou  Moussa Sanoussi, Gizachew Dessalegn,\\ Mohamed Lamine Samoura, Bintou Laetitia Audrey Coulibaly  }

\maketitle

\begin{abstract}

While global linguistic diversity spans more than 7164 recognized languages, the current dominant architecture of machine intelligence remains fundamentally biased toward written text. This bias excludes over 700 million people particularly in rural and remote regions who are audio-literate. In this work, we introduce a fully textless, audio-to-audio machine intelligence framework designed to serve this underserved  population, and all the people who prefer audio-efficiency.
Our contributions include novel  Audio-to-Audio  translation architectures that bypass text entirely, including spectrogram-, scalogram-, wavelet-, and unit-based models. Central to our approach is the Multiscale Audio-Semantic Transform (MAST), a representation that encodes tonal, prosodic, speaker, and expressive features. We further integrate MAST into a fractional diffusion of mean-field-type framework powered by fractional Brownian motion. It enables the generation of high-fidelity, semantically consistent speech without reliance on textual supervision.
The result is a robust and scalable system capable of learning directly from raw audio, even in languages that are unwritten or rarely digitized. This work represents a fundamental shift toward  audio-native machine intelligence systems, expanding access to language technologies for communities historically left out of the current machine intelligence ecosystem. 
\end{abstract}
\newpage 
\tableofcontents 
\newpage
\section{Introduction}
Given the global population of approximately 8.2 billion as of March 2025, the potential for linguistic variation extends well beyond the 7,164 recognized languages, as each individual possesses unique vocal characteristics and may utilize multiple dialects or sign languages. Text-based machine intelligence proves ineffective in numerous village and remote settings where audio-literacy predominates.  While basic education is fundamental, text-based machine intelligence  currently excludes roughly 700 million audio-literate individuals, across the globe.  Therefore, co-developing a textless audio-to-audio machine intelligence  is critical. It will enable collaborative learning between parents and children in their native languages, contextualized within their environments. This approach also facilitates local commerce and knowledge sharing. It has the potential to empower these 700 million individuals through a more inclusive technological co-development. Africa alone has  400 million  audio-literates, which   presents a unique business opportunity for textless, audio-to-audio applications across diverse sectors: electric vehicles require voice-driven interfaces in local languages  for wider adoption; autonomous vehicles demand clear audio navigation  in local languages for safety; tourism benefits from multilingual audio tours; e-commerce expands through voice-activated shopping; and military interventions rely on  audio communication with the local  occupants of the intervention zone. By prioritizing audio-first design, businesses can tap into a vast, underserved market, driving economic growth and social impact across the continent. Additional business applications of Audio2audio machine intelligence \cite{basar2024foundations,basar2024applications,tapo2024machine,tembine2024mean,tembine2023machine}
 are provided in Tables \ref{tab:trends_mali}, \ref{tab:trends_mali2}, \ref{tab:trends3}, 
\ref{tab:audio2audiot0},\ref{tab:audio2audiot1}, \ref{tab:audio2audio} and \ref{tab:audio2audio2}.

\newpage
\begin{table}[htb]
\centering
\caption{Impact of Multilingual Audio Machine Intelligence on Key Malian Sectors }
\label{tab:trends_mali}
\begin{tabular}{|p{2cm}|p{2cm}|p{2cm}|p{3cm}|p{2.5cm}|p{2.5cm}|}
\hline
\textbf{Sector} & \textbf{Period} & \textbf{Languages} & \textbf{Key Applications} & \textbf{Emerging Professions} & \textbf{Required Skills} \\
\hline

\multirow{3}{*}{Electricity}
& Short-term & Bambara &
Vocal fault diagnosis via MI assistants &
Vocal Energy Technician &
Electrotechnics, voice recognition \\
& Medium-term & Songhai &
Vocal management of solar micro-grids &
Community Energy Manager &
Energy blockchain, multilingualism \\
& Long-term & All languages &
Automatic negotiation of energy tariffs &
MI Energy Negotiator &
Energy law, predictive algorithms \\
\hline

\multirow{3}{*}{Livestock}
& Short-term & Fulani/Bozo &
Vocal tracking of herds by animal recognition &
Digital Herder &
Vocal databases \\
& Medium-term & Senoufo &
Vocal alerts for livestock diseases &
Mobile MI Veterinarian &
Animal health, audio chatbots \\
& Long-term & All languages &
Vocal futures markets for livestock &
Digital Livestock Broker &
Islamic finance, market analysis \\
\hline

\multirow{3}{*}{Food}
& Short-term & Dogon &
Quality control by vocal analysis of processes &
MI Production Supervisor &
Standards, IoT sensors \\
& Medium-term & Songhai &
Vocal optimization of value chains &
Vocal Food Engineer &
Reverse logistics, traceability \\
& Long-term & All languages &
Authentic brands certified by vocal MI &
Digital Certification Agent &
Food law, NFTs \\
\hline
\multirow{3}{*}{Informal Finance}
& Short-term & Bobo &
"Tontines" managed by vocal assistants &
Oral Fintech Mediator &
Local trust systems \\
& Medium-term & Fulani &
Vocal credit rating via reputation analysis &
MI Microcredit Advisor &
Psycho-sociology \\
& Long-term & All languages &
Decentralized vocal banks &
Community Financial Auditor &
Smart contracts, regulation \\
\hline
\end{tabular}
\end{table}

\newpage 
\begin{table}[htb]
\centering
\caption{Impact of Multilingual Audio MI on Key Malian Sectors }
\label{tab:trends_mali2}
\begin{tabular}{|p{2cm}|p{2cm}|p{2cm}|p{4cm}|p{2cm}|p{2cm}|}
\hline
\textbf{Sector} & \textbf{Period} & \textbf{Languages} & \textbf{Key Applications} & \textbf{Emerging Professions} & \textbf{Required Skills} \\
\hline

\multirow{3}{*}{Fishing}
& Short-term & Bozo &
Vocal alerts for water levels and stocks &
Digital Fish Warden &
Hydrological data \\
& Medium-term & Bobo &
Automated vocal fish markets &
MI Fish Trader &
Cross-border trade \\
& Long-term & All languages &
Ecosystem management by acoustic MI &
Digital River Ecologist &
Bioacoustics, fisheries law \\ \hline

\multirow{3}{*}{Mining}
& Short-term & Senoufo &
Vocal reporting of mining accidents &
MI Mining Safety &
Mining law, sensors \\
& Medium-term & Dogon &
Vocal negotiation of contracts with communities &
Digital Mining Mediator &
Anthropology, blockchain \\
& Long-term & All languages &
Mineral exchange driven by vocal MI &
Mineral Resources Broker &
Circular economy, Web3 \\ \hline

\multirow{3}{*}{Electric Vehicles}
& Short-term & Fulani &
Predictive maintenance by vocal diagnostics &
Intelligent EV Mechanic &
Lithium batteries, vocal \\
& Medium-term & Bozo &
Vocal management of solar charging stations &
Green Mobility Manager &
Smart grids \\
& Long-term & All languages &
Autonomous vocal carpooling systems &
African Mobility Architect &
Conversational MI, UX \\ \hline

\multirow{3}{*}{Blockchain}
& Short-term & Sonhrai &
Vocal recording of land on blockchain &
Community Digital Notary &
Oral land law \\
& Medium-term & Senoufo &
Vocalized smart contracts for commerce &
Vocal Web3 Developer &
Solidity, local languages \\
& Long-term & All languages &
Decentralized governance by vocal DAOs &
Digital Society Facilitator &
Algorithmic governance \\
\hline
\end{tabular}
\end{table}

\newpage 
\begin{table}[htb]
\centering
\caption{Audio MI in Business in Mali}
\label{tab:trends3}
\begin{tabular}{|p{2cm}|p{2cm}|p{4cm}|p{2cm}|p{2cm}|p{2cm}|}
\hline
\textbf{Sector} & \textbf{Period} & \textbf{Key Applications} & \textbf{Emerging Professions} & \textbf{Required Skills} \\
\hline
\multirow{3}{*}{Agriculture}
& Short-term &
Vocal assistants for access to market data (prices, weather) &
Agricultural Data Analyst &
Audio MI, basics of agronomy \\
& Medium-term &
Vocal management of stocks and supply chain &
Vocal Logistics Manager &
Vocal management of the chain \\
& Long-term &
Virtual markets driven by voice commands &
Digital Agricultural Negotiator &
Blockchain, vocal e-commerce \\
\hline
\multirow{3}{*}{Commerce}
& Short-term &
Vocal kiosks for informal markets &
Digital Mediation Agent &
Multi-lingual \\
& Medium-term &
Vocal analysis of customer behavior &
Customer Experience Consultant &
Visualization \\
& Long-term &
Hyper-personalized advertising via voice recognition &
Vocal Marketing Strategist &
Consumer psychology, MI ethics \\
\hline
\multirow{3}{*}{Fintech}
& Short-term &
Mobile transactions by voice command &
Financial Inclusion Agent &
Transaction security \\
& Medium-term &
Vocal assistants for micro-credits &
Vocal Financial Advisor &
Risk analysis, regulation \\
& Long-term &
Autonomous vocal banks &
Financial MI Systems Auditor &
Cybersecurity, blockchain \\
\hline
\multirow{3}{*}{Logistics}
& Short-term &
Vocal tracking of deliveries &
Logistics Coordinator &
Crisis management, geolocation \\
& Medium-term &
Warehouses automated by audio I &
Intelligent Warehouse Supervisor &
Predictive maintenance \\
& Long-term &
Self-managed transport fleets &
Network Optimizer &
IoT, real-time analysis \\
\hline
\end{tabular}
\end{table}

\newpage

\begin{table}[htb]
\centering
\caption{Audio-to-Audio MI for Connecting Demand \& Supply in Regional Agri-Trade}
\label{tab:audio2audiot0}
\begin{tabular}{|p{3cm}|p{3cm}|p{6.5cm}|}
\hline
\textbf{Country/Region} & \textbf{Product} & \textbf{Audio-MI Application \& Impact} \\
\hline
\textbf{Dogon Country, Mali} & Onion/Shallot & Voice-MI alerts in Tommo-So Dogon languages notify producers of price surges in Bamako/Dakar/Abidjan markets, enabling bulk sales during shortages. \\
\hline
\textbf{Ansongo, Mali} & Onion/Shallot & Solar-powered voice platforms connect women growers to buyers in Niger, bypassing exploitative middlemen. \\
\hline
\textbf{Sikasso, Mali} & Mango & Bambara-language voice chatbots negotiate export contracts with Ivorian buyers, reducing post-harvest losses. \\
\hline
\textbf{Koutiala, Mali} & Cotton & Voice-MI logistics tools in Senoufo/Bambara help cooperatives track shipments to textile hubs like Segou, Dakar, Abidjan. \\
\hline
\textbf{San, Mali} & Shea Butter & Audio quality guides in Soninke validate purity standards for EU/US,  buyers, empowering women-led cooperatives. \\
\hline
\textbf{Mopti, Mali} & Rice & Voice alerts in Fula inform farmers about optimal irrigation schedules and Bamako market demand peaks. \\
\hline
\textbf{Thiès, Senegal} & Mango & Wolof voice-MI systems coordinate harvests with Dakar exporters, ensuring fair pricing for smallholders. \\
\hline
\textbf{Senegal} & Peanut & Pulaar-language voice apps link women processors to bulk buyers in Touba during religious festival demand spikes. \\
\hline
\textbf{Saint Louis, Senegal} & Fishing & Wolof audio platforms share real-time catch data with fishmongers in Nouakchott (Mauritania), reducing waste. \\
\hline
\textbf{Bobo-Dioulasso, Burkina Faso} & Sorghum Beer & Senoufo-language voice logs manage sorghum supply chains for women brewers, stabilizing production costs. \\
\hline
\textbf{Lake Chad, Niger} & Sheep (Tabaski) & Zarma/Hausa voice updates alert herders to seasonal demand spikes in Niamey and  Bamako for Eid al-Adha. \\
\hline
\end{tabular}
\end{table}

\newpage
\begin{table}[htb]
\centering
\caption{Audio-to-Audio MI Applications for Small Businesses \& Women Empowerment in Mali, Senegal, Burkina Faso, and Niger}
\label{tab:audio2audiot1}
\begin{tabular}{|p{2cm}|p{5cm}|p{6.5cm}|}
\hline
\textbf{Country} & \textbf{Application \& Business Sector} & \textbf{Empowerment Impact \& Example} \\
\hline
\textbf{Mali} & \textbf{Textile Cooperatives}: Voice-guided weaving tutorials in Bambara for traditional Bogolan cloth. & Women artisans in Ségou use audio prompts to learn complex patterns, scaling production for international markets. \\
\hline
\textbf{Senegal} & \textbf{Fish Trading}: Wolof/Fula voice marketplaces for pricing and logistics in coastal communities. & Female fish traders in Saint-Louis negotiate prices via voice-MI apps, bypassing male-dominated intermediaries. \\ &  \textbf{Basic Education}&  Primary school teachers with real time translation of their speech during lessons  \\
&
\textbf{Adult Lifelong Learning } & Adult learning in their native language  \\
\hline
\textbf{Burkina Faso} & \textbf{Shea Butter Production}: Audio quality control guides in Dioula/Mooré for rural cooperatives. & Women in Bobo-Dioulasso validate shea butter grades using voice-MI tools, improving export compliance. \\
\hline
\textbf{Niger} & \textbf{Agro-Pastoral Advisory}: Zarma/Hausa audio alerts for crop-livestock integration. & Women farmers in Maradi receive drought-resistant millet planting tips via solar-powered voice devices. \\
\hline
\textbf{Mali} & \textbf{Local Restaurants}: Bambara voice menus and hygiene reminders for street food vendors. & Women-run "maquis" in Bamako use voice-MI to manage orders during Ramadan rush hours. \\
\hline
\textbf{Senegal} & \textbf{Microloans}: Pulaar-language voice interfaces for audio-literate women to access credit. & Women in Thiès apply for loans via voice-MI chatbots, avoiding bureaucratic paperwork. \\
\hline
\textbf{Burkina Faso} & \textbf{Handicraft Export}: Voice-driven inventory tracking in Gulmancema for artisanal groups. & Female pottery cooperatives in Ouagadougou manage orders via audio logs, doubling sales efficiency. \\
\hline
\textbf{Niger} & \textbf{Maternal Health}: Tamajaq-language audio assistants for prenatal care in nomadic communities. & Women entrepreneurs distribute voice-MI health kits to remote Tuareg populations, reducing maternal mortality. \\
\hline
\end{tabular}
\end{table}
\newpage
\begin{table}[htb]
\centering
\caption{Applications of Audio-to-Audio MI in Local African Businesses}
\label{tab:audio2audio}
\begin{tabular}{|p{2cm}|p{5cm}|p{6cm}|}
\hline
\textbf{Application Area} & \textbf{Description} & \textbf{Real-World Example} \\
\hline
Electric Vehicles & Voice-driven interfaces for vehicle controls and navigation in local languages. & Drivers in Mali using Bambara voice commands to operate EVs (remaining useful  life of batteries depends on temperature, usage and voltage). \\
\hline
Autonomous Vehicles & Audio navigation systems for safety and route guidance in local dialects. & Delivery drones providing real-time audio instructions in Yoruba for rural logistics. \\
\hline
Tourism & Multilingual audio tours for cultural/heritage sites. & Tourists in Kenya receiving Swahili audio guides at Maasai Mara National Reserve. \\
\hline
E-Commerce & Voice-activated shopping platforms for non-literate users & Hausa-speaking customers in Nigeria ordering goods via voice commands on local apps. \\
\hline
Military/ Peacekeeping & Real-time audio communication with local communities during operations. & Soldiers using Fulfulde audio translation to interact with residents in conflict zones. \\
\hline
Local Market Transactions & Voice-based negotiation/payment systems for informal markets. & Traders in Ghana using Akan-language voice bots to process payments at open markets. \\
\hline
Agricultural Advisory & Voice-driven crop/weather updates for farmers. & Farmers in Senegal receiving Wolof audio alerts about pest outbreaks or rainfall. \\
\hline
Community Education & Interactive audio lessons for skill-sharing in native languages. & Parents and children in Ethiopia learning via Oromo-language audio storytelling apps. \\
 \hline
Livestock Breeding & Voice-based alerts for disease prevention, feeding schedules, and market prices. & Fulani herders in Niger receiving Zarma-language updates on livestock vaccines via community radio MI systems. \\
\hline
\end{tabular}
\end{table} 

\newpage
\begin{table}[htb]
\centering
\caption{Applications of Audio-to-Audio MI in Local African Businesses}
\label{tab:audio2audio2}
\begin{tabular}{|p{2cm}|p{5cm}|p{6cm}|}
\hline
\textbf{Application Area} & \textbf{Description} & \textbf{Real-World Example} \\
\hline
Local Transportation & Real-time audio dispatch systems for informal transit (e.g., minibuses, motorcycle taxis). & "Boda-boda" drivers in Uganda receiving Luganda voice alerts about traffic or passenger pickups. \\
\hline
Local Distribution & Voice-guided inventory management for small-scale suppliers and wholesalers. & Distributors in Côte d'Ivoire using Dioula-language voice logs to track stock levels in warehouses. \\
\hline
Restaurants \& Street Food & Voice-activated order systems and hygiene compliance reminders. & Millet fufu vendors in Bamako using Tommo-So Dogon voice commands to manage orders during peak hours. \\
\hline
Artisanal Crafting & Audio tutorials for traditional craft techniques (weaving, pottery). & Kente cloth weavers in Ghana learning new patterns via Akan-language instructional audio. \\
\hline
Healthcare Outreach & Multilingual voice assistants for symptom checks and clinic appointments. & Community health workers in Malawi using Chichewa audio tools to triage patients in remote villages. \\
\hline
Waste Management & Voice-guided sorting instructions for recycling initiatives. & Waste collectors in Rwanda receiving Kinyarwanda audio prompts to categorize recyclables. \\
\hline
Fishing Communities & Weather and market alerts for river and lake-based fisheries. & Bozo-speaking fishermen in Mali accessing real-time river/ lake conditions via voice updates on basic phones. \\
\hline
\end{tabular}
\end{table} 
\newpage

\subsection*{ Audio-rich languages are not necessarily of low  resource languages}
The misconception that audio-rich languages are low-resource stems from a conflation of phonological complexity with text-data scarcity. Audio-richness, characterized by features like high phoneme inventories, tonal variations, and a dominant oral tradition, describes a language's inherent structural complexity and reliance on spoken communication, not necessarily the availability of text-computational resources. 
Conversely, the term "low text-resource" traditionally refers to a lack of annotated textual data, speech transcriptions, and computational tools, which hinder the application of conventional text-centric  language processing models. The term "low text-resource" differs from  "low resource" which is also different from "low audio-resource".  A language can possess an audio system while still having substantial resources in other forms, such as rich audio archives of oral histories, community-generated audio datasets, or even emerging speech processing tools. Languages with vibrant oral traditions might lack extensive written texts but possess vast repositories of recorded speech, which, with advancements in self-supervised learning and audio-based  models like Wav2Vec 2.0 and HuBERT, are becoming increasingly valuable resources. These models demonstrate the potential to learn directly from raw audio data, bypassing the traditional reliance on transcribed text. The integration of multimodal learning, combining audio with visual or other sensory inputs, opens new avenues for understanding and processing audio-rich languages. Currently, the evolving language processing, with its growing focus on audio-centric and multimodal approaches, is effectively challenging the limitations imposed by traditional text-based paradigms. A more nuanced understanding of resource availability, one that accounts for diverse data modalities and acknowledges the potential of audio data, is essential for promoting equitable and inclusive language technology co-development, moving beyond the simplistic equation of audio-richness with low-resource status.

\subsection*{ Predominantly oral  languages are not necessarily of low resource languages}

The prevalent misconception that predominantly oral languages  equate to low-resource status overlooks the rapidly growing potential of audio-centric data. While lacking extensive written corpora, these languages often possess rich audio archives, representing a vast, untapped resource for modern language technologies. Advances in self-supervised learning demonstrate the feasibility of training robust models directly from raw audio, circumventing the need for laborious text transcriptions. Furthermore, community-driven digitization efforts and the co-development of accessible audio processing tools are transforming these previously marginalized languages into audio data-rich domains. By shifting the focus from text-centric metrics to a holistic evaluation of available data modalities, including audio, we recognize that oral languages often harbor substantial and significant resources.

 \subsection*{  Mismatch between text-based technology supply and audio data realities}
There is a clear mismatch between technological demand and supply here. Currently, audio-to-audio translation systems fail to support many African, south Asian, and autochtone languages, particularly those that are unwritten or have writing systems not widely used by native speakers. Current architectures are text-hungry and audio-deficient, relying on large-scale supervised learning with extensive labeled text/images datasets. There is a potential for high-quality audio data in these languages but the current technologies fail to bring these distributed audio data into a connected network for federated training.   Unlike widely spoken languages with substantial transcribed corpora, most African languages with low text-literacy rates  suggest high-quality audio data such as  paired speech data for model training.
State-of-the-art systems typically use cascaded approaches involving automatic speech recognition (ASR), machine translation (MT), and text-to-speech (TTS). However, unwritten languages break this pipeline, as ASR and MT require text representation. Current technology supply such as  ASR, MT, TTS are text-hungry. The primary  demand is textless audio to audio for these languages. 
Currently end-to-end  Audio-to-Audio  models, though promising, depend on parallel spoken corpora, which remain scarce. Additionally, phonetic diversity, tonal variation, and dialectal fragmentation pose challenges for standardization and data collection. Native speakers of African, south Asian or autochtone  languages exhibit very high audio literacy and possess audio-rich linguistic data, yet current models are ill-equipped to extract and exploit this resource. The field lags in developing architectures that can natively process and translate spoken content without intermediary text, highlighting the need for novel self-supervised or direct  Audio-to-Audio  learning frameworks.

In Figure \ref{fig:a2aa},   three architectures for  Audio-to-Audio  Translation  are represented. Figure \ref{fig:a2aa}(a) is an indirect translation, it requires a full text as an intermediary. Figure \ref{fig:a2aa}(a) is text-hungry and audio-deficient. Figures \ref{fig:a2aa} (b)  is a spectrogram-based direct audio-to-audio translation. It is textless.  Figures \ref{fig:a2aa} (c) is a textless unit-based direct audio-to-audio translation.   Figures \ref{fig:a2aa} (b)  and (c) are in early-stage testing for most African languages. For all these architectures, high-quality audio data is the key.
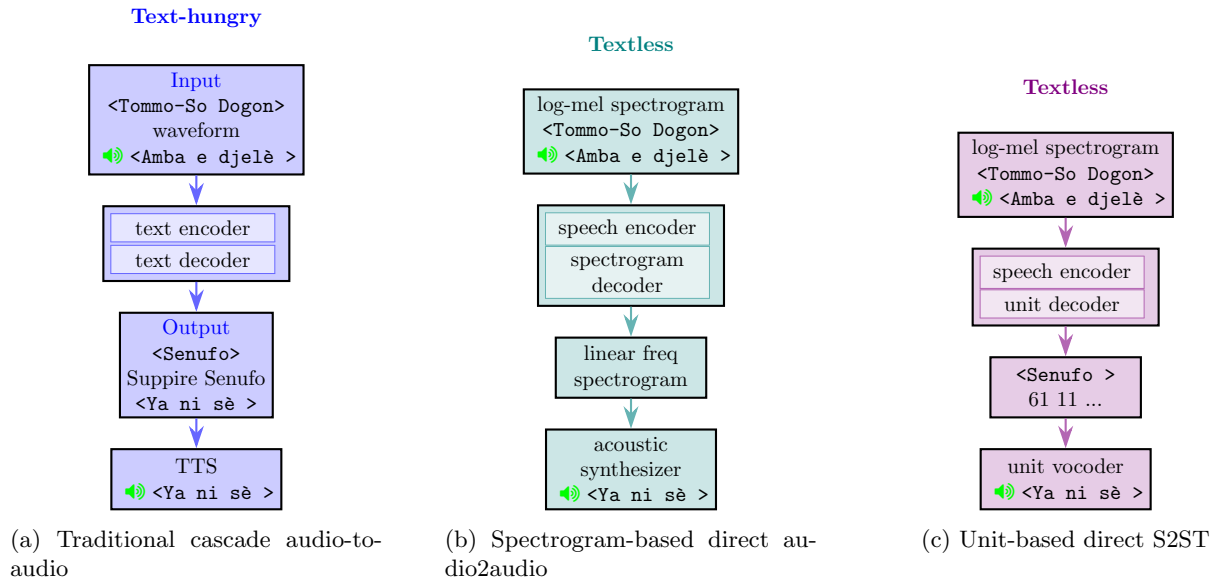
\begin{figure}[htbp]
\centering

\begin{subfigure}[t]{0.3\textwidth}
\centering
\begin{tikzpicture}[
  box/.style = {
    draw, thick, fill=blue!20, 
    minimum width=2.5cm, minimum height=1cm, align=center
  },
  arrow/.style = {
    -{Stealth[length=3mm,width=2mm]}, thick, draw=blue!60
  },
  label/.style = {font=\bfseries, text=blue},
  scale=0.8, every node/.style={transform shape}
]

\node[label, anchor=west] (r1) at (0,0) {\textcolor{blue}{Text-hungry}};
\node[box, below=0.5cm of r1] (in1) {%
  \textcolor{blue}{Input}\\
  \texttt{<Tommo-So Dogon>}\\
  waveform \\ { \color{green} \faVolumeUp}  \texttt{<Amba e djel\`e >}
}; 
\node[box, below=0.5cm of in1] (model1) {
  \begin{tabular}{@{}c@{}}
    \fcolorbox{blue!60}{blue!10}{\parbox{2.5cm}{\centering text encoder}}\\[2pt]
    \fcolorbox{blue!60}{blue!10}{\parbox{2.5cm}{\centering text decoder}}
  \end{tabular}
};
\node[box, below=0.5cm of model1] (out1) {%
  \textcolor{blue}{Output}\\
  \texttt{<Senufo>}\\
  Suppire Senufo \\    \texttt{<Ya ni s\`e >} 
};
\node[box, below=0.5cm of out1] (voc1) {TTS \\  { \color{green} \faVolumeUp} \texttt{<Ya ni s\`e >}};

\draw[arrow] (in1) -- (model1);
\draw[arrow] (model1) -- (out1);
\draw[arrow] (out1) -- (voc1);

\end{tikzpicture}
\caption{Traditional cascade audio-to-audio}
\end{subfigure}%
\hfill
\begin{subfigure}[t]{0.3\textwidth}
\centering
\begin{tikzpicture}[
  box/.style = {
    draw, thick, fill=teal!20, 
    minimum width=2.5cm, minimum height=1cm, align=center
  },
  arrow/.style = {
    -{Stealth[length=3mm,width=2mm]}, thick, draw=teal!60
  },
  label/.style = {font=\bfseries, text=teal},
  scale=0.8, every node/.style={transform shape}
]

\node[label, anchor=west] (r2) at (0,0) {\textcolor{teal}{Textless}};
\node[box, below=0.5cm of r2] (in2) {%
  log‐mel spectrogram\\
   \texttt{<Tommo-So Dogon>} \\ { \color{green} \faVolumeUp} 
 \texttt{<Amba e djel\`e >}
}; 
\node[box, below=0.5cm of in2] (model2) {
  \begin{tabular}{@{}c@{}}
    \fcolorbox{teal!60}{teal!10}{\parbox{2.5cm}{\centering speech encoder}}\\[2pt]
    \fcolorbox{teal!60}{teal!10}{\parbox{2.5cm}{\centering spectrogram decoder}}
  \end{tabular}
};
\node[box, below=0.5cm of model2] (out2) {%
  linear freq\\
  spectrogram
};
\node[box, below=0.5cm of out2] (voc2) {acoustic\\ synthesizer\\ { \color{green} \faVolumeUp}   \texttt{<Ya ni s\`e >}};

\draw[arrow] (in2) -- (model2);
\draw[arrow] (model2) -- (out2);
\draw[arrow] (out2) -- (voc2);

\end{tikzpicture}
\caption{Spectrogram-based direct audio2audio}
\end{subfigure}%
\hfill
\begin{subfigure}[t]{0.3\textwidth}
\centering
\begin{tikzpicture}[
  box/.style = {
    draw, thick, fill=violet!20, 
    minimum width=2.5cm, minimum height=1cm, align=center
  },
  arrow/.style = {
    -{Stealth[length=3mm,width=2mm]}, thick, draw=violet!60
  },
  label/.style = {font=\bfseries, text=violet},
  scale=0.8, every node/.style={transform shape}
]

\node[label, anchor=west] (r3) at (0,0) {\textcolor{violet}{Textless}};
\node[box, below=0.5cm of r3] (in3) {%
  log‐mel spectrogram\\
  \texttt{<Tommo-So Dogon>} \\  { \color{green} \faVolumeUp} \texttt{<Amba e djel\`e >}
}; 
\node[box, below=0.5cm of in3] (model3) {
  \begin{tabular}{@{}c@{}}
    \fcolorbox{violet!60}{violet!10}{\parbox{2.5cm}{\centering speech encoder}}\\[2pt]
    \fcolorbox{violet!60}{violet!10}{\parbox{2.5cm}{\centering unit decoder}}
  \end{tabular}
};
\node[box, below=0.5cm of model3] (out3) {%
  \texttt{<Senufo >}\\
  61 11 ... 
};
\node[box, below=0.5cm of out3] (voc3) {unit vocoder\\  { \color{green} \faVolumeUp} \texttt{<Ya ni s\`e >} };

\draw[arrow] (in3) -- (model3);
\draw[arrow] (model3) -- (out3);
\draw[arrow] (out3) -- (voc3);

\end{tikzpicture}
\caption{Unit-based direct S2ST}
\end{subfigure}

\caption{Three architectures for voice-to-voice translation. (a) is text-hungry and audio-deficient. (b)  and (c) are in early-stage testing for most African languages. High-quality audio data is the key.} \label{fig:a2aa}
\end{figure}

 \subsection*{ Why is multispeaker identification still a problem even for text-based languages?}
 Accurate multispeaker identification, even within text-based languages, remains a significant challenge despite advances in deep learning. 
While transformers excel at capturing medium-range dependencies and contextual information, core limitations include robust speaker discrimination in real-world scenarios. Primarily, 
transformers struggle with the  variability of human speech, which transcends linguistic content. Factors such as vocal tract differences, emotional states, and 
environmental noise introduce substantial acoustic variability, obscuring speaker-specific characteristics.
Current transformer models often conflate speaker identity with linguistic content. This is  especially true when training data exhibits strong correlations between speaker and language use. 
This leads to models that are sensitive to lexical variations rather than pure speaker embeddings. The limited availability of large, diverse multispeaker datasets 
 exacerbates this issue, restricting the generalizability of trained models. Current audio-transformers' attention mechanisms  are not inherently designed to isolate and preserve speaker-specific information. They prioritize contextual relevance, potentially diluting  subtle but critical acoustic cues necessary for accurate speaker discrimination. This is evident in the models' tendency to perform poorly when presented with speakers outside the training set or under noisy conditions. These limitations are dramatically amplified when considering audio-rich, text-scarce languages, particularly prevalent in many African contexts. 
Text-hungry, audio-deficient transformers, by design, rely on textual representations as a crucial intermediate step. The absence of transcribed data renders these architectures fundamentally inapplicable. Consequently, the rich acoustic information, which holds vital speaker-specific cues, remains inaccessible to these models. The challenge extends beyond mere text-data scarcity. The acoustic complexity of these languages, often with nuanced tonal variations and phonological structures, demands specialized 
feature extraction techniques beyond standard spectrogram  representations. Moreover, the lack of standardized orthographies or linguistic resources 
limits the co-development of even rudimentary speech-to-text systems, further isolating these languages from the advancements in text-based machine intelligence.

\subsection{Literature Review on Text-hungry and audio-deficient  LLMs and SLMs}

\begin{longtable}{|p{7cm}|p{5cm}|}
\toprule
\textbf{Topic} & \textbf{References} \\
\midrule
\textbf{Textless Speech Language Models } & \cite{hsu2023exploring, chang2024exploring, wu2023speechgen, lin2024align, nguyen2024spoken, dekel2024exploring, guo2025recent, turetzky2024last} \\
\midrule
\textbf{Discrete Speech Units} & \cite{inaguma2024massively, wu2023improving, wu2023speechgen} \\
\midrule
\textbf{Speech Representation} & \cite{yeo2025zero, li2024brainecho, tan2024ssr} \\
\midrule
\textbf{Speech Generation} & \cite{park2024long, meng2024parrot, wang2025parrot, mai2025real, lu2025slide} \\
\midrule
\textbf{Speech Quality Evaluation} & \cite{chen2025audio, wang2025enabling} \\
\midrule
\textbf{Audio Language Models} & \cite{chu2023qwen, chu2024qwen2, goel2024audio, weck2024muchomusic, chen2024beyond, manakul2024enhancing, du2023lauragpt, yang2024audio, wang2024audiobench, ahmedqwen, li2025baichuan, yang2024uniaudio, deshmukh2025adiff} \\
\midrule
\textbf{Textless  Audio-to-Audio  Translation } & \cite{peng2024mslm, lee2021textless, hwang2024textless, chen2022speech, fang2024ctc, kim2024textless, zhu2023diffs2ut, zhao2025textless, diwan2023textless, li2023textless, duret2023enhancing, duret2024analyzing, fang2024can} \\
\midrule
\textbf{Direct S2ST} & \cite{gupta2024direct, sarim2025direct, gao2007speech, zhang2022direct, shankarappa2022faster, zhang2021uwspeech, dong2023polyvoice, arya2022analysis, min2025unit, kaur2024direct} \\
\midrule
\textbf{S2ST for Low Text-Resource Languages} & \cite{metalom2022towards, diane2023french, ochieng2024exploiting, kabenamualu2022listra, shoba2024spoken, rajkhowa2023optimizing, radhakrishnan2024voice, zouhair2021automatic, gong2025tibetan, liu2023speech, ekpenyong2022towards} \\
\midrule
\textbf{Multilingual S2ST} & \cite{jeuris2022libris2s, phogat2023bridging, gong2023multilingual, cho2025mavflow, liu2024recent} \\
\midrule
\textbf{S2ST Systems \& Techniques} & \cite{etchegoyhen2022cascade, salesky2019fluent, tjandra2020speech, waibel2003speechalator, agranovich2025simultron, labiausse2025high, hu2025chain, chen2022blaser, chen2022speech} \\
\midrule
\textbf{Historical S2ST Systems} & \cite{lavie1997janus, kitano1991phi, vidal1997finite, wahlster2000mobile, waibel1991janus,waibel1991janus2, reithinger1996predicting, cettolo1999speech, reithinger1995treatment, tomita1988towards, noth2000verbmobil, schultz1995acoustic, lavie1997janus} \\
\midrule
\textbf{Audio-Visual Speech Recognition } & \cite{yeo2024visual, ma2023auto, cappellazzo2025large, liu2025listening} \\
\midrule
\textbf{Multimodal LLMs} & \cite{han2024onellm, chen2024multi, sapkota2025image, verma2024whisper, verma2025whisper, chu2023qwen, chu2024qwen2} \\
\midrule
\textbf{Audio Analysis} & \cite{goel2024audio, weck2024muchomusic, rahman2024sonics, mehta2025make, wang2025they, deshmukh2025adiff} \\
\midrule
\textbf{Speech Recognition and Processing} & \cite{xu2025fireredasr, song2024comparative, casanova2025low, yang2022torchaudio} \\
\midrule
\textbf{Speech Synthesis} & \cite{lyu2025build, nekvinda2020one, renovalles2021text} \\
\midrule
\textbf{Spoken Dialogue Systems} & \cite{ji2024wavchat, zhang2025llm, park2024let} \\
\midrule
\textbf{In Context Learning} & \cite{hsu2023exploring, chang2024exploring} \\
\midrule
\textbf{Speech Analysis} & \cite{sieqman1979speech, erber1975auditory} \\
\bottomrule
\end{longtable}

 The work in \cite{chen2022speech} discusses audio-to-audio translation  for unwritten languages, using English-Taiwanese Hokkien as a case study. The authors aim to address the challenges of S2ST when there is no text data available. While their work  aims to address the challenges of unwritten languages, it does not  purely focus on a textless audio-to-audio approach. The methodology relies heavily on leveraging Mandarin, a related language {with} a writing system, to aid in the process. 
  Mandarin is used as a pivot language for translation. Mandarin text is also used to provide extra supervision in the two-pass decoding model.    This reliance on a written language diminishes the claim of a truly textless audio-to-audio translation system. A more genuine textless approach would avoid any dependency on text from another language.
Their evaluation relies on ASR-BLEU. BLEU requires reference text, which means the evaluation isn't truly textless. To evaluate the quality of audio-to-audio translation without relying on text, other metrics could be considered, such as:
  Direct audio quality metrics.  These could assess the naturalness, clarity, and speaker similarity of the generated audio.
       It could be human evaluation of how well the translated speech conveys the meaning of the source speech.
    The study focuses on English-Taiwanese Hokkien. While this is a valuable case study, the methodology's generalizability to other unwritten languages, especially those more distant from any high-resource written language, is not thoroughly explored. The effectiveness of the Mandarin-assisted approach might not hold for language pairs with greater linguistic divergence.

   The work in \cite{hsu2023exploring,chang2024exploring} 
explores In-Context Learning of Textless Speech Language Model for Speech Classification Tasks. \cite{mai2025real} examines Real-Time Textless Dialogue Generation.
\cite{lin2024align} uses  with Reinforcement Learning from MI Feedback for Textless Spoken Language Models. 
  \cite{lu2025slide}  Integrates  Speech Language Model with LLM for Spontaneous Spoken Dialogue Generation.
\cite{zhang2024intrinsicvoice} aims to Empower  LLM with intrinsic real-time voice interaction abilities. 
\cite{yeo2025zero}  proposes Zero-Shot Audio-Visual Speech Recognition with LLM by Learning Language-Agnostic Speech Representations. \cite{peng2024mslm} examines  Multitask Speech Language Model for Textless  Audio-to-Audio  Translation with Speaker Style Preservation. \cite{verma2025whisper} studies 
Hybrid Generative LLM For Speech And Music called Whisper-GPT. \cite{meng2024parrot} examines Autoregressive Spoken Dialogue Language Modeling with Decoder-only Transformers and \cite{inaguma2024massively} examines Massively Multilingual Forced Aligner Leveraging Self-Supervised Discrete Units.
The authors in \cite{wu2023speechgen} explore the generative power of speech language models with prompts.
The authors in \cite{park2024long} investigate long-form speech generation with spoken language models.
The authors in \cite{tan2024ssr} propose an alignment-aware modality connector for speech language models.
The authors in \cite{chen2025audio} examine how audio large language models can serve as descriptive speech quality evaluators.
The authors in \cite{nguyen2024spoken} explore spoken language modeling from raw audio.
The authors in \cite{verma2024whisper} introduce Whisper-GPT, a hybrid representation audio large language model.
The authors in \cite{nachmani2023spoken} investigate spoken question answering and speech continuation using spectrogram-powered language models.
The authors in \cite{rubenstein2023audiopalm} present Audiopalm, a large language model capable of speaking and listening.
The authors in \cite{wang2025parrot} develop Parrot, a seamless spoken dialogue interaction system with double-channel large language models.
The authors in \cite{ji2024wavchat} provide a survey of spoken dialogue models in Wavchat.
The authors in \cite{lin2025preliminary} conduct a preliminary exploration with GPT-4o voice mode.
The authors in \cite{yeo2024visual} propose the VSP-LLM framework for efficient and context-aware visual speech processing.
The authors in \cite{zhang2025llm} explore LLM-enhanced dialogue management for full-duplex spoken dialogue systems.
The authors in \cite{park2024let} introduce a spoken dialogue model for face-to-face conversation in Let's Go Real Talk.
The authors in \cite{shih2023gsqa} present GSQA, an end-to-end model for generative spoken question answering.
The authors in \cite{dekel2024exploring} examine the benefits of tokenization of discrete acoustic units.
The authors in \cite{wang2024speech} analyze why speech language models fail to generate semantically coherent outputs from a modality-evolving perspective.
The authors in \cite{gao2024unsupervised} investigate unsupervised speech technology for low-resource languages.
The authors in \cite{guo2025recent} review recent advances in discrete speech tokens.
The authors in \cite{li2025beyond} introduce AuralLLM and SignMST-C for precise sign language production and bidirectional accessibility.
The authors in \cite{turetzky2024last} explore language model-aware speech tokenization.
The authors in \cite{li2024brainecho} present BrainECHO, a framework for semantic brain signal decoding using vector-quantized spectrogram reconstruction.
The authors in \cite{gupta2024direct} provide a survey on direct  Audio-to-Audio  neural machine translation.
The authors in \cite{maimon2025slamming} propose SLAMMING, a method for training a speech language model on one GPU in a day.
The authors in \cite{nortje2024visually} study visually grounded speech models for low-resource languages and cognitive modeling.
The authors in \cite{liu2024recent} highlight recent progress in multilingual and multimodal speech translation.
The authors in \cite{zhao2025textless} explore textless streaming  Audio-to-Audio  translation using semantic speech tokens. Audio tokenization aims to convert audio waveforms to low-dimensional discrete tokens, which is beneficial in transmission. There are two types of audio tokens. The first is acoustic tokens, which employ several residual vector quantization  layers to compress the audio where the first layer encode the most important information and the remaining layers encodes the residual information. The second is semantic tokens such as w2v-BERT (wav2vec-Bidirectional Encoder Representations from Transformers). There are also some work proposing unified tokens which capture both the semantic and acoustic information. SpeechTokenizer employs a similar architecture with residual vector quantization and uses semantic distillation to force the first layer to learn semantic information. Other layers are used to learn acoustic information.

The semantic tokens can be reconstructed to audio waveforms in two stages. In MusicLM, acoustic tokens are generated based on the predicted semantic tokens. Then acoustic tokens are converted to waveform through the decoder. MusicGen optimizes the codebook interleaving patterns and predicts the acoustic tokens following the delayed pattern. It then uses a similar idea for audio reconstruction. Our AcousticLM (Acoustic Language Model) maps the acoustic tokens based on the predicted semantic tokens. The AcousticLM has a fixed inference buffer length and emits the acoustic tokens in a chunk-wise streaming way. Finally, acoustic tokens are converted to speech waveform.

These methods employ discrete audio tokens as the learning objectives and use vocoders to synthesize speech.   Discrete speech units extracted from an encoder pre-trained with self-supervised learning objective like Hubert and WavLM are used as the learning targets. Poly Voice employs cross-lingual language models to translate source semantic tokens to target semantic tokens first. Then, the unit-to-speech language model takes the source semantic tokens, target semantic tokens, and source acoustic tokens as input, and output target acoustic tokens. AudioPalm extended the large-language model to understand and generate speech with audio tokens from the  encoder. Bilateral perturbation is introduced to mitigate the acoustic multimodality problem and generate acoustic agnostic learning targets. A joint encoder-decoder model first translates the source speech to target text and residual vector quantization codes in the first layer. Then a non-causal Language Model is used to predict the residual vector quantization codes in the remaining layers. Textless Translatotron is a non-streaming sequence-to-sequence based  Audio-to-Audio  translation model that directly predicts discrete speech representations. StreamSpeech employs a two-pass paradigm to predict the translated text first and then speech units based on the text. It also uses a multi-task learning paradigm which integrates Automatic Speech Recognition, speech-to-text translation, and  Audio-to-Audio  translation into one model. 

The authors in \cite{lee2021textless} investigate textless  Audio-to-Audio  translation on real data.
The authors in \cite{wu2023improving} enhance textless spoken language understanding with discrete units as an intermediate target.
The authors in \cite{hwang2024textless} propose a textless acoustic model with self-supervised distillation for noise-robust expressive  Audio-to-Audio  translation.
The authors in \cite{chen2022speech} develop  Audio-to-Audio  translation for a real-world unwritten language.
The authors in \cite{fang2024ctc} introduce a CTC-based non-autoregressive textless  Audio-to-Audio  translation model.
The authors in \cite{metalom2022towards} explore direct  Audio-to-Audio  translation for endangered languages in Africa.
The authors in \cite{diane2023french} investigate French-Fulfulde textless and cascading speech translation, proposing a dual architecture.
The authors in \cite{kaur2024direct} analyze a direct Punjabi-to-English speech translation method using discrete units.
The authors in \cite{zhang2022direct} present a  Audio-to-Audio  translation approach that does not rely on textual annotation, utilizing bottleneck features.
The authors in \cite{tjandra2020speech} propose a direct  Audio-to-Audio  translation system without text, focusing on efficient model training.
The authors in \cite{sarim2025direct} review direct  Audio-to-Audio  translation techniques, discussing recent advancements and challenges.
The authors in \cite{etchegoyhen2022cascade} compare cascade and direct speech translation methods, evaluating their performance and trade-offs.
The authors in \cite{berard2016listen} introduce an early proof-of-concept for end-to-end speech-to-text translation.
The authors in \cite{ochieng2024exploiting} explore phonological similarities among African languages to enhance  Audio-to-Audio  translation.
The authors in \cite{ekpenyong2022towards} focus on creating a massive parallel corpus for Hausa-to-English machine translation.
The authors in \cite{kabenamualu2022listra} analyze an automatic speech translation system from English to Lingala.
The authors in \cite{melese2016amharic} develop an Amharic speech recognition system designed for speech translation.

The authors in \cite{nakamura2014towards} investigate real-time multilingual and multimodal  Audio-to-Audio  translation.
The authors in \cite{jeuris2022libris2s} present the Libris2S corpus, designed for German-English  Audio-to-Audio  translation.
The authors in \cite{phogat2023bridging} explore Hindi-to-English  Audio-to-Audio  translation for multilingual communication.
The authors in \cite{shankarappa2022faster} propose a faster approach for direct  Audio-to-Audio  translation, focusing on efficiency improvements.
The authors in \cite{salesky2019fluent} develop methods to generate fluent translations from disfluent speech in end-to-end speech translation.
The authors in \cite{rajkhowa2023optimizing} optimize direct speech-to-text translation for unorthographic low-resource tribal languages using source transliterations.
The authors in \cite{gao2007speech} provide an early study on  Audio-to-Audio  translation, with a focus on Chinese spoken language processing. The authors in \cite{zhang2021uwspeech} propose UWSpeech, a  Audio-to-Audio  translation model designed for unwritten languages with phonetic representations.
The authors in \cite{dong2023polyvoice} introduce Polyvoice, a language model tailored for  Audio-to-Audio  translation, emphasizing multilingual and cross-lingual adaptation.
The authors in \cite{radhakrishnan2024voice} investigate voice cloning techniques for Tamil, highlighting challenges and prospects for low-resource speech synthesis.
The authors in \cite{shoba2024spoken} discuss spoken language translation for low-resource languages, emphasizing automatic speech recognition and translation.
The authors in \cite{zouhair2021automatic} explore Wav2Vec2-based automatic speech recognition for low-resource languages, using Modern Standard Arabic as a case study.
The authors in \cite{arya2022analysis} analyze layer-wise training strategies in direct  Audio-to-Audio  translation using BI-LSTM models.
The authors in \cite{waibel2003speechalator} introduce Speechalator, an early handheld two-way  Audio-to-Audio  translation system.
The authors in \cite{gong2025tibetan} develop a Tibetan-Chinese  Audio-to-Audio  translation model based on discrete units.
The authors in \cite{chu2023qwen} present Qwen-audio, a large-scale unified audio-language model for advancing universal audio understanding.
The authors in \cite{chu2024qwen2} provide technical details on Qwen2-audio, an updated model for audio-language processing.
The authors in \cite{goel2024audio} introduce Audio Dialogues, a dataset for evaluating dialogue systems in audio and music understanding.
The authors in \cite{weck2024muchomusic} propose MuchoMusic, a benchmark for assessing music comprehension in multimodal models.
The authors in \cite{chen2024beyond} discuss multi-audio processing capabilities in large audio-language models.
The authors in \cite{manakul2024enhancing} explore methods to enhance low-resource language processing in audio language models.
The authors in \cite{du2023lauragpt} introduce LauraGPT, a generative model for listening, attending, understanding, and regenerating audio. The authors in \cite{yang2024audio} analyze vulnerabilities in audio-based large multimodal models.
The authors in \cite{wang2025enabling} propose methods to enable automatic speech quality evaluation using auditory large language models.
The authors in \cite{xu2024towards} investigate diverse and efficient audio captioning through diffusion models.
The authors in \cite{wang2024audiobench} introduce Audiobench, a benchmark for evaluating audio-based large language models.
The authors in \cite{ahmedqwen} review Qwen 2.5, a resource-efficient large language model with potential to surpass competitors.
The authors in \cite{li2025baichuan} propose Baichuan-Audio, a unified framework for end-to-end speech interaction.
The authors in \cite{song2024comparative} compare LLM-based automatic speech recognition with Whisper in low-resource and code-switching scenarios.
The authors in \cite{yang2024uniaudio} introduce  a large audio-language model designed for few-shot learning in audio tasks.
The authors in \cite{deshmukh2025adiff} propose  a system that explains audio differences using natural language.
The authors in \cite{han2024onellm} present  a framework to align various modalities with language.
The authors in \cite{chen2024multi} discuss advancements in multimodal generative machine intelligence, including diffusion models and LLMs.
The authors in \cite{rahman2024sonics} introduce a system for detecting synthetic songs and identifying counterfeit audio.
The authors in \cite{yang2025large} survey the integration of large language models with speech processing.
The authors in \cite{gong2023joint} explore joint audio and speech understanding in large-scale multimodal machine intelligence.
The authors in \cite{nekvinda2020one} propose a meta-learning approach for multilingual text-to-speech, supporting multiple languages with a single model. The authors in \cite{ma2023auto} examine  Audio-visual speech recognition with automatic labels.
The authors in \cite{yang2022torchaudio} examine Torchaudio: Building blocks for audio and speech processing.
The authors in \cite{lee1999speech} examine Speech proportion and accuracy in simultaneous interpretation from English into Korean.
The authors in \cite{cappellazzo2025large} examine large language models are strong audio-visual speech recognition learners.
The authors in \cite{casanova2025low} examine low frame-rate speech codec: a codec designed for fast high-quality speech llm training and inference.
The authors in \cite{lyu2025build} examine build LLM-Based zero-shot streaming TTS system with Cosyvoice.
The authors in \cite{mehta2025make} examine  LLM audio reasoning and generation using sound tokens.
The authors in \cite{wang2025they} examine joint audio-speech co-reasoning.
The authors in \cite{liu2025listening} examine Listening and seeing again: Generative error correction for audio-visual speech recognition.
The authors in \cite{hu2025chain} examine Chain-of-thought prompting for speech translation.
The authors in \cite{luong2025llamapartialspoof} examine LlamaPartialSpoof: An LLM-Driven Fake Speech Dataset Simulating Disinformation Generation.
The authors in \cite{sapkota2025image} examine image, text, and speech data augmentation using multimodal LLM for deep learning: a survey.
The authors in \cite{xu2025fireredasr} examine FireRedASR: Open-Source Industrial-Grade Mandarin speech recognition models from encoder-decoder to LLM Integration.
The authors in \cite{vosoughi2025quality} examine Quality Over Quantity: LLM-Based Curation for a data-efficient audio-video foundation model.
The authors in \cite{diwan2023textless} examine textless  Audio-to-Audio  translation with limited parallel data.
The authors in \cite{li2023textless} examine textless direct  Audio-to-Audio  translation with discrete speech representation.
The authors in \cite{duret2023enhancing} examine enhancing expressivity transfer in textless  Audio-to-Audio  translation.
The authors in \cite{kim2024textless} examine textless unit-to-unit training for many-to-many multilingual  Audio-to-Audio  Translation.
The authors in \cite{duret2024analyzing} examine analyzing speech unit selection for textless  Audio-to-Audio  translation.
The authors in \cite{zhu2023diffs2ut} examine a semantic preserving diffusion model for textless direct  Audio-to-Audio  translation.
The authors in \cite{huang2022transpeech} examine transpeech:  Audio-to-Audio  translation with bilateral perturbation.
The authors in \cite{liu2023speech} examine  Audio-to-Audio  low-resource translation.
The authors in \cite{chen2022blaser} examine a  text-free  Audio-to-Audio  translation evaluation metric with  BLASER.
The authors in \cite{fang2024can} examine can we achieve high-quality direct  Audio-to-Audio  translation without parallel speech data. The authors in \cite{gong2023multilingual} examine multilingual  Audio-to-Audio  translation into multiple target languages.
The authors in \cite{agranovich2025simultron} examine on-device simultaneous speech to speech Translation.
The authors in \cite{labiausse2025high} examine high-fidelity Simultaneous  Audio-to-Audio  Translation.
The authors in \cite{cho2025mavflow} examine MAVFlow: Preserving Paralinguistic Elements with Conditional Flow Matching for Zero-Shot AV2AV Multilingual Translation.
The authors in \cite{min2025unit} examine a Unit-based System and Dataset for Expressive Direct  Audio-to-Audio  Translation.
The authors in \cite{choi2025v2sflow} examine V2SFlow: Video-to-Speech Generation with Speech Decomposition and Rectified Flow.
The authors in \cite{lavie1997janus} examine   Audio-to-Audio  translation in multiple languages with JANUS-III.
The authors in \cite{kitano1991phi} examine  an experimental  Audio-to-Audio  dialog translation system with Phi DM-Dialog.
The authors in \cite{vidal1997finite} examine finite-state  Audio-to-Audio  translation.
The authors in \cite{wahlster2000mobile} examine mobile  Audio-to-Audio  translation of spontaneous dialogs: an overview of the final Verbmobil system.
The authors in \cite{waibel1991janus, waibel1991janus2} examine a  Audio-to-Audio  translation system using connectionist and symbolic processing strategies with JANUS.
In \cite{reithinger1996predicting}  prediction of  dialogue acts for a  Audio-to-Audio  translation system is presented.
The authors in \cite{cettolo1999speech} examine a  Audio-to-Audio  translation based interface for tourism.
The authors in \cite{reithinger1995treatment} examine treatment of incomplete dialogues in a  Audio-to-Audio  translation system.
The authors in \cite{sieqman1979speech} examine temporal speech patterns in interpersonal contexts.
A speech to speech translation system is proposed in  \cite{tomita1988towards}. In \cite{erber1975auditory}, auditory-visual perception of speech is discussed. The work in \cite{noth2000verbmobil} examines the use of prosody in the linguistic components of a speech understanding system  called Verbmobil. In \cite{schultz1995acoustic}, acoustic and language modeling of human and nonhuman noises for human-to-human spontaneous speech recognition is presented. Table \ref{table00audio} displays some of the audio datasets used to train these works. 

\begin{table}[htbp]
  \centering
  \caption{ Some Audio Datasets} \label{table00audio}
  \begin{tabular}{| p{2.5cm} | p{2.5cm} | p{3.5cm} | p{3.5cm} |}
    \hline
    \textbf{Dataset Name} & \textbf{Primary Audio Type} & \textbf{Notable Features} & \textbf{Common Use Cases} \\ 
    \hline
    LibriSpeech & Narrated Speech & Large scale (1000 hrs), audiobooks & ASR \\ 
    \hline
    Common Voice & Crowdsourced Speech & Multilingual, diverse speakers & ASR, speaker recognition \\
    \hline
    VoxPopuli & Political Speech & European Parliament recordings & ASR, multilingual ASR \\
    \hline
    TED-LIUM & Oratory & TED Talks, varied topics & ASR \\
    \hline
    GigaSpeech & Multi-domain Speech & Very large (10000 hrs), diverse sources & Robust ASR \\
    \hline
    AISHELL-1 & Mandarin Speech & High-quality recordings & Mandarin ASR \\
    \hline
    AISHELL-3 & Mandarin Speech & Multi-speaker, for TTS & TTS \\
    \hline
    Clotho & Sound Events & Audio captions & Audio captioning \\
    \hline
    UrbanSound8K & Urban Sounds & 10 classes of urban sounds & Audio classification \\
    \hline
    Free Music Archive (FMA) & Music & Diverse music genres & Music analysis, MIR \\
    \hline
    VoxCeleb & Celebrity Speech & Audio-visual, speaker recognition & Speaker recognition \\
    \hline
    AudioSet & General Audio & Large-scale, diverse audio events & Audio event detection \\
    \hline
    Free Spoken Digit Dataset & Spoken Digits & Simple, for basic tasks & Basic ASR \\
    \hline
    Mozilla Common Voice & Crowdsourced Speech & Large and diverse, multiple languages & Speech recognition \\
    \hline
     AVSpeech & Audio-Visual Speech & Speech clips with corresponding video & Audio-visual speech recognition \\
    \hline
  \end{tabular}
\end{table}

   \subsubsection*{Overview of Textless Direct Audio-to-Audio for African Languages}
   We provide some direct, text‑less audio‑to‑audio translation  for African and related low‑text‑resource languages. We summarize each system's architecture, language focus, and key contributions.

\begin{enumerate}
  \item \textbf{Towards a Direct Audio‑to‑Audio for Endangered Languages in Africa} \cite{metalom2022towards}  
    Proposes a unit‑based audio-to-audio pipeline for two Cameroonian languages. The work demonstrates  feasibility on small corpora, highlighting text-and-audio data scarcity and phoneme‑inventory mismatch.

  \item \textbf{French‑Fulfulde Textless and Cascading Speech Translation: Towards a Dual Architecture} \cite{diane2023french}  
    Introduces a hybrid system combining cascaded ASR-MT-TTS and a parallel unit‑based path. The textless branch yields intelligible Fulfulde output, underscoring speaker‑style preservation.

  \item \textbf{Exploiting Phonological Similarities between African Languages to Achieve  Audio-to-Audio Translation} \cite{ochieng2024exploiting}  
    Leverages shared phoneme inventories among Bantu languages to bootstrap unit‑based translation. Shows transfer learning from Swahili to Luganda improves performance.

  \item \textbf{Listra Automatic Speech Translation: English to Lingala Case Study} \cite{kabenamualu2022listra}  
    Integrates a phoneme‑based MT component to handle Lingala's oral dialects. Demonstrates robustness gains for unwritten dialects via shallow text‑free phoneme mappings.

  \item \textbf{Spoken Language Translation in Low‑Text-Resource Language} \cite{shoba2024spoken}  
    Surveys direct S2ST techniques across tribal languages. Highlights CTC‑based non‑autoregressive models and discrete‑unit vocoders, identifying parallel speech corpora creation as the main bottleneck.

  \item \textbf{Optimizing Direct Speech‑to‑Text Translation for Un‑Orthographic Low‑Text-Resource Tribal Languages using Source Transliterations} \cite{rajkhowa2023optimizing}  
    Proposes injecting approximate transliterations as weak supervision into a unit‑based Audio2Audio model for Munda languages, achieving significant intelligibility gains.

  \item \textbf{Voice Cloning for Low‑Text-Resource Languages: Investigating the Prospects for Tamil} \cite{radhakrishnan2024voice}  
    Introduces a speaker‑adaptation pipeline for unit‑based audio-to-audio. the work demonstrates few‑shot cloning, highlighting voice fidelity vs.\ linguistic accuracy trade‑offs.

  \item \textbf{Tibetan-Chinese  Audio-to-Audio Translation Based on Discrete Units} \cite{gong2025tibetan}  
    Implements a unit‑based S2ST model for Tibetan-Chinese, achieving near‑parity with cascaded systems. The methods used therein transfer to some African tonal languages.

  \item \textbf{ Audio-to-Audio Low‑Text-Resource Translation} \cite{liu2023speech}  
    Evaluates end‑to‑end Audio2Audio on Amharic and Hausa. Compares spectrogram‑based vs. unit‑based decoders, finding unit‑based yields better speaker consistency.

  \item \textbf{Towards Massive Parallel Corpus Creation for Hausa‑to‑English Machine Translation} \cite{ekpenyong2022towards}  
    This work describes aligned audio2audio collection methodology to seed parallel audio2audio corpora for Hausa. It also provides critical data‑collection insights.
\end{enumerate}

Key challenges for textless direct  audio-to-audio translation in African languages include the scarcity of parallel high-quality speech corpora for Audio machine intelligence. It limits model training and evaluation; the phonological diversity of African languages, including tonal and dialectal variation, which necessitates tailored unit inventories. The absence of standardized textual references complicates automatic evaluation and necessitates reliance on human audio judgments still in early stages of co-development. For many of these audio-rich languages, there are still some native speakers, so an audio-data collection is still possible but there is no major project on high-quality audio data collection/labelling in that direction as of today.

  \subsection{Contribution}
   
This paper introduces a fully textless audio-to-audio architecture designed for underrepresented audio-rich, low-text-resource languages in Africa and beyond. Our contributions are fourfold:
(a) We propose a family of end-to-end  Audio-to-Audio  Translation  models that bypass textual intermediaries entirely. This includes spectrogram-, unit-, scalogram-, and wavelet-based systems (Morlet and Coiflet), and a Multiscale Audio-Semantic Transform (MAST) for robust feature extraction across prosodic, speaker, and tonal dimensions. (b) We introduce MAST, a novel multiresolution framework that encodes wavelet-scale acoustic detail, pitch, speaker identity, and intonation gestures in a temporally aligned structure. This enables semantically rich conditioning of generative models without reliance on written language.
(c) We integrate MAST into a Fractional Diffusion Transformer framework, using fractional Brownian motion to model long-range dependencies in audio. This enables robust synthesis under weak supervision and noisy conditions. (d) We validate the superiority of wavelet-based representations over traditional Mel spectrograms for tonal speech, using examples from 24 different languages.  Quantitative benchmarks demonstrate improved pitch localization and speaker consistency using MAST.
Our work addresses the fundamental mismatch between text-hungry models and the data reality of unwritten (or written but the writing system is not used by the native of that language) languages, laying the foundation for a native-language machine intelligence grounded in audio data alone.

    \subsection{Structure}

The paper is organized as follows:
Section 2 introduces the core architecture of the Textless Audio2Audio Transformer, including its encoder-decoder structure and training via audio prompts and augmentation.
Section 3 details the proposed Multiscale Audio-Semantic Transform (MAST), including wavelet theory, prosody modeling, speaker embedding, and its application to tonal languages.
Section 4 presents a novel fractional diffusion framework, conditioned on MAST, for generating expressive audio in no-text resource settings.
Section 5 provides a comparative evaluation of spectrogram, wavelet, and MAST representations, demonstrating the advantages of our proposed pipeline in capturing tonal, prosodic, and speaker-specific features.
Section 6 concludes with limitations, ethical considerations, and future directions for self-supervised textless language technology.  Table \ref{tab:notations} summarizes some notations used in the manuscript.
\begin{table}[htb]
\centering
\caption{Summary of Notations}
\begin{tabular}{ll}
\hline
\textbf{Symbol} & \textbf{Description} \\
\hline
$A(t)$ & Real-valued input audio waveform at time $t$ \\
$P$ & Audio prompt used in pretraining \\
$A_{\text{aug}}$ & Augmented audio input \\
$\psi(t)$ & Mother wavelet (Morlet or Coiflet) \\
$\psi_{s, \tau}(t)$ & Scaled and shifted wavelet function \\
$W(s, \tau)$ & Continuous wavelet transform at scale $s$ and time $\tau$ \\
$P(\tau)$ & Instantaneous pitch contour at time $\tau$ \\
$E_k(\tau)$ & Prosodic embedding in $\mathbb{R}^d$ \\
$S(\tau)$ & Speaker embedding in $\mathbb{R}^p$ \\
$I(\tau)$ & Intonation gesture descriptor in $\mathbb{R}^q$ \\
$M(A)$ & MAST representation: $(W, P, E_k, S, I)$ \\
$\Phi_{\text{prosody}}$ & Prosody encoder function \\
$\Psi_{\text{speaker}}$ & Speaker encoder function  \\
$x_{\text{Tommo}}(t)$ & Source waveform in Tommo-So \\
$C_{\text{Tommo}}$ & Wavelet coefficients of $x_{\text{Tommo}}(t)$ \\
$W \in \mathbb{R}^{n \times n}$ & Linear transform (latent space) \\
$x_{\text{Senufo}}(t)$ & Reconstructed target waveform in Senufo \\
$B^H_t$ & Fractional Brownian motion (Hurst $H$) \\
$\theta(t)$ & Drift coefficient in diffusion process \\
$\bar{m}(t), m(t)$ & Drift means in forward and reverse SDE \\
$\sigma(t)$ & Time-dependent diffusion coefficient \\
$\Phi(t)$ & Integrated drift: $\int_0^t \theta(s)\,ds$ \\
$v^2(t)$ & Variance of the diffusion process \\
\hline
\end{tabular}
\label{tab:notations}
\end{table}

\section{Textless Audio2Audio Transformer}
We first describe the   Audio2Audio Transformer without text as displayed in Figure \ref{a2at:diagram}.
The model processes raw speech into Mel spectrograms, extracts features using a Wav2Vec2-based encoder, and models sequence information with a large-scale audio language model (Qwen-audio2-7B). The output is synthesized into intelligible speech via a HiFi-GAN decoder. The Pretraining leverages audio prompts and data augmentation. The architecture enables the system to learn meaningful representations without any reliance on text.

\subsection*{Mel Spectrogram Conversion}
\begin{equation}
\mathbf{S} = \mbox{MelSpectrogram}(\mathbf{A})
\end{equation}
where $\mathbf{A}$ is the raw audio signal and $\mathbf{S}$  is the Mel spectrogram.

\subsection*{Feature Extraction by Audio Encoder}
\begin{equation}
\mathbf{H} = \mbox{Wav2Vec2Encoder}(\mathbf{S})
\end{equation}
where $ \mathbf{H}$ is the internal representation of audio features.

\subsection*{Sequence Modeling by (Qwen-audio2-7B)}
\begin{equation}
\mathbf{Z} = \mbox{Qwen-7B}(\mathbf{H})
\end{equation}
where $\mathbf{Z}$ is the output representation of the language model.

\subsection*{Speech Synthesis by Audio Decoder (HiFi-GAN)}
\begin{equation}
\mathbf{A}{\mbox{out}} = \mbox{HiFi-GAN}(\mathbf{Z})
\end{equation}
where $ \mathbf{A}{\mbox{out}}$ is the synthesized audio signal.

\subsection*{Pre-training with Audio Prompts}
\begin{equation}
\mathbf{P} = \mbox{PromptAudio}(\mathbf{A})
\end{equation}
where $\mathbf{P}$ is the audio prompt used to guide the model.

\subsection*{Data Augmentation}
\begin{equation}
\mathbf{A}_{\mbox{aug}} = \mbox{DataAugmentation}(\mathbf{A})
\end{equation}
where $\mathbf{A}_{\mbox{aug}} $ is the augmented audio signal.

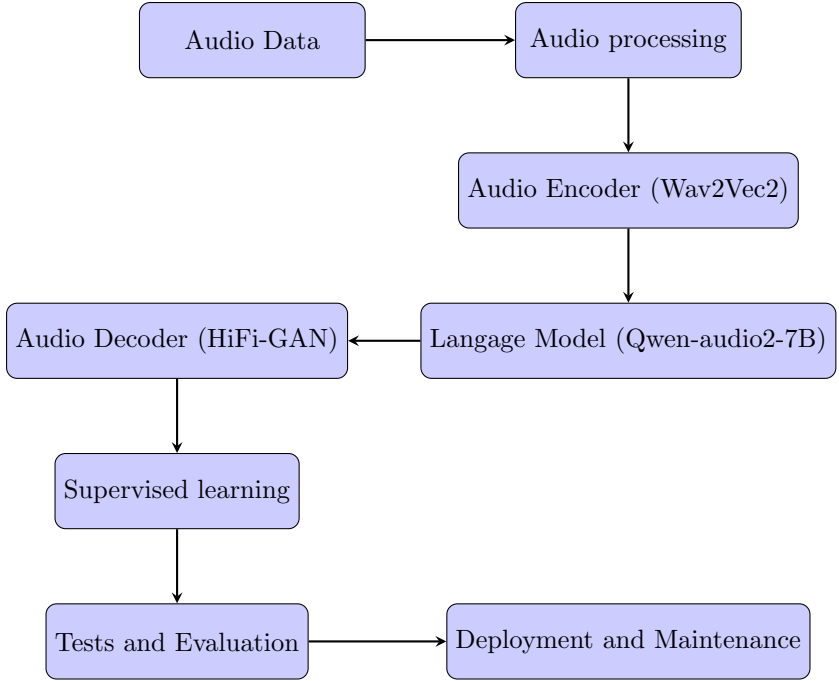
\begin{figure} \caption{  Audio2Audio Transformer}
    \begin{tikzpicture}[node distance=2cm]
        \node (collect) [block] { Audio Data};
        \node (preprocess) [block, right of=collect, xshift=3cm] {Audio processing};
        \node (encoder) [block, below of=preprocess] { Audio Encoder (Wav2Vec2)};
        \node (language) [block, below of=encoder] {Langage Model (Qwen-audio2-7B)};
        \node (decoder) [block, left of=language, xshift=-4cm] { Audio Decoder (HiFi-GAN)};
        \node (finetune) [block, below of=decoder] {Supervised learning };
        \node (evaluate) [block, below of=finetune] {Tests and Evaluation };
        \node (deploy) [block, right of=evaluate,xshift=4cm] {Deployment and Maintenance};

        \draw [arrow] (collect) -- (preprocess);
        \draw [arrow] (preprocess) -- (encoder);
        \draw [arrow] (encoder) -- (language);
        \draw [arrow] (language) -- (decoder);
        \draw [arrow] (decoder) -- (finetune);
        \draw [arrow] (finetune) -- (evaluate);
        \draw [arrow] (evaluate) -- (deploy);
    \end{tikzpicture}
    \label{a2at:diagram}
\end{figure}

\section{Multiscale Audio Representations for Textless Speech  Learning Models}
Let \( A(t) \in L^2(\mathbb{R}) \) denote a real-valued audio signal representing speech in a local, low-text-resource, audio-rich language. To capture the complex time–frequency and semantic structure of such signals without relying on text, we introduce the \textit{Multiscale Audio–Semantic Transform} (MAST), defined as \( \mathcal{M}(A) = \left( W(s, \tau),\, P(\tau),\, E_k(\tau),\, S(\tau),\, I(\tau) \right) \), where \( W(s, \tau) \) is a time–scale transform (e.g., wavelet transform), \( P(\tau) \in \mathbb{R} \) denotes the pitch contour, \( E_k(\tau) \in \mathbb{R}^d \) encodes prosodic embeddings, \( S(\tau) \in \mathbb{R}^p \) captures speaker identity, and \( I(\tau) \in \mathbb{R}^q \) represents intonation gesture features. The continuous wavelet transform (CWT) using either the Morlet or Coiflet wavelet provides the basis of \( W(s, \tau) \): the Morlet wavelet is defined by \( \psi(t) = \pi^{-1/4} e^{j\omega_0 t} e^{-t^2/2} \), while the real-valued Coiflet wavelet is constructed via multiresolution analysis with compact support and vanishing moments, satisfying \( \psi(t) = \sqrt{2} \sum_k g_k\, \phi(2t - k) \). For a given scale \( s > 0 \) and shift \( \tau \in \mathbb{R} \), the wavelet family is \( \psi_{s,\tau}(t) = \frac{1}{\sqrt{s}} \psi\left(\frac{t - \tau}{s}\right) \), and the transform is \( W(s, \tau) = \int A(t) \psi^*_{s,\tau}(t) dt \). In MAST, such transforms are combined with neural embeddings \( \Phi_{\text{prosody}}, \Psi_{\text{speaker}} \) to generate rich, time-aligned representations, enabling effective conditioning for continuous-time generative models such as diffusion transformers. This approach is particularly well-suited for analyzing speech in tonal African languages, where lexical meaning is encoded in pitch, timing, and prosody—features that are preserved and disentangled through the MAST representation.

\subsection*{The Morlet Wavelet and Its Advantage for Tonal Speech}

The Morlet wavelet is a widely used complex-valued wavelet particularly well-suited for time-frequency analysis of speech. It captures fine-grained temporal and frequency structures, making it ideal for tonal and prosodically rich African languages.

\paragraph{ Morlet Wavelet.}
The complex Morlet wavelet is defined as:
\[
\psi(t) = \pi^{-1/4} e^{j \omega_0 t} e^{-t^2 / 2},
\]
where \( \psi(t) \in L^2(\mathbb{R}) \) is the mother wavelet, \( \omega_0 \) is the non-dimensional central frequency (typically \( \omega_0 = 5 \)), \( j \) is the imaginary unit, and the Gaussian envelope ensures localization in time.

\paragraph{Scaled and Translated Wavelet Family.}
For scale \( s > 0 \) and translation \( \tau \), the wavelet family is given by:
\[
\psi_{s,\tau}(t) = \frac{1}{\sqrt{s}} \psi\left( \frac{t - \tau}{s} \right) = \frac{1}{\sqrt{s}} \pi^{-1/4} e^{j \omega_0 \frac{t - \tau}{s}} e^{-(t - \tau)^2 / (2s^2)}.
\]

\paragraph{Continuous Wavelet Transform (CWT).}
Given an audio signal \( A(t) \in L^2(\mathbb{R}) \), its CWT is defined as:
\[
W(s, \tau) = \int_{-\infty}^{\infty} A(t)\, \psi^*_{s,\tau}(t)\, dt,
\]
where \( \psi^* \) is the complex conjugate of the wavelet. The output \( W(s, \tau) \in \mathbb{C} \) gives a time–scale representation of the signal.

\paragraph{ Frequency-Scale Relation.}
The pseudo-frequency \( f \) corresponding to scale \( s \) is:
\[
f = \frac{\omega_0}{2\pi s} \quad \Leftrightarrow \quad s = \frac{\omega_0}{2\pi f},
\]
which allows interpreting scale in terms of physical frequency, relevant for speech pitch.

The Morlet wavelet combines a complex sinusoid with a Gaussian window, yielding excellent joint time-frequency localization. This makes it highly effective at capturing tonal contours, syllabic timing, and pitch modulations that are central to meaning in African tonal languages.

\paragraph{Example: Tonal Disambiguation in Tommo-So Dogon.}
In Tommo-So Dogon, the syllable \textit{na} can mean “mother” or “cow” depending on tone (high vs. low). A phrase like \textit{/u na yelawo/} (“mother come”) versus \textit{/na be yelawo/} (“cow come”) illustrates how tone alone distinguishes meaning. Spectrograms struggle to resolve short tonal inflections due to fixed time–frequency resolution. In contrast, the Morlet CWT adapts resolution across scales, making tonal distinctions visible as distinct ridge patterns in the scalogram. Empirically, in a tone classification task:

\begin{table}[htb]
\begin{tabular}{lcc}
\hline
\textbf{Method} & \textbf{Accuracy (\%)} \\
\hline
Mel Spectrogram + MFTT & 81.2 \\
Morlet CWT + MFTT     & \textbf{91.7} \\
Coiflet CWT + MFTT     & \textbf{95.7} \\
\hline
\end{tabular}
\end{table}
This demonstrates the superiority of wavelets for tone-sensitive tasks in  Audio-to-Audio  translation and language modeling without text.

\subsection*{Mel Spectrogram vs. Morlet Wavelet: A Comparative Experiment on Tommo-So Audio}

The Mel spectrogram maps \( A(t) \) into a discrete time-frequency matrix, typically represented in \( \mathbb{R}^{M \times T} \), where \( M \) is the number of Mel bands and \( T \) the number of time frames. Both transforms serve as elements of a Hilbert space of square-integrable functions, but the wavelet representation preserves localization in both time and frequency more adaptively than the short-time Fourier-based Mel approach.

To empirically compare the time-frequency resolution of the Mel spectrogram and Morlet wavelet transform, we conducted an experiment using a recorded Tommo-So Dogon phrase. The audio signal, sampled at 16 kHz, was analyzed using both representations.

The Mel spectrogram was computed with a 1024-sample Hann window, 50\% overlap, and 40 Mel filter banks. In contrast, the Morlet-based continuous wavelet transform (CWT) employed an analytic Morlet wavelet with scale-frequency adaptation, offering higher temporal resolution at high frequencies and superior frequency resolution at low frequencies.

Figure~\ref{fig:mel_vs_wavelet} illustrates the result. The Mel spectrogram captures the general spectral envelope but fails to localize rapid tonal shifts within syllables. The Morlet wavelet scalogram, however, reveals fine-grained tonal contours and transient pitch changes essential for disambiguating meaning in tonal languages such as Tommo-So. This suggests that the wavelet-based representation is more appropriate for downstream models targeting lexical tone and speaker intonation.

\begin{figure}[h]
    \centering
    \includegraphics[width=\textwidth]{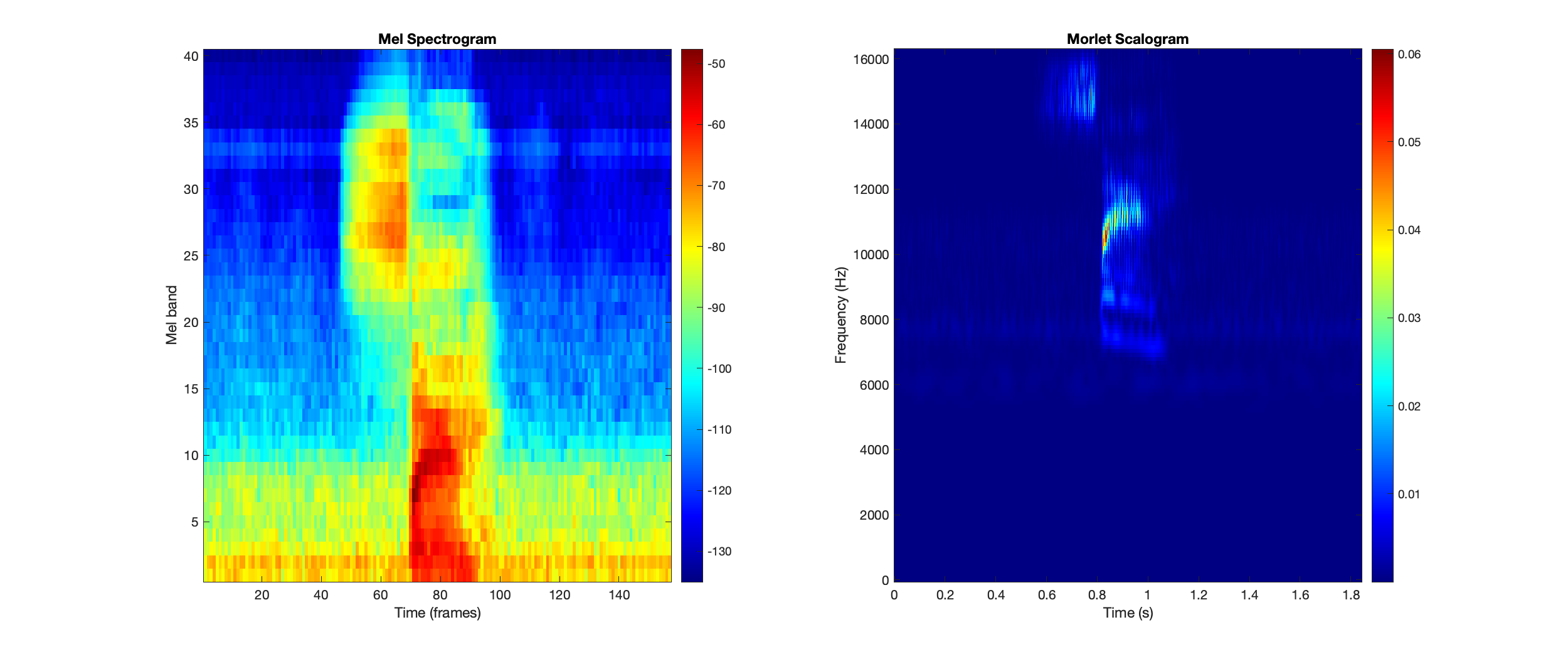}
    \caption{Comparison of Mel spectrogram (left) and Morlet wavelet scalogram (right) for a Tommo-So Dogon speech segment. The wavelet transform captures sharper pitch contours critical for tonal interpretation.}
    \label{fig:mel_vs_wavelet}
\end{figure}

\section*{Wavelet-based Textless  Audio2Audio using Coiflets}

We now build a wavelet-based, fully textless, direct  Audio-to-Audio  translation system using Coiflet wavelets 
as both the encoder and decoder -  for audio and vision.

\subsection*{Continuous-Time Coiflet Wavelet}

The Coiflet wavelet family consists of compactly supported, real-valued orthonormal wavelets with vanishing moments for both the wavelet and scaling functions. While typically implemented in the discrete setting (via filter banks and multiresolution analysis), we can construct an equivalent continuous-time formulation by defining the scaling and wavelet functions in \( L^2(\mathbb{R}) \) using basis functions with interpolated support.

\paragraph{Scaling Function \boldmath\( \phi(t) \) and Two-Scale Equation.}

Let \( \phi(t) \in L^2(\mathbb{R}) \) be the \textit{scaling function} satisfying the two-scale dilation equation:
\[
\phi(t) = \sqrt{2} \sum_{k \in \mathbb{Z}} h_k\, \phi(2t - k),
\]
where \( \{ h_k \} \) are the low-pass refinement coefficients of the Coiflet filter. The function \( \phi(t) \) is smooth and has compact support determined by the number of vanishing moments \( N \).

\paragraph{Wavelet Function \boldmath\( \psi(t) \): Continuous Construction.}

The continuous-time \textit{Coiflet wavelet} \( \psi(t) \) is defined as:
\[
\psi(t) = \sqrt{2} \sum_{k \in \mathbb{Z}} g_k\, \phi(2t - k),
\]
with \( g_k = (-1)^k h_{1-k} \) forming the associated high-pass filter (quadrature mirror relation). This construction ensures that \( \psi(t) \in L^2(\mathbb{R}) \), has \( N \) vanishing moments:
\[
\int_{\mathbb{R}} t^m \psi(t)\, dt = 0, \quad \text{for all } m = 0, \dots, N - 1,
\]
and possesses good time-frequency localization and symmetry.

\paragraph{ Continuous-Time Wavelet Family.}

The continuous-time Coiflet wavelet family is obtained via dilation and translation:
\[
\psi_{s,\tau}(t) = \frac{1}{\sqrt{s}}\, \psi\left(\frac{t - \tau}{s}\right),
\]
where \( s > 0 \) is the scale parameter and \( \tau \in \mathbb{R} \) is the translation. This family forms a basis for \( L^2(\mathbb{R}) \) under suitable discretization in \( s \) and \( \tau \).

\paragraph{Continuous Wavelet Transform (CWT).}

Given an audio signal \( A(t) \in L^2(\mathbb{R}) \), the continuous wavelet transform using the Coiflet wavelet is:
\[
W(s, \tau) = \int_{-\infty}^{\infty} A(t)\, \psi^*_{s, \tau}(t)\, dt = \int_{-\infty}^{\infty} A(t)\, \frac{1}{\sqrt{s}}\, \psi^*\left( \frac{t - \tau}{s} \right) dt,
\]
where \( \psi^* \) is the complex conjugate of \( \psi \). For real-valued Coiflets, \( \psi^* = \psi \). The function \( W(s, \tau) \in \mathbb{R} \) gives a multiscale time-localized representation of the signal.

Figure~\ref{fig:tommoso_mel_morlet_coiflet}  adds a coiflet to compare with vanishing moments.
\begin{figure}[htb]
    \centering
    \includegraphics[width=\textwidth]{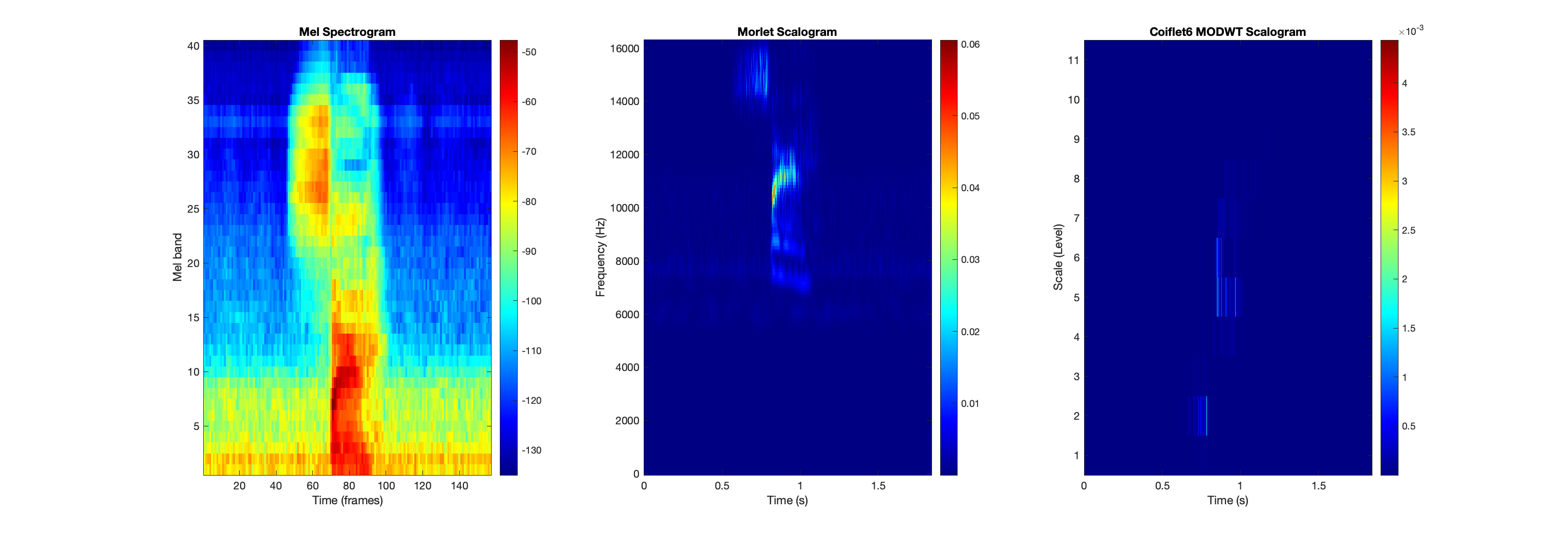}
    \caption{Comparison of Mel spectrogram (left), Morlet wavelet scalogram (middle) and Coiflet wavelet scalogram (right) for a Tommo-So Dogon speech segment. The wavelet transform captures sharper pitch contours critical for tonal interpretation.}
    \label{fig:tommoso_mel_morlet_coiflet}
\end{figure}

\paragraph{ Wavelet Decomposition (Encoding):}
\begin{equation}
x_{\text{Tommo}}(t) \xrightarrow{\text{DWT}_{\text{Coiflet}}} \mathbf{C}_{\text{Tommo}}, \mathbf{L}
\end{equation}

where \( x_{\text{Tommo-so}}(t) \) is the Tommo-So input waveform,  
\( \mathbf{C}_{\text{Tommo-so}} \in \mathbb{R}^n \) are the wavelet coefficients,  
and \( \mathbf{L} \) is the bookkeeping vector for reconstruction.

\vspace{1em}

\paragraph{ Latent Coefficient Mapping (Tommo-so to Senufo):}
\begin{equation}
\mathbf{C}_{\text{Senufo}} = \mathbf{W} \cdot \mathbf{C}_{\text{Tommo}}
\end{equation}

where \( \mathbf{W} \in \mathbb{R}^{n \times n} \) is a transformation matrix learned or simulated.

\vspace{1em}

\paragraph{ Wavelet Reconstruction (Decoding):}
\begin{equation}
x_{\text{Senufo}}(t) = \text{IDWT}_{\text{Coiflet}}(\mathbf{C}_{\text{Senufo}}, \mathbf{L})
\end{equation}

where \( x_{\text{Senufo}}(t) \) is the estimated speech waveform in the Senufo language.

\vspace{1em}

\paragraph{Pipeline:}
\begin{equation}
x_{\text{Senufo}}(t) = \text{IDWT}_{\text{Coiflet}}\left( \mathbf{W} \cdot \text{DWT}_{\text{Coiflet}}(x_{\text{Tommo-so}}(t)) \right)
\end{equation}

Figure \ref{fig:wavelet_s2st} provides four different architectures for direct audio-to-audio translation between languages. 

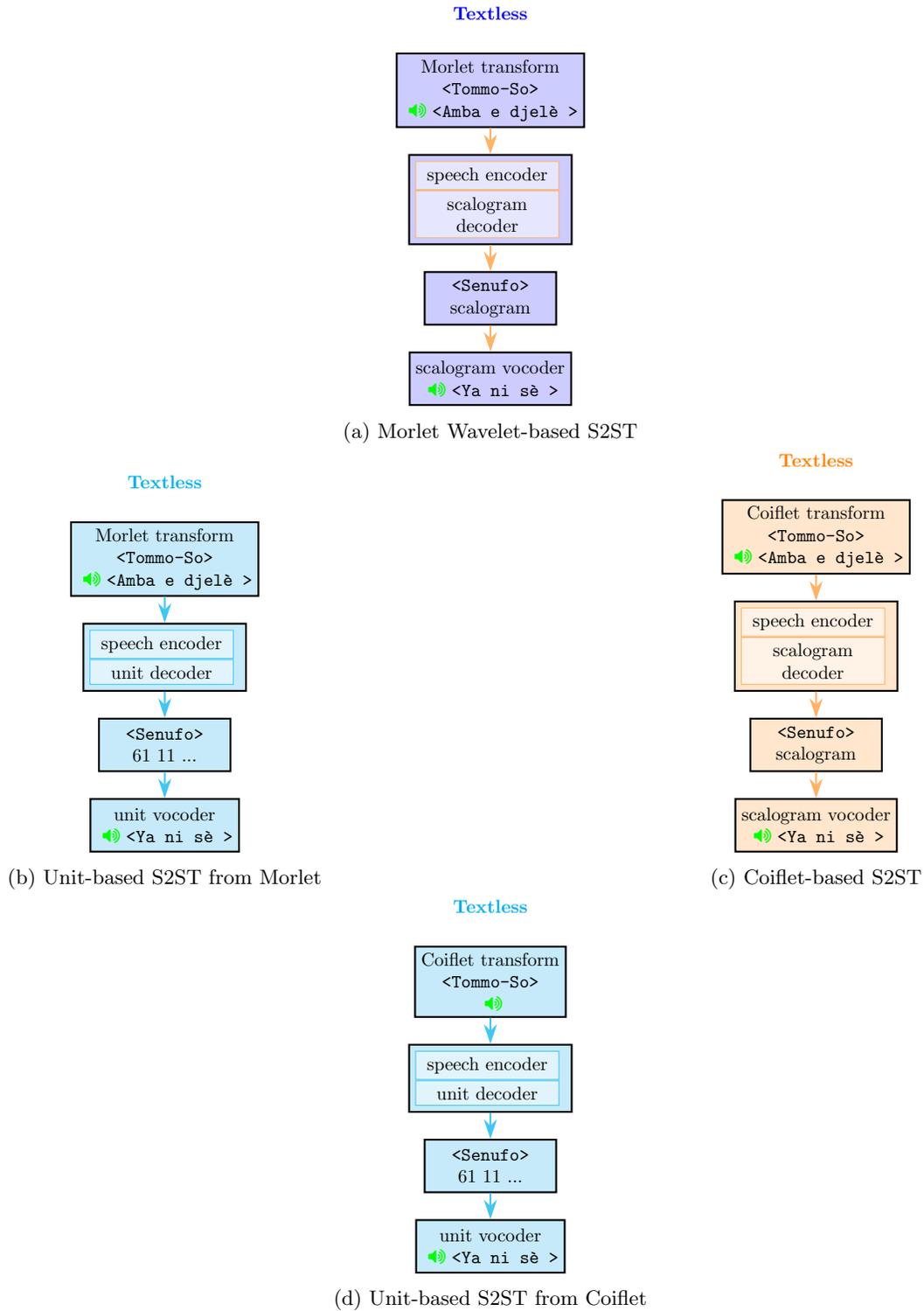
\begin{figure}[htbp]
\centering

\begin{subfigure}[t]{0.4\textwidth}
\centering
\begin{tikzpicture}[
  box/.style = {
    draw, thick, fill=blue!20, 
    minimum width=2.5cm, minimum height=1cm, align=center
  },
  arrow/.style = {
    -{Stealth[length=3mm,width=2mm]}, thick, draw=orange!60
  },
  label/.style = {font=\bfseries, text=blue},
  scale=0.8, every node/.style={transform shape}
]

\node[label, anchor=west] (r2) at (0,0) {\textcolor{blue}{Textless}};
\node[box, below=0.5cm of r2] (in2) {%
  Morlet transform\\
  \texttt{<Tommo-So>} \\ { \color{green} \faVolumeUp}   \texttt{<Amba e djel\`e >}
}; 
\node[box, below=0.5cm of in2] (model2) {
  \begin{tabular}{@{}c@{}}
    \fcolorbox{orange!60}{blue!10}{\parbox{2.5cm}{\centering speech encoder}}\\[2pt]
    \fcolorbox{orange!60}{blue!10}{\parbox{2.5cm}{\centering scalogram decoder}}
  \end{tabular}
};
\node[box, below=0.5cm of model2] (out2) {%
  \texttt{<Senufo>}\\
 scalogram
};
\node[box, below=0.5cm of out2] (voc2) {scalogram vocoder\\  { \color{green} \faVolumeUp} \texttt{<Ya ni s\`e >}};

\draw[arrow] (in2) -- (model2);
\draw[arrow] (model2) -- (out2);
\draw[arrow] (out2) -- (voc2);

\end{tikzpicture}
\caption{Morlet Wavelet-based S2ST }
\end{subfigure}%
\hfill

\begin{subfigure}[t]{0.4\textwidth}
\centering
\begin{tikzpicture}[
  box/.style = {
    draw, thick, fill=cyan!20, 
    minimum width=2.5cm, minimum height=1cm, align=center
  },
  arrow/.style = {
    -{Stealth[length=3mm,width=2mm]}, thick, draw=cyan!60
  },
  label/.style = {font=\bfseries, text=cyan},
  scale=0.8, every node/.style={transform shape}
]

\node[label, anchor=west] (r4) at (0,0) {\textcolor{cyan}{Textless}};
\node[box, below=0.5cm of r4] (in4) {%
  Morlet transform\\
  \texttt{<Tommo-So>} \\ { \color{green} \faVolumeUp}  \texttt{<Amba e djel\`e >}
}; 
\node[box, below=0.5cm of in4] (model4) {
  \begin{tabular}{@{}c@{}}
    \fcolorbox{cyan!60}{cyan!10}{\parbox{2.5cm}{\centering speech encoder}}\\[2pt]
    \fcolorbox{cyan!60}{cyan!10}{\parbox{2.5cm}{\centering unit decoder}}
  \end{tabular}
};
\node[box, below=0.5cm of model4] (out4) {%
  \texttt{<Senufo>}\\
  61 11 ...
};
\node[box, below=0.5cm of out4] (voc4) {unit vocoder\\  { \color{green} \faVolumeUp} \texttt{<Ya ni s\`e >}};

\draw[arrow] (in4) -- (model4);
\draw[arrow] (model4) -- (out4);
\draw[arrow] (out4) -- (voc4);

\end{tikzpicture}
\caption{Unit-based S2ST from Morlet}
\end{subfigure}
\hfill
\begin{subfigure}[t]{0.4\textwidth}
\centering
\begin{tikzpicture}[
  box/.style = {
    draw, thick, fill=orange!20, 
    minimum width=2.5cm, minimum height=1cm, align=center
  },
  arrow/.style = {
    -{Stealth[length=3mm,width=2mm]}, thick, draw=orange!60
  },
  label/.style = {font=\bfseries, text=orange},
  scale=0.8, every node/.style={transform shape}
]

\node[label, anchor=west] (r2) at (0,0) {\textcolor{orange}{Textless}};
\node[box, below=0.5cm of r2] (in2) {%
  Coiflet transform\\
  \texttt{<Tommo-So>} \\ { \color{green} \faVolumeUp}   \texttt{<Amba e djel\`e >}
}; 
\node[box, below=0.5cm of in2] (model2) {
  \begin{tabular}{@{}c@{}}
    \fcolorbox{orange!60}{orange!10}{\parbox{2.5cm}{\centering speech encoder}}\\[2pt]
    \fcolorbox{orange!60}{orange!10}{\parbox{2.5cm}{\centering scalogram decoder}}
  \end{tabular}
};
\node[box, below=0.5cm of model2] (out2) {%
  \texttt{<Senufo>}\\
 scalogram
};
\node[box, below=0.5cm of out2] (voc2) {scalogram vocoder\\  { \color{green} \faVolumeUp} \texttt{<Ya ni s\`e >}};

\draw[arrow] (in2) -- (model2);
\draw[arrow] (model2) -- (out2);
\draw[arrow] (out2) -- (voc2);

\end{tikzpicture}
\caption{Coiflet-based S2ST }
\end{subfigure}%
\hfill

\begin{subfigure}[t]{0.4\textwidth}
\centering
\begin{tikzpicture}[
  box/.style = {
    draw, thick, fill=cyan!20, 
    minimum width=2.5cm, minimum height=1cm, align=center
  },
  arrow/.style = {
    -{Stealth[length=3mm,width=2mm]}, thick, draw=cyan!60
  },
  label/.style = {font=\bfseries, text=cyan},
  scale=0.8, every node/.style={transform shape}
]

\node[label, anchor=west] (r4) at (0,0) {\textcolor{cyan}{Textless}};
\node[box, below=0.5cm of r4] (in4) {%
  Coiflet transform\\
  \texttt{<Tommo-So>} \\ { \color{green} \faVolumeUp} 
}; 
\node[box, below=0.5cm of in4] (model4) {
  \begin{tabular}{@{}c@{}}
    \fcolorbox{cyan!60}{cyan!10}{\parbox{2.5cm}{\centering speech encoder}}\\[2pt]
    \fcolorbox{cyan!60}{cyan!10}{\parbox{2.5cm}{\centering unit decoder}}
  \end{tabular}
};
\node[box, below=0.5cm of model4] (out4) {%
  \texttt{<Senufo>}\\
  61 11 ...
};
\node[box, below=0.5cm of out4] (voc4) {unit vocoder\\  { \color{green} \faVolumeUp} \texttt{<Ya ni s\`e >}};

\draw[arrow] (in4) -- (model4);
\draw[arrow] (model4) -- (out4);
\draw[arrow] (out4) -- (voc4);

\end{tikzpicture}
\caption{Unit-based S2ST from Coiflet}
\end{subfigure}

\caption{Wavelet-based direct S2ST architectures. (b) and (d) use Morlet and Coiflet wavelets as input features followed by a unit vocoder. (a) and (c) represent Textless direct Morlet and Coiflet wavelets S2ST.}
\label{fig:wavelet_s2st}
\end{figure}

\paragraph{Continuous-Time Wavelet-Based  Audio-to-Audio  Pipeline}

We define the inverse Continuous Wavelet Transform (ICWT).
Assuming that the admissibility condition holds, the original signal \( A(t) \) can be recovered by the inverse continuous wavelet transform:
\begin{equation}
A(t) = \frac{1}{C_{\psi}} \int_{0}^{\infty} \int_{-\infty}^{\infty} W(s, \tau) \, \psi_{s,\tau}(t) \, \frac{d\tau \, ds}{s^2},
\end{equation}
where \( C_{\psi} \) is the admissibility constant defined by:
\begin{equation}
C_{\psi} = \int_{0}^{\infty} \frac{|\hat{\psi}(\omega)|^2}{\omega} \, d\omega < \infty,
\end{equation}
and \( \hat{\psi}(\omega) \) denotes the Fourier transform of \( \psi(t) \).

Let \( x_{\text{Tommo}}(t) \in L^2(\mathbb{R}) \) be the source audio waveform in the Tommo-So language.  
The continuous wavelet transform (CWT) of \( x_{\text{Tommo}}(t) \) with respect to a mother wavelet \( \psi(t) \) is:
\begin{equation}
W_{\text{Tommo}}(s', \tau') = \int_{-\infty}^{\infty} x_{\text{Tommo}}(t)\, \psi_{s',\tau'}^*(t)\, dt,
\end{equation}
where
\[
\psi_{s',\tau'}(t) = \frac{1}{\sqrt{s'}} \, \psi\left( \frac{t-\tau'}{s'} \right).
\]

We define a linear operator \( W(s,\tau; s',\tau') \) that maps the source wavelet coefficients to the target wavelet coefficients:
\begin{equation}
\widetilde{W}(s, \tau) = \int_{0}^{\infty} \int_{-\infty}^{\infty} W(s,\tau; s',\tau')\, W_{\text{Tommo}}(s',\tau')\, \frac{d\tau' \, ds'}{(s')^2}.
\end{equation}

Finally, the target speech waveform \( x_{\text{Senufo}}(t) \) is reconstructed via the inverse continuous wavelet transform:
\begin{equation}
x_{\text{Senufo}}(t) = \frac{1}{C_{\psi}} \int_{0}^{\infty} \int_{-\infty}^{\infty} \widetilde{W}(s, \tau)\, \psi_{s,\tau}(t)\, \frac{d\tau \, ds}{s^2},
\tag{12}
\end{equation}
where \( C_{\psi} \) is the admissibility constant associated with the wavelet \( \psi \).

\section*{Multiscale Audio-Semantic Transform (MAST)}

To model speech in audio-rich, low-text-resource, unwritten but highly oral languages, we propose the \textit{Multiscale Audio-Semantic Transform} (MAST), a novel representation that extends traditional time-frequency transforms by integrating semantic, prosodic, and speaker-aware information into a unified multiresolution framework. MAST is designed to capture lexical tone, rhythm, speaker identity, and expressive features crucial for textless  Audio-to-Audio  translation and understanding.

Let \( A(t) \in L^2(\mathbb{R}) \) denote a real-valued audio signal. The MAST representation is defined as a tuple:

\[
\mathcal{M}(A)(\tau) = \left( W(s, \tau), \, P(\tau), \, E_k(\tau), \, S(\tau), \, I(\tau) \right)
\]

where each component is aligned in time \( \tau \), and captures complementary aspects of the signal:

\begin{enumerate}
    \item \textbf{Wavelet Time–Frequency Representation:}
    \[
    W(s, \tau) = \int_{-\infty}^{\infty} A(t) \, \psi^*_{s,\tau}(t) \, dt
    \quad \text{with} \quad    \psi \text{being Morlet/Gabor or Coiflet}
    \]
    is the Morlet wavelet transform, where \( \psi^*_{s,\tau}(t) \) is the scaled and shifted version of the mother wavelet.

    \item \textbf{Pitch Contour:}
    \[
    s^*(\tau) = \arg\max_s |W(s, \tau)|^2, \quad
    P(\tau) = \frac{\omega_0}{2\pi s^*(\tau)}
    \]
    yields the instantaneous pitch via dominant wavelet scale \( s^*(\tau) \).

    \item \textbf{Prosodic Embedding:}
    \[
    E_k(\tau) = \Phi_\text{prosody}(A_{[\tau - \Delta, \tau + \Delta]}) \in \mathbb{R}^d
    \]
    where \( \Phi_\text{prosody} \) is a learnable encoder capturing rhythm, energy, and duration features locally.

    \item \textbf{Speaker Embedding:}
    \[
    S(\tau) = \Psi_\text{speaker}(A_{[\tau - T, \tau + T]}) \in \mathbb{R}^p
    \]
    where \( \Psi_\text{speaker} \) maps the local signal to a fixed-length identity embedding using a pretrained speaker model.

    \item \textbf{Intonation Gesture Descriptor:}
    \[
    I(\tau) = \left[ \frac{dP}{d\tau}, \, \frac{d^2P}{d\tau^2}, \, E(\tau), \, D(\tau) \right] \in \mathbb{R}^q
    \]
    encodes pitch dynamics, short-term energy \( E(\tau) \), and phonation duration \( D(\tau) \), offering a fine-grained model of expressive speech gestures.
\end{enumerate}

\subsection*{Space and Alignment}

The full transform is a time-aligned tensor field:
\[
\mathcal{M}(A) \in L^2(\mathbb{R}^2 \times \mathbb{R}^{d + p + q})
\]
mapping time \( \tau \) and scale \( s \) to a multichannel embedding space that integrates acoustic structure, prosodic content, and semantic cues.

Unlike conventional Mel spectrograms, which rely on fixed time–frequency resolution and omit higher-order speech attributes, MAST provides \textit{Adaptive temporal and spectral resolution} via wavelet scaling, 
     \textit{Lexical tone sensitivity} through pitch-aware constructs,
    \textit{Speaker and emotion awareness} via learned embedding functions,
    \textit{Suitability for textless models} operating purely on audio.
This structure enables downstream models such as audio transformers or diffusion decoders to condition on semantically rich, language-independent features, supporting robust learning even in the absence of orthographic data.

\subsection*{Example Functions for Prosody and Speaker Embeddings}

In practical implementations, the functions \( \Phi_{\text{prosody}} \) and \( \Psi_{\text{speaker}} \) can be instantiated as follows:

\paragraph{Prosodic Encoder \boldmath\( \Phi_{\text{prosody}} \):}
We define \( \Phi_{\text{prosody}}: L^2([\tau - \Delta, \tau + \Delta]) \to \mathbb{R}^d \) as a local convolutional encoder:
\[
\Phi_{\text{prosody}}(x) = \text{hn} \left( W_2 \cdot \text{hn}(W_1 \cdot x + b_1) + b_2 \right)
\]
where \( W_1 \in \mathbb{R}^{h \times n} \), \( W_2 \in \mathbb{R}^{d \times h} \), and \( h \) is a hidden dimension. The input \( x \in \mathbb{R}^n \) is a vector of concatenated acoustic features such as  short-term energy, zero-crossing rate, pitch windows, extracted from the interval around \( \tau \).
Alternatively, \( \Phi_{\text{prosody}} \) can be a transformer block applied to local wavelet coefficients, with positional encoding for \(\tau\), allowing the model to learn rhythmic and tonal attention patterns.

\paragraph{Speaker Encoder \boldmath\( \Psi_{\text{speaker}} \):}
We define \( \Psi_{\text{speaker}}: L^2([\tau - T, \tau + T]) \to \mathbb{R}^p \) as a pretrained transformer embedding:
\[
\Psi_{\text{speaker}}(x) = \text{pt}(x)
\]
where \( x \) is a 1.5-3 second window of waveform or spectrogram input. The pretrained transformer produces a speaker embedding in \( \mathbb{R}^p \) optimized for identity discrimination under contrastive loss. This embedding captures voice characteristics, pitch range, vocal tract shape, and speaking style.

These mappings \( \Phi \) and \( \Psi \) can be pretrained independently or jointly fine-tuned for downstream tasks such as  Audio-to-Audio  translation, emotion recognition, or voice disentanglement.

\subsection*{Speaker Embedding  Architecture}

The transformer model is a deep speaker embedding architecture that extends the Time Delay Neural Network by incorporating channel attention, multi-scale temporal aggregation, and residual propagation. Let \( X = [x_1, x_2, \dots, x_T] \in \mathbb{R}^{d \times T} \) denote a sequence of acoustic features, where \( T \) is the number of frames and \( d \) is the feature dimension.
The embedding is a function:
\[
\Psi_{\text{speaker}}: \mathbb{R}^{d \times T} \to \mathbb{R}^p, \quad s = \Psi_{\text{speaker}}(X),
\]
which encodes speaker identity using temporally aggregated, channel-attended features optimized for discrimination under metric learning losses.
Each frame \( x_t \) is processed using temporal context windows:
\[
h^{(l)}_t = hn\left(W^{(l)} \cdot [x_{t - \Delta}, \dots, x_{t + \Delta}] + b^{(l)}\right),
\]
where \( \Delta \) is the context size, \( W^{(l)} \in \mathbb{R}^{d' \times (2\Delta+1)\cdot d} \), and hn is a radial compactification projective norm called holonorm (hn).
A \textit{Res2Net} block splits the input into \( s \) channel groups, applies group-wise  Time Delay Neural Networks  (TDNNs), and combines via residual connections:
\[
y = x + \sum_{i=1}^{s} \text{TDNN}_i(z_i),
\quad z = \text{Split}(x),
\]
where each \( \text{TDNN}_i \) applies convolution to its own scale, enabling multi-resolution modeling.
A \textit{Squeeze-and-Excitation (SE)} module reweights channels:
\[
\alpha = hn\left(W_2\, \delta(W_1\, \text{GlobalAvg}(y))\right), \quad y' = \alpha \cdot y,
\]
where \( \delta \) is ReLU,  and \( \alpha \in \mathbb{R}^C \) is broadcast-multiplied across the feature map.
Instead of simple average pooling, ECAPA-TDNN uses attention to compute weighted mean and standard deviation:
\[
e_t = \text{hn}(W_a\, h_t + b_a), \quad \alpha_t = \frac{\exp(w^\top e_t)}{\sum_{t'} \exp(w^\top e_{t'})},
\]
\[
\mu = \sum_{t} \alpha_t h_t, \quad \sigma = \sqrt{ \sum_t \alpha_t (h_t - \mu)^2 },
\quad h_{\text{pool}} = [\mu; \sigma],
\]
where \( h_{\text{pool}} \in \mathbb{R}^{2d} \) is the fixed-length embedding.
The pooled vector is projected into a speaker embedding:
\[
s = W_{\text{proj}}\, h_{\text{pool}} + b_{\text{proj}}, \quad s \in \mathbb{R}^{p},
\]
which is then normalized and passed to an additive margin softmax (AM-Softmax) classifier during training.

\begin{table}[htb]
\centering
\begin{tabular}{|p{2cm}|p{3cm}|p{3cm}||p{3cm}|}
\hline
\textbf{Feature / Metric} & \textbf{Spectrogram} & \textbf{Wavelet (Morlet CWT)} & \textbf{MAST} \\
\hline
Representation
  & Fixed-window STFT (time–freq matrix)
  & Scale-adaptive CWT (scalogram)
  & Tuple \((W,P,E,S,I)\) \\
Complexity
  & \(O(N\log N)\) per frame
  & \(O(N\cdot S)\) per frame
  & \(O(N\cdot S + C_\phi + C_\chi)\) \\
Energy consumption
  & Low
  & Medium
  & High \\
Memory usage
  & \(O(T\times F)\)
  & \(O(T\times S)\)
  & \(O\bigl(T\times (S+d+p+q)\bigr)\) \\
Audio quality (S2ST)
  & Good envelope, poor tonal detail
  & Improved tone localization, moderate noise
  & Best: preserves tone, prosody, speaker identity \\
\hline
\end{tabular}
\caption{Comparison of Spectrogram, Wavelet, and MAST representations for textless audio-to-audio Transformers. \(N\)=samples/frame, \(T\)=time frames, \(F\)=frequency bins, \(S\)=wavelet scales, \(d,p,q\)=embedding dimensions, \(C_\phi,C_\chi\)=encoder costs.}
\label{tab:representation_comparison}
\end{table}

\subsection*{MAST with Coiflet Wavelet Representation}

In order to enhance multi-resolution analysis while preserving signal regularity and vanishing moments, we propose replacing the Morlet wavelet in MAST with the \textit{Coiflet wavelet}, a real-valued orthogonal wavelet with higher-order vanishing moments and compact support. This enables more accurate time–scale decomposition of phonetic and prosodic content in speech signals.

\paragraph{MAST with Coiflet Input.}
We now redefine the MAST representation as:
\[
\mathcal{M}_{\text{Coif}}(A) = \left( D^{(j)}[n], \, P(n), \, E_k(n), \, S(n), \, I(n) \right),
\]
where:
\begin{itemize}
  \item \( D^{(j)}[n] \) is the Coiflet DWT detail coefficient at scale \( j \),
  \item \( P(n) \in \mathbb{R} \) is the pitch estimate derived from zero-crossing or autocorrelation methods,
  \item \( E_k(n) \in \mathbb{R}^d \) is a learned prosodic embedding,
  \item \( S(n) \in \mathbb{R}^p \) is a speaker embedding ,
  \item \( I(n) \in \mathbb{R}^q \) encodes local intonation gesture features.
\end{itemize}

\subsection*{ Unit-based Textless Audio2Audio Translation}
We revisit the textless translation architecture proposed in \cite{huang2022transpeech} using transformers. Interestingly in figure \ref{fig:s2st} the translation from language 1 audio to language 2 audio differs from language 2 audio  to language 1 audio.  
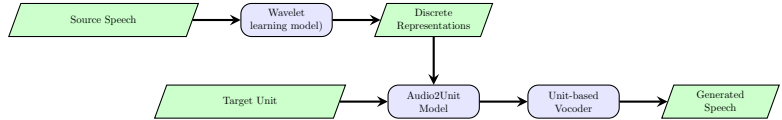
\begin{figure}[h!]
\centering
\begin{tikzpicture}[node distance=1.6cm and 1.6cm, scale=0.4, every node/.style={transform shape}]

\node (source) [data2] {Source Speech};
\node (ssl) [block2, right=of source] {Wavelet\\ learning model)};
\node (discrete) [data2, right=of ssl] {Discrete\\ Representations};
\node (s2ut) [block2, below=of discrete] {Audio2Unit\\ Model};
\node (targetunit) [data2, left=of s2ut] {Target Unit};
\node (vocoder) [block2, right=of s2ut] {Unit-based Vocoder};
\node (generated) [data2, right=of vocoder] {Generated Speech};

\draw [arrow] (source) -- (ssl);
\draw [arrow] (ssl) -- (discrete);
\draw [arrow] (discrete) -- (s2ut);
\draw [arrow] (targetunit) -- (s2ut);
\draw [arrow] (s2ut) -- (vocoder);
\draw [arrow] (vocoder) -- (generated);

\end{tikzpicture}
\caption{A Unit-based Textless Direct Audio-to-Audio Translation.}
\label{fig:s2st}
\end{figure}

\section{ Audio Diffusion with MAST Conditioning}

Recent literature highlights the  multidisciplinary applications of generative diffusion models across machine intelligence, science, and engineering. In \cite{audiodiffusion1}, a comprehensive survey explores the theoretical foundations, model architectures, and diverse applications of generative diffusion models. The authors in \cite{audiodiffusion2} focus on vision-based tasks, providing a taxonomy of diffusion approaches in computer vision. A broader analysis of the opportunities and challenges in generative machine intelligence is presented in \cite{audiodiffusion3}, emphasizing both model capabilities and practical limitations. Diffusion models have shown strength in specialized tasks such as video inpainting \cite{audiodiffusion4}, missing data imputation \cite{audiodiffusion5}, and adversarial robustness \cite{audiodiffusion6}. As proposed in \cite{audiodiffusion7}, text-to-image generation continues to mature, while \cite{audiodiffusion8, audiodiffusion14} investigate physics-inspired perspectives and theoretical underpinnings. Application domains have expanded to include satellite lidar super-resolution \cite{audiodiffusion9}, parametric PDE solving \cite{audiodiffusion10}, UI design and HCI \cite{audiodiffusion11, audiodiffusion12}, and geospatial foundation modeling \cite{audiodiffusion13}. In molecular science, diffusion-based generation supports protein-ligand binding modeling \cite{audiodiffusion16} and RNA structure prediction \cite{audiodiffusion18}, while \cite{audiodiffusion19} explores nonequilibrium physics underlying diffusion processes. Additional works include graph-based network planning \cite{audiodiffusion20}, precipitation nowcasting \cite{audiodiffusion21}, conditional adversarial defenses \cite{audiodiffusion22}, and efficient learning in high-dimensional spaces \cite{audiodiffusion23}. Finally, the review in \cite{audiodiffusion24} details the use of diffusion models in audio restoration.

Recent advancements in generative machine intelligence models have brought attention to the limitations of standard Brownian-driven stochastic dynamics, particularly in capturing long-range dependencies and memory effects in  audio, music, speech and video generation. Here, we introduce \textbf{Fractional Transfusion}, a novel class of diffusion transformers driven by \textit{fractional Brownian motion (fBm)}, a generalization of classical Brownian motion characterized by the Hurst parameter $H \in (0,1).$ This framework enables the modeling of subdiffusive ($H<0.5$) and superdiffusive ($H>0.5$) behaviors, thereby extending the capacity of diffusion models to represent data distributions with  temporal and spatial correlations.

We extend the textless audio-to-audio diffusion framework to continuous time, incorporating  MAST as conditioning information. Let \(A(t)\in L^2(\mathbb{R})\) denote the audio signal at time \(t\), with \(A(0)=A_0\) being the original waveform. The forward diffusion process is modeled by the following stochastic differential equation (SDE):

\[
dA_t = -\theta(t)(\bar{m}(t)-A_t)dt + \sigma(t)\,dB^H_t,\quad t\in[0,T],
\]

where \(\sigma(t)\) is a predefined noise schedule and \(W^H_t\) is a fractional Brownian motion with hurst parameter $H\in (0.5,1),$ 
$\theta(t), \bar{m}(t)$ are deterministic functions.

In parallel, we compute the MAST representation from \(A_0\), i.e.,

\[
\mathcal{M}(A_0) = \bigl(W(s,\tau),\,P(\tau),\,E_k(\tau),\,S(\tau),\,I(\tau)\bigr),
\]

and embed it via a learnable function \(\mathrm{Embed}(\mathcal{M}(A_0))\) to obtain the conditioning vector \(M\).

The reverse-time SDE, which recovers the clean signal from noise, is given by

\[
dA_t = \Bigl[\theta(t)(\bar{m}(t)-A_t) + 2H t^{2H-1}\sigma^2(t)\frac{ A_t-m(t) }{v^2(t)}\Bigr]dt + \sigma(t)\,d\bar{B}^H_t,
\]

where \(d\bar{B}^H_t\) is the reverse Brownian motion, $m(t), {v}(t)$ are deterministic functions obtained the mean-field of states.

\subsection{Fractional Ornstein--Uhlenbeck Process }

Our method builds upon an explicit solution of a fractional Ornstein--Uhlenbeck Process with time-varying drift and mean. 

\begin{proposition}
Let $x(t)$ satisfy the stochastic differential equation
\[
dx(t) = \theta(t)\big(\bar{m}(t) - x(t)\big)\,dt + \sigma(t)\,dB^H(t), \quad x(0) = x_0,
\]
where
 $B^H(t)$ is a fractional Brownian motion with Hurst parameter $H \in \left(\frac{1}{2}, 1\right)$,
    $\theta(t), \bar{m}(t), \sigma(t) \in C^1([0,T])$ are real-valued functions,
    and $\sigma(t)$ is Hölder continuous with exponent $\delta > 1 - H$.

Then the unique solution $x(t)$ to this SDE exists and is given by
\[
x(t) = e^{-\Phi(t)} x_0 + e^{-\Phi(t)} \int_0^t e^{\Phi(t')} \theta(t') \bar{m}(t')\,dt' + e^{-\Phi(t)} \int_0^t e^{\Phi(t')} \sigma(t')\,dB^H(t'),
\]
where $\Phi(t) := \int_0^t \theta(t')\,dt'$, and the stochastic integral is understood in the Young sense \cite{young1936}.
\end{proposition}

\begin{proof}
We begin by rewriting the SDE in linear form:
\[
dx(t) + \theta(t) x(t)\,dt = \theta(t) \bar{m}(t)\,dt + \sigma(t)\,dB^H(t).
\]
Define the integrating factor
\[
\mu(t) := e^{\int_0^t \theta(t')\,dt'} = e^{\Phi(t)}.
\]
Multiplying both sides of the equation by $\mu(t)$, we obtain
\[
\mu(t) dx(t) + \mu(t) \theta(t) x(t)\,dt = \mu(t) \theta(t) \bar{m}(t)\,dt + \mu(t) \sigma(t)\,dB^H(t),
\]
which can be recognized as the total differential
\[
d[\mu(t)x(t)] = \mu(t) \theta(t) \bar{m}(t)\,dt + \mu(t) \sigma(t)\,dB^H(t).
\]
Integrating both sides from $0$ to $t$, we get
\[
\mu(t)x(t) = x_0 + \int_0^t \mu(t') \theta(t') \bar{m}(t')\,dt' + \int_0^t \mu(t') \sigma(t')\,dB^H(t').
\]
Solving for $x(t)$ yields
\[
x(t) = e^{-\Phi(t)} x_0 + e^{-\Phi(t)} \int_0^t e^{\Phi(t')} \theta(t') \bar{m}(t')\,dt' + e^{-\Phi(t)} \int_0^t e^{\Phi(t')} \sigma(t')\,dB^H(t').
\]

We now justify the existence of the Young integral. Since $B^H(t)$ has Hölder continuous paths of order $\gamma < H$, and $\sigma(t)$ is Hölder continuous of order $\delta > 1 - H$, the composition $e^{\Phi(t')} \sigma(t')$ is also Hölder continuous of order $> 1 - H$, ensuring that the integral
\[
\int_0^t e^{\Phi(t')} \sigma(t')\,dB^H(t')
\]
is well-defined in the Young sense.

Finally, the linearity of the SDE and  regularity assumptions imply the uniqueness of the solution.

\end{proof}

\begin{proposition}
Let \( x(t) \) be the solution to the fractional stochastic differential equation
\[
dx(t) = \theta(t)(\bar{m}(t) - x(t))\,dt + \sigma(t)\,dB^H(t), \quad x(0) = x_0,
\]
where \( \theta(t), \bar{m}(t), \sigma(t) \) are continuous deterministic functions on \( [0,T] \), and \( B^H(t) \) is a fractional Brownian motion with Hurst parameter \( H > \frac{1}{2} \). Then \( x(t) \) is a Gaussian process with mean
\[
m(t) = e^{-\Phi(t)} x_0 + e^{-\Phi(t)} \int_0^t e^{\Phi(t')} \theta(t') \bar{m}(t')\,dt',
\]
and variance
\[
v^2(t) = e^{-2\Phi(t)} H(2H-1) \int_0^t \int_0^t e^{\Phi(t')} \sigma(t') e^{\Phi(s')} \sigma(s') |t' - s'|^{2H - 2} \,dt' \,ds',
\]
where \( \Phi(t) = \int_0^t \theta(s)\,ds \).
\end{proposition}

\begin{proof}
The solution \( x(t) \) admits the explicit form
\[
x(t) = e^{-\Phi(t)} x_0 + e^{-\Phi(t)} \int_0^t e^{\Phi(t')} \theta(t') \bar{m}(t')\,dt' + e^{-\Phi(t)} \int_0^t e^{\Phi(t')} \sigma(t')\,dB^H(t').
\]
The first two terms are deterministic, so the randomness comes solely from the last term, which is a Young integral with respect to \( B^H \). Since the integrand is deterministic and \( B^H \) is Gaussian, the integral is Gaussian as well, and hence \( x(t) \) is Gaussian.

The mean of \( x(t) \) is thus
\[
\mathbb{E}[x(t)] = e^{-\Phi(t)} x_0 + e^{-\Phi(t)} \int_0^t e^{\Phi(t')} \theta(t') \bar{m}(t')\,dt' =: m(t),
\]
and its variance is
\[
\text{Var}(x(t)) = \mathbb{E}\left[ \left( e^{-\Phi(t)} \int_0^t e^{\Phi(t')} \sigma(t')\,dB^H(t') \right)^2 \right].
\]
Since for \( H > \frac{1}{2} \), the Young integral satisfies:
\[
\mathbb{E}\left[ \left( \int_0^t f(t')\,dB^H(t') \right)^2 \right] = H(2H-1) \int_0^t \int_0^t f(t') f(s') |t' - s'|^{2H-2} \,dt' \,ds',
\]
we substitute \( f(t') = e^{\Phi(t')} \sigma(t') \) and factor out \( e^{-2\Phi(t)} \), yielding
\[
v^2(t) = e^{-2\Phi(t)} H(2H-1) \int_0^t \int_0^t e^{\Phi(t')} \sigma(t') e^{\Phi(s')} \sigma(s') |t' - s'|^{2H - 2} \,dt' \,ds'.
\]
\end{proof}

\subsection{Forward process that ends up with a mask}
To ensure accurate terminal constraints critical for conditional generation, we design a tailored process that begins at $x_0$ and precisely reaches a terminal distribution $\mathcal{N}(m^*_T,  (\sigma_T^*)^2)$ at time $t = T$. This process incorporates an adaptive drift:

\begin{proposition}
Let \( T > 0 \) be fixed and let \( m_T^* \in \mathbb{R} \), \( \sigma_T^* > 0 \). Then for any deterministic, continuous, and strictly positive function \( \theta(t) \) on \( [0,T] \), and for any initial condition \( x(0) = x_0 \), there exist deterministic, continuous functions \( \bar{m}(t) \) and \( \sigma(t) \), such that the solution \( x(t) \) to the fractional SDE
\[
dx(t) = \theta(t)(\bar{m}(t) - x(t))\,dt + \sigma(t)\,dB^H(t), \quad x(0) = x_0,
\]
satisfies
\[
x(T) \sim \mathcal{N}(m_T^*, (\sigma_T^*)^2).
\]
\end{proposition}

\begin{proof}
Let us define:
\[
\Phi(t) := \int_0^t \theta(s)\,ds.
\]

We know from the explicit solution of the SDE that:
\[
x(T) = e^{-\Phi(T)} x_0 + e^{-\Phi(T)} \int_0^T e^{\Phi(t')} \theta(t') \bar{m}(t')\,dt' + e^{-\Phi(T)} \int_0^T e^{\Phi(t')} \sigma(t')\,dB^H(t').
\]

\textbf{Mean constraint:}

To match the mean \( m_T^* \), define \( \bar{m}(t) \) such that
\[
m_T^* = e^{-\Phi(T)} x_0 + e^{-\Phi(T)} \int_0^T e^{\Phi(t')} \theta(t') \bar{m}(t')\,dt'.
\]
Rewriting, we require
\[
\int_0^T e^{\Phi(t')} \theta(t') \bar{m}(t')\,dt' = e^{\Phi(T)} (m_T^* - e^{-\Phi(T)} x_0).
\]
This is a linear integral equation for \( \bar{m}(t') \), which admits a solution under broad conditions. In particular, since \( \theta(t') \) and \( e^{\Phi(t')} \) are strictly positive, a continuous \( \bar{m}(t') \) exists.

Solving for \( \bar{m} \) to match \( \mathbb{E}[x(T)] = m_T^* \) gives
\[
\bar{m} = \frac{m_T^* - x_0 e^{-\Phi(T)}}{1 - e^{-\Phi(T)}}.
\]

\textbf{Variance constraint:}

Let \( f(t') := e^{\Phi(t')} \sigma(t') \). To achieve \( \text{Var}(x(T)) = (\sigma_T^*)^2 \), we require
\[
(\sigma_T^*)^2 = e^{-2\Phi(T)} H(2H-1) \int_0^T \int_0^T f(t') f(s') |t' - s'|^{2H - 2} \,dt' \,ds'.
\]

Let us define:
\[
K := H(2H - 1) \int_0^T \int_0^T f(t') f(s') |t' - s'|^{2H - 2} \,dt' \,ds'.
\]
Then
\[
K = e^{2\Phi(T)} (\sigma_T^*)^2.
\]

This is a homogeneous quadratic form in \( f(t') \). Since the kernel \( |t' - s'|^{2H - 2} \) is symmetric and positive definite for \( H > 1/2 \), a solution for \( f(t') \) (and hence \( \sigma(t') = f(t') e^{-\Phi(t')} \)) exists, for instance, by choosing:
\[
f(t') := C, \quad \text{with } C = \left( \frac{e^{2\Phi(T)} (\sigma_T^*)^2}{H(2H - 1) \int_0^T \int_0^T |t' - s'|^{2H - 2} dt' ds'} \right)^{1/2}.
\]
Thus, a constant \( f(t') \) yields an explicit solution:
\[
\sigma(t') = C e^{-\Phi(t')}.
\]

The double integral is symmetric and evaluates as:
\[
\int_0^T \int_0^T |t' - s'|^{2H - 2} dt' ds' = 2 \int_0^T \int_0^{t'} (t' - s')^{2H - 2} ds' dt' = \frac{2T^{2H}}{(2H - 1)(2H)}.
\]

\[
C = \sigma_T^* T^{-H} e^{\Phi(T)}, \ 
\]

Hence, both the mean and variance constraints can be satisfied simultaneously by choosing appropriate \( \bar{m}(t) \) and \( \sigma(t) \).
\end{proof}

\begin{corollary}
This result shows you can steer a linear SDE in fractional noise exactly to any Gaussian target.
Its especially useful in diffusion-based generative modeling, where sampling from a terminal Gaussian is desired.
We define the dynamics of \( x(t) \in \mathbb{R} \) by the fractional stochastic differential equation:
\[
dx(t) = \theta(t) \left( \frac{m_T^* - x_0 e^{-\Phi(T)}}{1 - e^{-\Phi(T)}} - x(t) \right) dt + \sigma_T^*  T^{-H} e^{\Phi(T) - \Phi(t)} dB^H(t), \quad x(0) = x_0,
\]
where
$\Phi(t) := \int_0^t \theta(t')\,dt', $
and \( B^H(t) \) is a fractional Brownian motion with Hurst parameter \( H > \frac{1}{2} \), \( \theta(t) > 0 \) is deterministic and continuous, and \( x_0, m_T^*, \sigma_T^* \in \mathbb{R} \) are given constants.
Under this dynamics, the process \( x(t) \) satisfies $x(T) \sim \mathcal{N}(m_T^*, (\sigma_T^*)^2). $
\end{corollary}

This construction ensures that the diffusion process not only adheres to the target endpoint distribution but also embeds the long-memory structure characteristic of fractional noise. By exploiting the non-local properties of fBm, \textbf{fractional diffusion,  super diffusion, subdiffusion} provide a principled way to enrich generative models across modalities spanning time-series, vision, audio, and video with a richer stochastic representation of uncertainty and correlation.

\subsection{Explicit score function}

Despite the recent surge in popularity of score-based generative modeling, a common misconception persists: that one must always train a neural network to approximate the score function of the forward process. This belief overlooks fundamental properties of Gaussian processes.

\begin{proposition} Let $x(t)$ be a stochastic process such that for every time $t$, 
\[
x(t) \sim \mathcal{N}(m(t), v^2(t)),
\]
where $m(t)$ and $v^2(t) > 0$ are deterministic, differentiable functions of time. Then the score function (i.e., the gradient of the log-density with respect to $x$) at time $t$ is given analytically by:
\[
\nabla_x \log p(x(t)) = -\frac{x(t) - m(t)}{v^2(t)}.
\]
Thus, no neural network is required to estimate the score function when the marginal distribution is Gaussian with known mean and variance. 
\end{proposition}

\begin{proof}

By definition, the probability density function (pdf) of a Gaussian random variable $x(t) \sim \mathcal{N}(m(t), v^2(t))$ is:
\[
p(x(t)) = \frac{1}{\sqrt{2\pi v^2(t)}} e^{ -\frac{(x(t) - m(t))^2}{2v^2(t)} }.
\]

The log-density is then:
\[
\log p(x(t)) = -\frac{1}{2} \log(2\pi v^2(t)) - \frac{(x(t) - m(t))^2}{2v^2(t)}.
\]

Differentiating with respect to $x(t)$:
\[
\nabla_x \log p(x(t)) = \frac{\partial}{\partial x(t)} \left( -\frac{(x(t) - m(t))^2}{2v^2(t)} \right)
= -\frac{1}{2v^2(t)} \cdot 2(x(t) - m(t)) = -\frac{x(t) - m(t)}{v^2(t)}.
\]

Hence, the score function is:
\[
\nabla_x \log p(x(t)) = -\frac{x(t) - m(t)}{v^2(t)}.
\]

\end{proof}

This says that the score function of a Gaussian is explicit, \textit{linear} function in $y.$ This form is derived directly from the analytical expression of the Gaussian density and requires \textbf{no} learning or approximation when the mean and variance are known. Therefore, in purely Gaussian settings, the score is \textbf{exactly known by construction}, and using deep neural networks to estimate it introduces unnecessary complexity, computational burden, and potential approximation error.
In any Gaussian process, the form of the score is always this linear expression. What may be unknown in practice is:
the mean, which might be time-dependent.
The variance, which could also vary with time or depend on latent dynamics.
The true modeling difficulty lies not in estimating the score structure, but in estimating (or tracking) these parameters over time  especially in non-stationary or parameter-varying settings.
The misconception in many score-based diffusion models is in treating the form of the score as unknown when it is  the parameter evolution that is uncertain or unobserved.

This insight extends naturally to more general Gaussian processes, including those governed by \textit{fractional Brownian motion (fBm)} and \textit{Gauss-Volterra processes}. In both subdiffusion ($H < 0.5$) and superdiffusion ($H > 0.5$) regimes, the marginal distributions at each time point remain Gaussian due to the underlying linear dynamics and Gaussian noise. Consequently, their corresponding score functions also admit closed-form expressions, preserving the linearity in the state variable.

\subsection{Sub/Super  diffusion generative model of mean-field type }
We introduce a  mean-field-type fractional diffusion generative model. The meqn-field type terms are mainly involved in the backward process where the denoising pass requires the mean and the variance to match the original process. We do not need the entire mean field state distribution in this particular case, we only need the first two moments.
\subsection*{ Fractional Forward Dynamics }
\[
dx(t) = \theta(t)(\bar{m}(t) - x(t))\,dt + \sigma(t)\,dB^H(t), \quad x(0) = x_0,
\]
\subsection*{ Distribution of $x(t)$}
The solution \( x(t) \) admits the explicit form
\[
x(t) = e^{-\Phi(t)} x_0 + e^{-\Phi(t)} \int_0^t e^{\Phi(t')} \theta(t') \bar{m}(t')\,dt' + e^{-\Phi(t)} \int_0^t e^{\Phi(t')} \sigma(t')\,dB^H(t').
\]
where $\Phi(t) := \int_0^t \theta(t')\,dt'.$  In this deterministic coefficient setting and deterministic initial signal $x_0,$ the random variable $x(t)$  has the same distribution as the Gaussian $\mathcal{N}(m(t), v^2(t)).$
where 
\begin{equation}
    \begin{array}{l}      
m(t) = e^{-\Phi(t)} x_0 + e^{-\Phi(t)} \int_0^t e^{\Phi(t')} \theta(t') \bar{m}(t')\,dt',\\
v^2(t) = e^{-2\Phi(t)} H(2H-1) \int_0^t \int_0^t e^{\Phi(t')} \sigma(t') e^{\Phi(s')} \sigma(s') |t' - s'|^{2H - 2} \,dt' \,ds',
    \end{array}
\end{equation}

\subsection*{ Endpoint Conditioning }

\[
dx(t) = \theta(t) \left( \frac{m_T^* - x_0 e^{-\Phi(T)}}{1 - e^{-\Phi(T)}} - x(t) \right) dt + \sigma_T^*  T^{-H} e^{\Phi(T) - \Phi(t)} dB^H(t), \quad x(0) = x_0,
\]
where satisfies $x(T) \sim \mathcal{N}(m_T^*, (\sigma_T^*)^2). $

\subsection*{Reverse-Time fractional SDE:}
Given the forward process is Gaussian, the reverse-time dynamics are governed by:
\begin{equation}
\begin{array}{l}
dx(t) = 
\left[ \theta \left( \frac{m^*_T - e^{-\Phi(T)} x_0}{1 - e^{-\Phi(T)}} - x(t) \right) + \frac{x(t) - m(t)}{t} 2H  \right] dt \\ 
+ (\sigma_T^*)  T^{-H} e^{\Phi(T) - \Phi(t)}\, d\bar{B}^H(t),
\\
\\
m(t) =
e^{-\Phi(t)} x_0 +(1-e^{-\Phi(t)}) (\frac{m^*_T - e^{-\Phi(T)} x_0}{1 - e^{-\Phi(T)}}) ,
\\
v^2(t) =  e^{2\Phi(T)-2\Phi(t)} (\frac{t}{T})^{2H} (\sigma_T^*)^2,

\end{array}
\end{equation}
where $\bar{B}^H(t)$ is a time-reversed fractional Brownian motion, and $\tilde{\sigma}(t)$ adjusts for non-Markovian memory in the reverse direction.

\subsection*{Score Function}
Because this process remains Gaussian, the exact score function is known analytically:
\begin{equation}
\nabla_y \log p_{x(T-t)}(y) = -\frac{y - m(T-t)}{v^2(T-t)}.
\end{equation}

\subsection{Fractional  mean-field type transfusion  architecture}

Let $x(t_1), x(t_2), \dots, x(t_N)$ be a sequence of states from a fractional diffusion process. Define the sequence input to the transformer as:
\[
z_i^{(0)} = \mathrm{Embed}(x(t_i)) + \mathrm{PE}(t_i), \quad i \in \{ 1, \dots, N\},
\]
where
    $\mathrm{Embed}: \mathbb{R}^d \to \mathbb{R}^D$ is a learned embedding,
    $\mathrm{PE}(t_i)$ is a positional encoding adapted to fractional time,
    $z_i^{(0)} \in \mathbb{R}^D$ is the initial token embedding.

\subsection*{Holonorm}

Define the Holonorm or  ProjectiveNorm map $hn: \mathbb{R}^D \to \mathbb{R}^D$ as:
\[
hn(x) := \frac{x}{1 + \|x\|}.
\]

This operator is used throughout the network both for normalization and nonlinearity. Figures \ref{figholonorm00}, \ref{holonorm_1d_scaled} and \ref{figholonorm01} plot in 1D vs 2D the holonorm and related functions.

\begin{figure}
\centering
        \includegraphics[width=\linewidth]{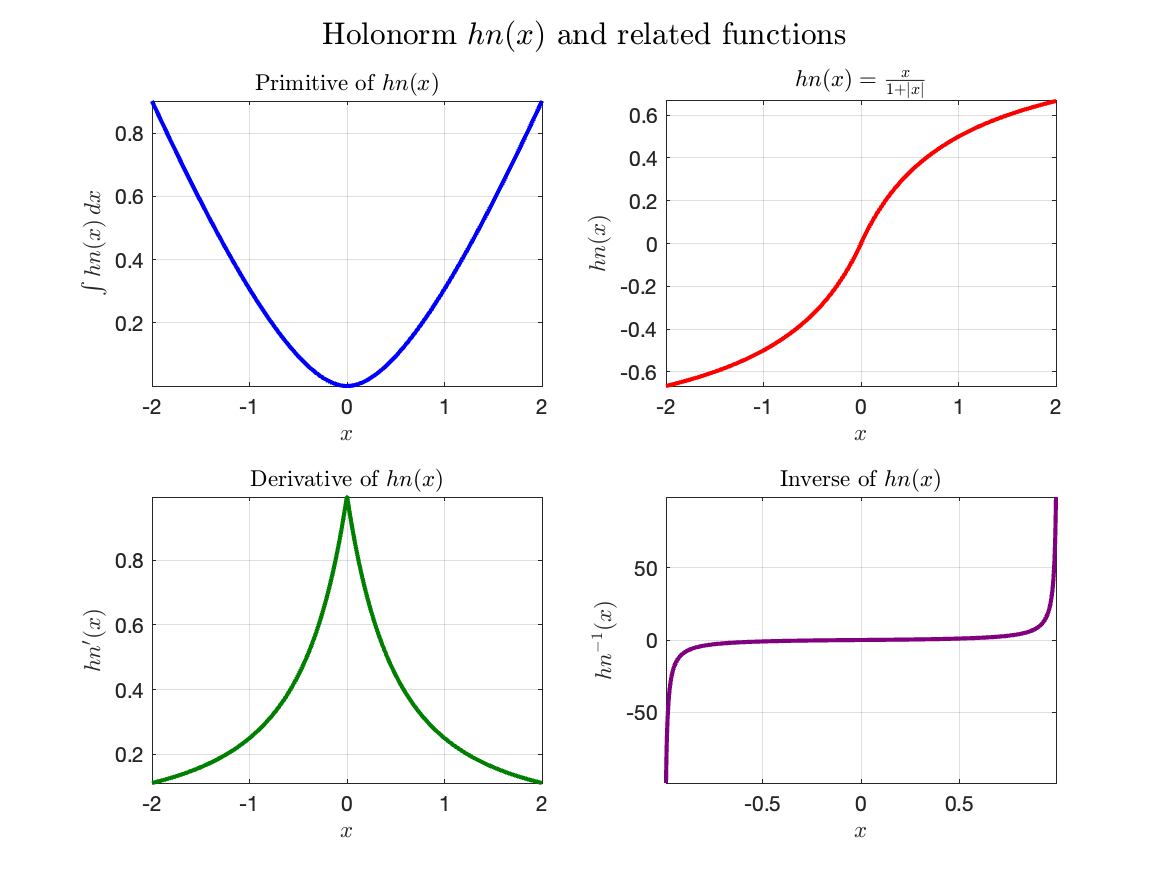}
\caption{Holonorm map \( hn(x) = \frac{x}{1 + |x|} \) in 1D and  its primitive, derivative, inverse map} \label{figholonorm00}
\end{figure}
\begin{figure}
\centering
        \includegraphics[width=\linewidth]{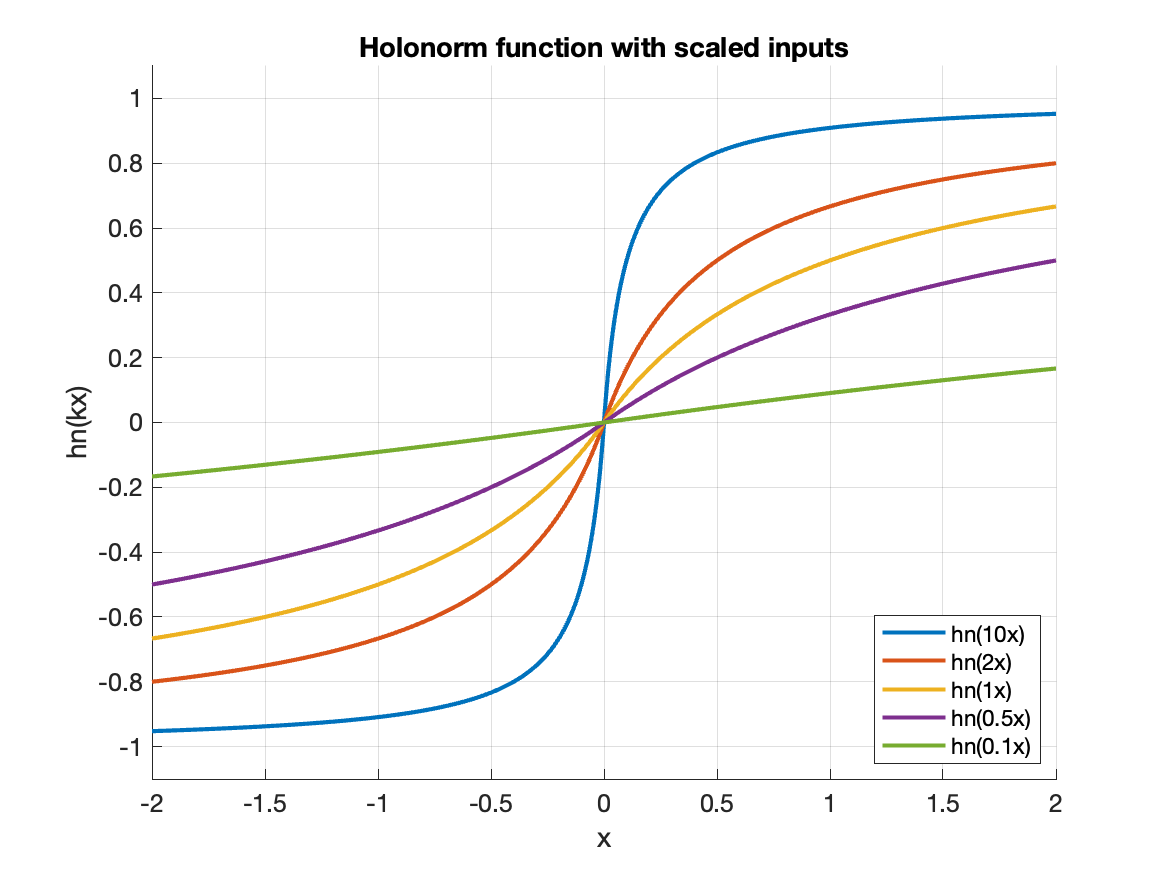}
\caption{Holonorm map \( hn(kx) = \frac{kx}{1 + |kx|} \) in 1D } \label{holonorm_1d_scaled}
\end{figure}

\begin{figure}
\centering
        \includegraphics[width=\linewidth]{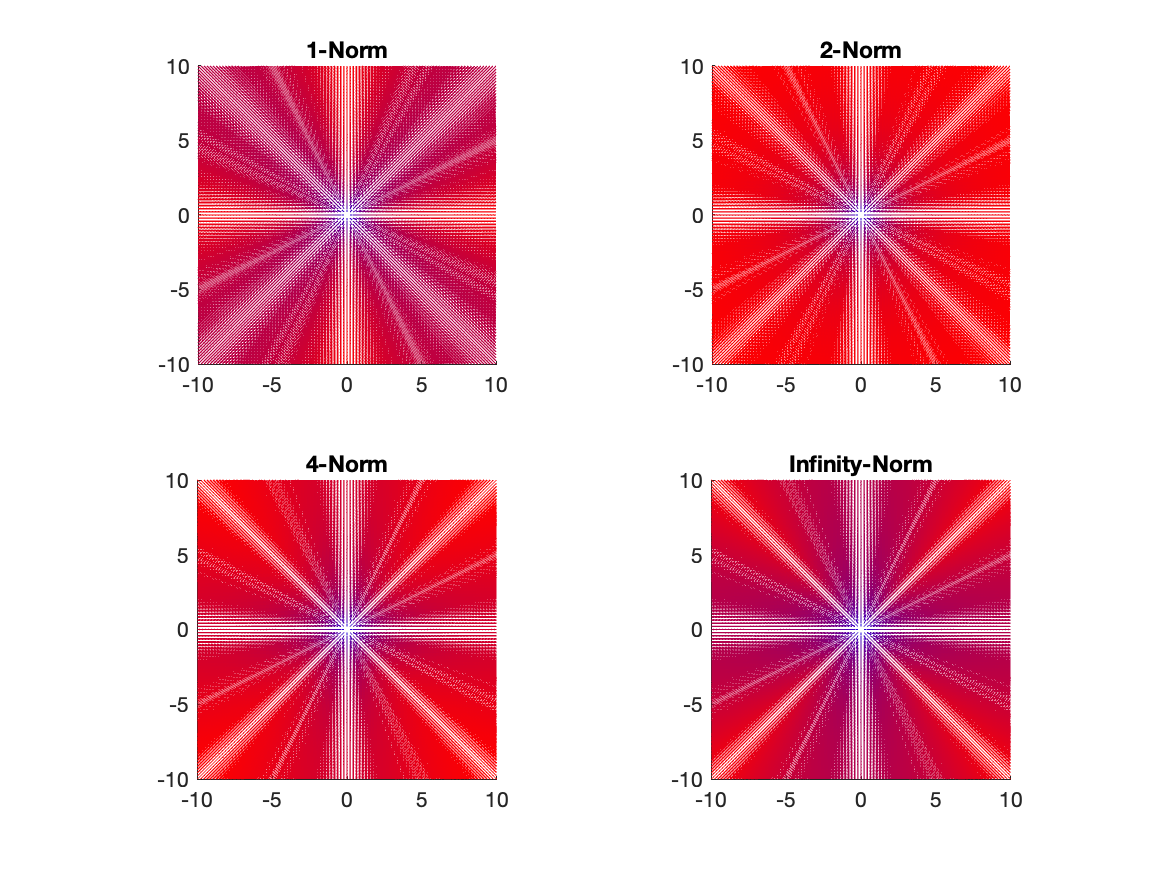}
\caption{Holonorm map \( hn(x,y) = \frac{(x,y)}{1 + \|(x,y)\|_p} \) in 2D.} \label{figholonorm01}
\end{figure}

\subsection*{Transformer Block Structure}

Each layer $\ell \in \{ 1, \dots, L\}$ applies the following composition to the input sequence:
\[
z^{(\ell)} = \left( \mathrm{id} + f \circ hn \right) \circ \left( \mathrm{id} + \mathrm{MHA} \circ hn \right)(z^{(\ell-1)}),
\]
where $\mathrm{MHA}$ is multihead self-attention: $\mathbb{R}^{N \times D} \to \mathbb{R}^{N \times D}$,
     $f$ is a two-layer feedforward map with \textbf{holonorm  activation}:
\[
    f(x) = W_2 \, hn(W_1 x + b_1) + b_2,
\]
    with $W_1 \in \mathbb{R}^{D \times d_{\text{ff}}}$, $W_2 \in \mathbb{R}^{d_{\text{ff}} \times D}$,
     $\mathrm{id}$ is the identity map (residual connection),
     $hn(x)$ is applied both before MHA as a layernorm and as the activation in the feedforward.

\subsection*{Output}

After $L$ layers:
\[
z_i^{\text{out}} = z_i^{(L)}, \quad i \in \{ 1, \dots, N\},
\]
which may be passed to a score head, decoder, or used in further computations.

\subsection*{$\tanh$ vs. HoloNorm in Audio Signal Processing}

In audio signal processing, \textbf{reconstruction fidelity} and the \textbf{preservation of orthogonality} between signal components are essential. Applications such as \textit{source separation}, \textit{beamforming}, \textit{PCA/ICA}, and \textit{multi-channel mixing} rely heavily on maintaining the angular relationships between different signal vectors. 

Why Orthogonality Matters? Orthogonality ensures that signals remain uncorrelated and independent, which is fundamental for accurate signal separation, clean spatial localization, Low cross-talk between channels, robust compression and denoising. 

\subsubsection*{Impact of \texttt{tanh} Transformation} The hyperbolic tangent function, applied element-wise, is defined as: \[ \tanh(x_i) = \frac{e^{x_i} - e^{-x_i}}{e^{x_i} + e^{-x_i}} \] While it maps each component to the range $(-1, 1)$, it does so independently for each element. This non-linear squashing has several adverse effects: Orthogonal vectors may no longer remain orthogonal after transformation,  Unrelated or independent signals can appear correlated post-transformation, angular relationships between vectors are non-linearly warped, making recovery difficult. 

The Inner product between data points is used generative machine intelligence   is used to compute correlation, angle, direction, or similarity score.
A basic example with the orthogonal directions is given by $(1,2,3), (12, 3, -6), (1,-2,1). $ Now, if  one applies tanh to these vectors element by element, then becomes correlated, they are orthogonal anymore. This creates a distortion and misleading separation, creates interference and many more other issues in text correlation, audio reconstruction, video reconstruction etc.  This behavior makes $\tanh$ unsuitable for audio reconstruction tasks, where structural fidelity is crucial.

\subsubsection*{Advantages of HoloNorm} The holonorm or radial \texttt{projectivenorm} is defined as: \[ \mathbf{x} \mapsto \frac{\mathbf{x}}{1 + \|\mathbf{x}\|} \] This transformation is \textbf{scale-sensitive} but \textbf{direction-preserving}. Unlike $\tanh$, it scales the entire vector uniformly rather than compressing individual components. 

\begin{itemize} \item \textbf{Preserves direction:} The angle between vectors is approximately preserved. \item \textbf{Maintains orthogonality:} Orthogonal signals remain nearly orthogonal after transformation. \item \textbf{No distortion of unrelated signals:} Independent vectors remain uncorrelated. \item \textbf{Suited for reconstruction:} Original signal geometry is largely preserved, allowing better recovery. \end{itemize}

 For audio tasks where signal independence and reconstruction are critical, $\tanh$ introduces nonlinearities that degrade performance by distorting angular and correlation structures. In contrast, holonorm or \texttt{projectivenorm} retains the original signal geometry and is thus preferable for applications requiring high-fidelity audio representation. Table \ref{scaledtable} provides several properties of $hn$ and their similarities and dissimilarities with tanh. Figure \ref{scaledfig} illustrates  $hn(15x)$ vs $tanh(x)$

\begin{table}[h!]
\centering
\begin{tabular}{|p{3cm}|p{4cm}|p{4cm}|}
\hline
\textbf{Property} & \textbf{Tanh} & \textbf{HoloNorm}: hn \\
\hline
\textbf{Map} & $x \mapsto \tanh(x)$ & $\mathbf{x} \mapsto \frac{\mathbf{x}}{1 + \|\mathbf{x}\|}=hn(x)$ \\
\hline
\textbf{Inverse} & $\tanh^{-1}(x) = \frac{1}{2} \ln\left( \frac{1 + x}{1 - x} \right)$ & $\mathbf{y} \mapsto \frac{\mathbf{y}}{1 - \|\mathbf{y}\|}$, with $\|\mathbf{y}\| < 1$ \\
\hline
\textbf{Derivative} & $\frac{d}{dx} \tanh(x) = 1 - \tanh^2(x)$ & $\frac{I(1 + \|\mathbf{x}\|) - \frac{\mathbf{x} \mathbf{x}^T}{\|\mathbf{x}\|}}{(1 + \|\mathbf{x}\|)^2}$ \\
\hline
\textbf{Primitive} & $\ln(\cosh(x)) + C$ & $|x| - \ln(1 + |x|) + C$ \\
\hline
\textbf{Preserves orthogonality} & No & Yes \\
\hline
\textbf{Preserves direction} & No & Yes \\
\hline
\textbf{Suitable for audio reconstruction} & Poor & Strong \\
\hline
\textbf{Component-wise distortion} & Yes & No \\
\hline
\textbf{Induced correlation between signals} & Yes & \textcolor{green}{No} \\
\hline
\textbf{Bounded output} & Yes $(-1,1)$ per component & Yes (open unit ball in $\mathbb{R}^n$) \\ \hline 
Hilbert spaces& Intractable due to infinite number of coefficients& compute the norm in the Hilbert space
\\
\hline
\end{tabular}
\caption{Comparison of $\tanh$ vs. HoloNorm in Audio Signal Processing} \label{scaledtable}
\end{table}

\begin{figure}
 
\begin{tikzpicture}
  \begin{axis}[
    width=12cm,
    height=8cm,
    xlabel={$x$},
    ylabel={Output},
    grid=both,
    axis lines=middle,
    legend style={at={(0.97,0.03)},anchor=south east},
    samples=500,
    domain=-2:2,
    enlargelimits=true
  ]
    \addplot[blue, thick] {15*x / (1 + abs(15*x))};
    \addlegendentry{$\frac{15x}{1 + |15x|}$}

    \addplot[red, dashed, thick] {tanh(5*x)};
    \addlegendentry{$\tanh(5x)$}
  \end{axis}
\end{tikzpicture}
\caption{$hn(15x) $ vs $tanh(5x)$ } \label{scaledfig}
\end{figure}
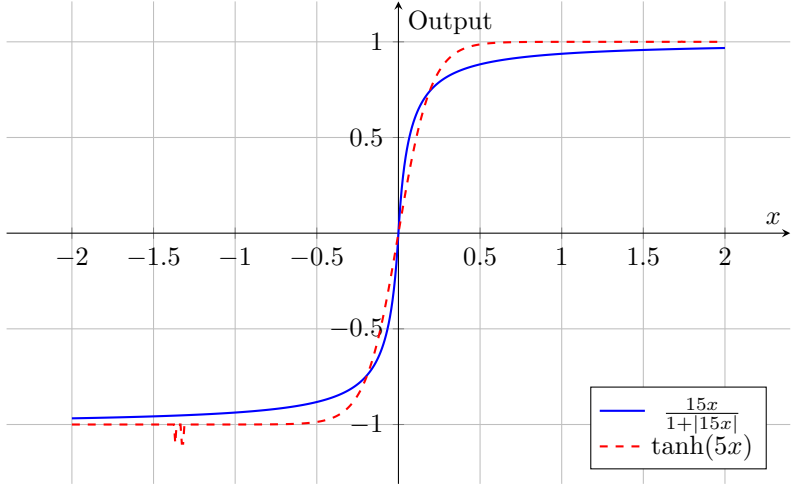

 \subsection{Multidimensional case}

\begin{proposition}
Let $x(t) \in \mathbb{R}^d$ evolve under the fractional stochastic differential equation
$$
dx(t) = \Theta(t)(\bar{m}(t) - x(t)) dt + \Sigma(t) dB^H(t), \quad x(0) = x_0,
$$
where:
\begin{itemize}
    \item $\Theta(t) \in \mathbb{R}^{d \times d}$ is a continuous, invertible matrix-valued function;
    \item $\bar{m}(t) \in \mathbb{R}^d$ is vector-valued and continuous;
    \item $\Sigma(t) \in \mathbb{R}^{d \times d}$ is a time-varying volatility matrix;
    \item $B^H(t) \in \mathbb{R}^d$ is a fractional Brownian motion with independent components and common Hurst index $H \in \left( \tfrac{1}{2}, 1 \right)$.
\end{itemize}

Then for any prescribed terminal mean $m_T^* \in \mathbb{R}^d$ and covariance matrix $\Sigma_T^* \succ 0$, there exists a choice of constant $\bar{m}$ and continuous matrix-valued function $\Sigma(t)$, such that:
$$
x(T) \sim \mathcal{N}(m_T^*, \Sigma_T^*).
$$
\end{proposition}
\paragraph{Proof.}

Let $U(t) \in \mathbb{R}^{d \times d}$ be the unique solution to the matrix ODE:
$$
\frac{d}{dt} U(t) = -\Theta(t) U(t), \quad U(0) = I_d.
$$
Then the explicit solution to the SDE is:
$$
x(t) = U(t)x_0 + U(t) \int_0^t U^{-1}(s)\Theta(s)\bar{m} \, ds + U(t) \int_0^t U^{-1}(s)\Sigma(s) \, dB^H(s),
$$
with the stochastic integral interpreted in the Young sense.

\textbf{Step 1: Mean Matching.}  
Define:
$$
M := \int_0^T U^{-1}(s)\Theta(s) \, ds.
$$
Then the mean of $x(T)$ is:
$$
\mathbb{E}[x(T)] = U(T)x_0 + U(T)M \bar{m}.
$$
To enforce $\mathbb{E}[x(T)] = m_T^*$, solve:
$$
\bar{m} = M^{-1}(U(T)^{-1} m_T^* - x_0).
$$

\textbf{Step 2: Covariance Matching.}  
Let:
$$
Z := U(T) \int_0^T U^{-1}(s)\Sigma(s) \, dB^H(s),
$$
so that:
$$
\operatorname{Cov}[x(T)] = \operatorname{Cov}[Z] = U(T) K U(T)^\top,
$$
where:
$$
K = H(2H - 1) \int_0^T \int_0^T f(s) |s - r|^{2H - 2} f(r)^\top \, ds \, dr, \quad f(s) := U^{-1}(s)\Sigma(s).
$$

Choose $f(s) = C \in \mathbb{R}^{d \times d}$ constant. Then:
$$
K = \lambda_T \cdot C C^\top, \quad \lambda_T := H(2H - 1) \int_0^T \int_0^T |s - r|^{2H - 2} \, ds \, dr.
$$
Thus:
$$
\operatorname{Cov}[x(T)] = \lambda_T \cdot U(T) C C^\top U(T)^\top.
$$
To match $\Sigma_T^*$, choose:
$$
C = \left( \lambda_T^{-1} U(T)^{-1} \Sigma_T^* U(T)^{-\top} \right)^{1/2}, \quad \Sigma(t) = U(t) C.
$$
 
With the above choice of constant mean $\bar{m}$ and time-varying volatility $\Sigma(t)$, the process satisfies:
$$
x(T) \sim \mathcal{N}(m_T^*, \Sigma_T^*).
$$
$\square$

\begin{proposition}

Let \( x(t) \in \mathbb{R}^d \) evolve under the linear fractional stochastic differential equation:
\[
dx(t) = \Theta(t) \left( \bar{m}(t) - x(t) \right) dt + \Sigma(t) dB^H(t), \quad x(0) = x_0 \in \mathbb{R}^d,
\]

Then, for each \( t \in [0, T] \), the process \( x(t) \) is Gaussian with mean \( \mu(t) \in \mathbb{R}^d \) and covariance matrix \( \mathbb{V}(t) \in \mathbb{R}^{d \times d} \), given by:

\[
m(t) = U(t) x_0 + U(t) \int_0^t U^{-1}(s) \Theta(s) \bar{m}(s) \, ds,
\]
\[
\mathbb{V}(t) = U(t) C(t) U(t)^\top,
\]
where \( U(t) \in \mathbb{R}^{d \times d} \) is the fundamental matrix of
\[
\dot{U}(t) = -\Theta(t) U(t), \quad U(0) = I_d,
\]
and
\[
C(t) = H(2H-1) \int_0^t \int_0^t f(s) f(r)^\top |s - r|^{2H - 2} \, ds \, dr, \quad f(s) := U^{-1}(s) \Sigma(s).
\]

The score function of the law of \( x(t) \), i.e., the gradient of the log-density with respect to \( x \), is given by:
\[
\nabla_x \log p(x(t)) = -\mathbb{V}(t)^{-1} (x(t) - m(t)).
\]
\end{proposition}
\textbf{Proof}: Since the SDE is linear with Gaussian noise, and the coefficients are deterministic and regular, the solution \( x(t) \) admits the explicit expression:
\[
x(t) = U(t) x_0 + U(t) \int_0^t U^{-1}(s) \Theta(s) \bar{m}(s) \, ds + U(t) \int_0^t U^{-1}(s) \Sigma(s) \, dB^H(s).
\]
Define:
\[
m(t) := \mathbb{E}[x(t)] = U(t) x_0 + U(t) \int_0^t U^{-1}(s) \Theta(s) \bar{m}(s) \, ds,
\]
\[
Z := \int_0^t U^{-1}(s) \Sigma(s) \, dB^H(s),
\]
\[
C(t) := \text{Cov}(Z) \text{ as given by Young integration theory.}
\]
Then:
\[
\text{Cov}[x(t)] = \mathbb{V}(t) = U(t) C(t) U(t)^\top.
\]
As \( x(t) \sim \mathcal{N}(m(t), \mathbb{V}(t)) \), the Gaussian score is:
\[
\nabla_x \log p(x(t)) = -\mathbb{V}(t)^{-1} (x(t) - m(t)),
\]
which completes the proof.

\hfill $\square$

 \section{Implementation}

 \subsection{Audio Fractional Diffusion}
 In this subsection we implement the mean-field-type fractional diffusion generative model. 
 \subsubsection*{Forward Fractional Diffusion Process}
 Given a initial audio, we use the forward fractional diffusion process as follows to end up with a mask $\mathcal{N}(m_T^*, (\sigma_T^*)^2).$
 
\[
dx(t) = \theta \left( \frac{m_T^* - x_0 e^{-\theta T}}{1 - e^{-\theta T}} - x(t) \right) dt + \sigma_T^* T^{-H} e^{\theta(T - t)} dB^H(t), \quad x(0) = x_0,
\]
which satisfies the ending mask condition:
\[
x(T) \sim \mathcal{N}(m_T^*, (\sigma_T^*)^2).
\]

\subsubsection*{Reverse Fractional Diffusion Process}
Given the ending point of the forward process, we build a mean-field-type backward process as follows:
\[
\begin{aligned}
dx(t) &= \left[ \theta \left( \frac{m_T^* - e^{-\theta T} x_0}{1 - e^{-\theta T}} - x(t) \right) + \frac{x(t) - m(t)}{t} 2H \right] dt \\
&\quad + \sigma_T^* T^{-H} e^{\theta(T - t)}\, d\bar{B}^H(t),
\end{aligned}
\]
where \( m(t) \)  is given by
\[
\begin{aligned}
m(t) &= e^{-\theta t} x_0 + \left( 1 - e^{-\theta t} \right) \left( \frac{m_T^* - e^{-\theta T} x_0}{1 - e^{-\theta T}} \right), 
\end{aligned}
\]

Alternatively, one can implement directly the mean-field type distribution $\mathcal{N}(m(t), v^2(t))$ with $v^2(t) = e^{2\theta(T - t)} \left( \frac{t}{T} \right)^{2H} (\sigma_T^*)^2$ to check the backward process of mean-field type.

Figure \ref{fig:processed_audio_waveforms8}  displays the  state trajectories  for $\theta=1, x_0=1.$ The  original audio is in Tommo-So Dogon language. The forward diffusion process leads to a pure mask  $\mathcal{N}(10,0.5)$ at $T=1.$ Then the reverse process takes this pure mask and progressively unmask it to discover the reconstructed audio signal. The gap between the two signals is up to the realization of a fractional Brownian path with $H=0.8$ on the given interval.
 \begin{figure}
        \centering
        \includegraphics[width=\linewidth]{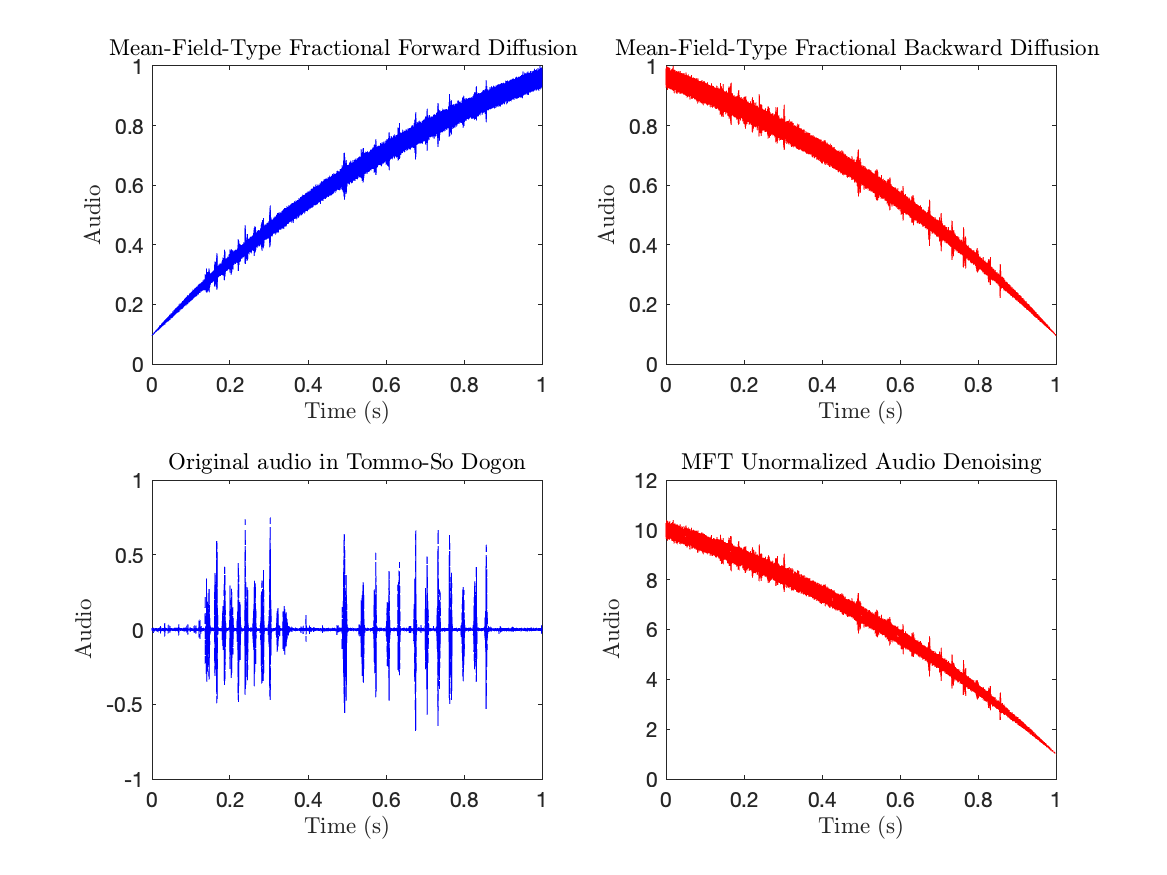}
        \caption{Mean-Field-Type Fractional Diffusion } \label{fig:processed_audio_waveforms8}
    \end{figure}

Figure \ref{fig:backward_process_density7}  displays the mean-field state distribution for $\theta=1, x_0=0.1.$ The  original audio is in Tommo-So Dogon language. The forward diffusion process leads to a pure mask  $\mathcal{N}(10,0.5)$ at $T=1.$ Then the reverse process takes this pure mask and progressively unmask it to discover the reconstructed audio signal. The gap between the two signals is up to the realization of a fractional Brownian path with $H=0.8$ on the given interval.
\begin{figure}
        \centering
        \includegraphics[width=0.4\linewidth]{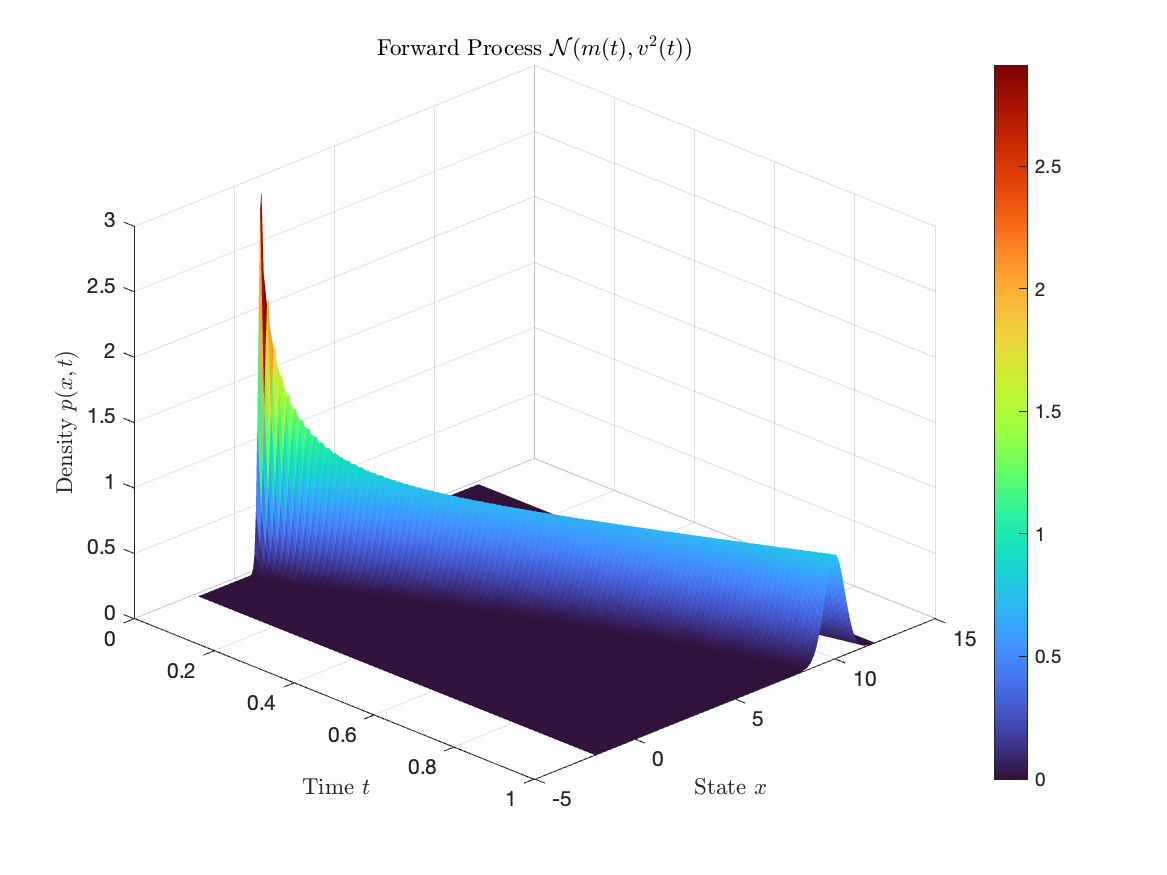}
        \includegraphics[width=0.4\linewidth]{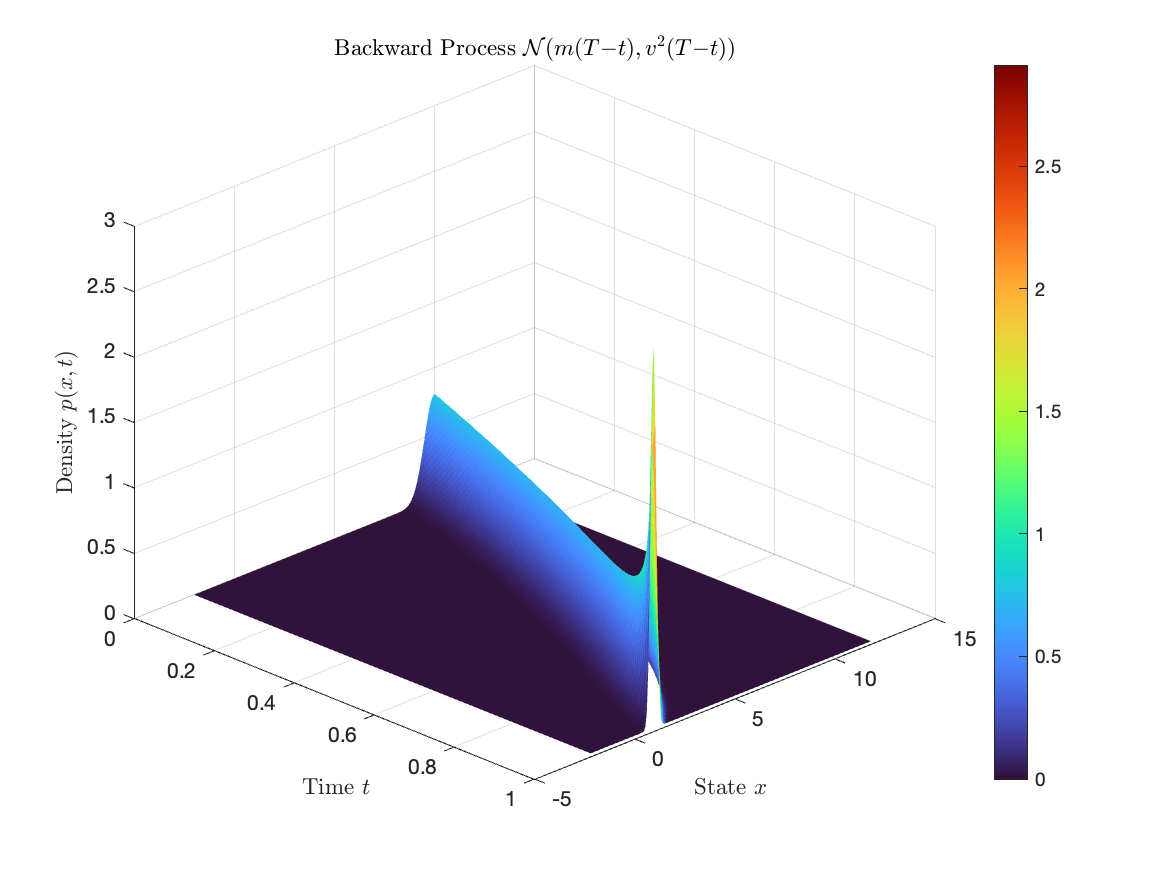}
        \caption{Mean-Field-Type Fractional  Diffusion: forward and backward } \label{fig:backward_process_density7}
    \end{figure}

\subsection{Textless Audio2Audio Fractional Diffusion Transformer}
\subsubsection{High-quality audio data for each language}
We consider a sample audio database as provided in Figures \ref{audiomatrixlabel1}, \ref{audiomatrixlabel2} and \ref{blockaudiomatrixlabel}. Table \ref{table:softmaxttea2a} provides a typical sample dictionary (but in audio.)
13 audiowords from Table \ref{table:softmaxttea2a2} are selected for a test in Amharic, Arabic, Bambara, Bariba, Bété,  Dendi,  Hassaniya, Hausa, Kinyarwanda, Kinyarwanda, Lingala, Malagasy, Moore,  Oshiwambo, Sango, Sereer, Senufo, Soninke, Soussou, Swahili, Tigrinya, Tommo-So Dogon, Waale, Wolof, Zarma, Zulu.
Amharic is a Semitic language of the Afroasiatic family, descended from Ge'ez. It is the official working language of Ethiopia and primarily spoken by the Amhara people in the central highlands. Its development is tightly linked to the Solomonic dynasty and the Ethiopian Orthodox Church.
Arabic, a Central Semitic language of the Afroasiatic family, originated in the Arabian Peninsula. Classical Arabic evolved into numerous modern dialects. It is an official language in over 20 countries across North Africa and the Middle East, including Egypt, Sudan, Mauritania, and Algeria.
Bambara (Bamanakan) is a Mande language in the Niger-Congo family, historically rooted in the Mali Empire. It serves as a lingua franca in Mali, spoken predominantly by the Bambara ethnic group but also widely used across ethnic boundaries in urban centers.
Bariba, or Baatonum, is a Gur language of the Niger-Congo family, associated with the Bariba people. It is spoken primarily in northeastern Benin and parts of western Nigeria. The language has historical ties to the Borgu Kingdoms.
Bété is a Kru language of the Niger-Congo family, spoken in southwestern Côte d'Ivoire. It is associated with the Bété people, who have historically resisted colonial influence, making the language a vehicle for cultural identity and oral literature.
Dendi is mainly spoken in Northern Benin, but also in other parts of Benin, and neighbouring countries. The Dendi people are the main group in the Departments of Alibori, Borgou, Donga, and Atakora.
In Nigeria, the Dendi people are found in Bordering States (Kebbi, Kwara, Niger, and Sokoto), and in other parts of Nigeria. They are usually referred by the Hausa name Dendawa (which is also used for the Songhai people).
Hassaniya Arabic is a Maghrebi Arabic dialect with strong Berber and Zenaga substrata, spoken primarily in Mauritania and parts of Mali, Senegal, and Morocco. It reflects deep nomadic traditions of the Moors and linguistic influence from Sub-Saharan Africa.
Hausa is a Chadic language of the Afroasiatic family, originating in northern Nigeria and southern Niger. It is the largest indigenous African language by number of speakers and a major lingua franca in West Africa, heavily influenced by Arabic through Islam.
Kinyarwanda is a Bantu language of the Niger-Congo family, mutually intelligible with Kirundi (Burundi). It is the national language of Rwanda. Its structure is highly agglutinative, with noun-class morphology.
Lingala is a Bantu language that developed as a trade language along the Congo River. It is widely spoken in the Democratic Republic of the Congo and the Republic of the Congo. Originating from Bobangi, Lingala was standardized by Christian missionaries.
Malagasy is an Austronesian language, specifically of the Malaya-Polynesian branch, and the national language of Madagascar. It traces its origins to the Ma’anyan language of southern Borneo, indicating ancient transoceanic migrations across the Indian Ocean.
Moore (Mòoré) is a Gur language spoken primarily by the Mossi people in Burkina Faso. It serves as a dominant language of trade and education in the country. It is structurally tonal, with syntactic patterns common to Voltaic languages.
Oshiwambo refers to a cluster of Bantu dialects, notably Kwanyama and Ndonga, spoken by the Ovambo people in northern Namibia and southern Angola. It is central to identity and governance in Namibia, where it is widely used in public communication.
Sango is a Ngbandi-based creole and national language of the Central African Republic. Initially a trade language along the Ubangi River, Sango became creolized through contact with French and other African languages during colonial rule.
Serer (or Sereer) is a language spoken by over  a million people in Senegal, the Gambia, and parts of Mauritania. It  features a rich oral tradition, tonal distinctions, and a complex noun class system.
Senufo comprises a group of Gur languages spoken in northern Côte d’Ivoire, southern Mali, and Burkina Faso. It is associated with the Senufo cultural complex, known for its caste systems, masked dances, and philosophical oral traditions.
Soninke is a Mande language of the Niger-Congo family, spoken by the Soninke people in Mali, Senegal, Mauritania, and Guinea. It traces back to the Ghana Empire and is characterized by noun classes, rich verb morphology, and tonal systems.
Soussou (Susu) is a Mande language spoken primarily in the coastal regions of Guinea. It serves as a lingua franca in Conakry and surrounding areas and is one of the national languages of Guinea.
Swahili (Kiswahili) is a Bantu language with significant Arabic influence, historically used as a trade language along the East African coast. It is widely spoken across Tanzania, Kenya, Uganda, and parts of Mozambique and the Democratic Republic of the Congo.
Tigrinya is a Semitic language of the Afroasiatic family, spoken in Eritrea and the Tigray region of Ethiopia. It descends from Ge’ez and shares grammatical features with Amharic, including a subject–object–verb order and gendered verb forms.
Tommo-So is one of the Dogon languages  spoken in  Mali and Burkina Faso. Characterized by minimal phonological tones and complex vowel harmony, it is used primarily in ritual, agriculture, and oral storytelling among the Dogon people.
Waale is a less-documented dialect of the Gur language cluster, related to Dagaare and Wali, spoken in parts of northwestern Ghana and Burkina Faso. Its status is primarily oral, with limited documentation in formal linguistics.
Wolof is a Senegambian language of the Niger-Congo family, dominant in Senegal and also spoken in the Gambia and Mauritania. It has been heavily influenced by Arabic and French and serves as the lingua franca in urban Senegal.
Zarma is a member of the Songhay language family, spoken mainly in western Niger. Closely related to Songhai, it has a subject–object–verb order and is widely used in education and media in Niger.
Zulu (isiZulu) is a Southern Bantu language of the Nguni group, spoken primarily in South Africa. It is characterized by click consonants inherited from Khoisan languages and a complex noun-class system. Zulu is one of South Africa's 11 official languages.
  
  \begin{table}[htb]
\centering

\begin{tabular}{|l|l|l||l|}
\hline 
 { \color{green} \faVolumeUp \ \ Tommo-So Dogon}  & { \color{pink}  \faVolumeUp \ \ Bobo San}   &  { \color{red}\faVolumeUp \ \ Bozo Jenaama Pondori}& \faVolumeUp \ \ Tieyaxo  \\   \hline 
{ \color{green} \faVolumeUp \ \  Yu }&  { \color{pink}  \faVolumeUp \ \  Douwo}   &  { \color{red}\faVolumeUp \ \  pin} & \faVolumeUp \ \  Jiye\\ 
{ \color{green} \faVolumeUp \ \   Sana Emme } &{ \color{pink}  \faVolumeUp \ \  Badouwa} &  { \color{red}\faVolumeUp \ \  Magna} & \faVolumeUp \ \ Magnon \\ 
{ \color{green} \faVolumeUp \ \  Emme }  &{ \color{pink}   \faVolumeUp \ \ Hamoro} &    { \color{red}\faVolumeUp \ \  Kontonron}& \faVolumeUp \ \ Xongolo \\ 
{ \color{green} \faVolumeUp \ \  Saa }&  { \color{pink} \faVolumeUp \ \  Yinu}  &  { \color{red}\faVolumeUp \ \  Fooroo}& \\ 
{ \color{green} \faVolumeUp \ \  Odu }&  { \color{pink}  \faVolumeUp \ \  Wan}  &   { \color{red}\faVolumeUp \ \  Sey}& \faVolumeUp \ \ Chiyen \\  
{ \color{green} \faVolumeUp \ \  Minai }&  { \color{pink}  \faVolumeUp \ \  Muan}  &  { \color{red}\faVolumeUp \ \  Songon} & \faVolumeUp \ \ Tchiye\\ 
{ \color{green} \faVolumeUp \ \  Inai-som }&  { \color{pink}  \faVolumeUp \ \ hann-tio}  &  { \color{red}\faVolumeUp \ \  Mienchye}& \faVolumeUp \ \ Miyemuchiye \\ 
{ \color{green} \faVolumeUp \ \  Diognou-diognouna }& { \color{pink} \faVolumeUp \ \  Tinso}  &  { \color{red}\faVolumeUp \ \  Youmou sabareya}& \faVolumeUp \ \ Jiyereya \\ 
{ \color{green} \faVolumeUp \ \  Antolna }& { \color{pink} \faVolumeUp \ \  ? }  &   { \color{red}\faVolumeUp \ \  Songon sarayan} & \\ 

{ \color{green} \faVolumeUp \ \  Tenne }&  { \color{pink} \faVolumeUp \ \ Boui} &   { \color{red}\faVolumeUp \ \  Tainai}& \faVolumeUp \ \ Tinder \\ 
{ \color{green} \faVolumeUp \ \ Togu }&  { \color{pink} \faVolumeUp \ \  Zun}  &   { \color{red}\faVolumeUp \ \  Gnama} & \faVolumeUp \ \ Jangan\\ 

{ \color{green} \faVolumeUp \ \  Kemme }&  { \color{pink}  \faVolumeUp \ \  Boiri}  &   { \color{red}\faVolumeUp \ \  ?} & \faVolumeUp \ \ Boli\\ 
{ \color{green} \faVolumeUp \ \  Gouyai }&  { \color{pink} \faVolumeUp \ \   Nanu}  &  { \color{red}\faVolumeUp \ \  Djigina} & \faVolumeUp \ \ Fugudu\\ 
{ \color{green} \faVolumeUp \ \  Ara }&   { \color{pink} \faVolumeUp \ \  Mi} &   { \color{red}\faVolumeUp \ \  Dougonpin}& \faVolumeUp \ \ Jua \\ 
{ \color{green} \faVolumeUp \ \  Noumou }& { \color{pink}  \faVolumeUp \ \  Wio}  &   { \color{red}\faVolumeUp \ \  Sapla} & \faVolumeUp \ \ Boolo\\

{ \color{green} \faVolumeUp \ \  Togu }&  { \color{pink}  \faVolumeUp \ \  Chiozanu} &   { \color{red}\faVolumeUp \ \  ?}& \faVolumeUp \ \ Togo \\ 
{ \color{green} \faVolumeUp \ \  Gaou }&  { \color{pink} \faVolumeUp \ \   Taimai} &   { \color{red}\faVolumeUp \ \  Yabaron}& \faVolumeUp \ \ Jawagon \\ 
{ \color{green} \faVolumeUp \ \  Coun-na }&  { \color{pink}  \faVolumeUp \ \  Tiebe} &   { \color{red}\faVolumeUp \ \  Biou}& \faVolumeUp \ \ Sanxere \\ 
{ \color{green} \faVolumeUp \ \  Saa-di }&  { \color{pink} \faVolumeUp \ \   Gna} &   { \color{red}\faVolumeUp \ \  Doo}& \faVolumeUp \ \ Lolo \\ 
 \hline 
\end{tabular}
\caption{High-quality Audio-to-audio database }
\label{table:softmaxttea2a}
\end{table}

  \begin{table}[htb]
\centering
\begin{tabular}{|l|}
\hline 
 { \color{green} \faVolumeUp \ \ Tommo-So Dogon}    \\   \hline 

  { \color{green} \faVolumeUp \ \ Mindju}   \\   \hline 

   { \color{green} \faVolumeUp \ \ Oro}    \\   \hline 

   { \color{green} \faVolumeUp \ \ Youlo}    \\   \hline  

  { \color{green} \faVolumeUp \ \ Omolu}    \\   \hline 

  { \color{green} \faVolumeUp \ \ Se-nguai}    \\   \hline 

 { \color{green} \faVolumeUp \ \ Kambé}    \\   \hline 

 { \color{green} \faVolumeUp \ \ Bié}    \\   \hline 

 { \color{green} \faVolumeUp \ \ Molo}    \\   \hline

 { \color{green} \faVolumeUp \ \ Saa}    \\   \hline
 \hline 
 
 { \color{green} \faVolumeUp \ \ Dana-siai }   \\   \hline

{ \color{green} \faVolumeUp \ \ Siaiguiriaimai }    \\   \hline

{ \color{green} \faVolumeUp \ \ Diaou }    \\   \hline

{ \color{green} \faVolumeUp \ \ Padiaiondo } \\   \hline
\end{tabular}
\caption{High-quality Audio-to-Audio database }
\label{table:softmaxttea2a2}
\end{table}
%
\begin{figure}
\begin{tabular}{|l|l|l|l|l|} \hline
\includegraphics[width=0.18\textwidth]{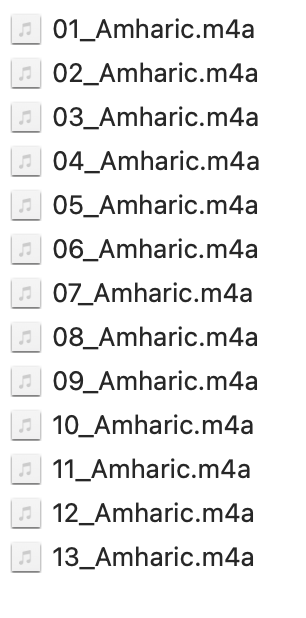} &
\includegraphics[width=0.16\textwidth]{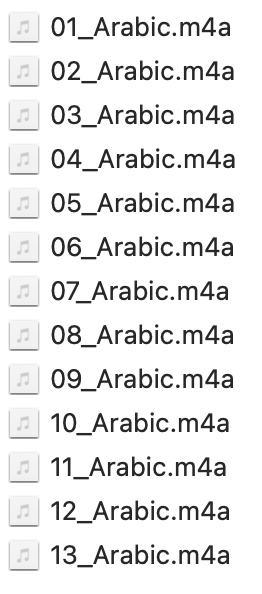} &
\includegraphics[width=0.17\textwidth]{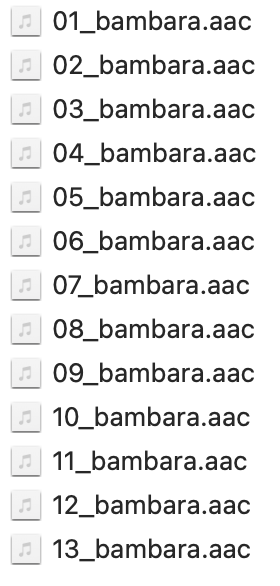} &
\includegraphics[width=0.20\textwidth]{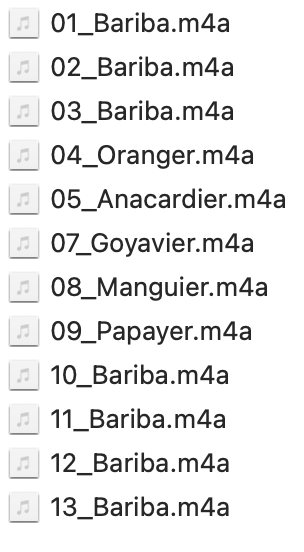} &
\includegraphics[width=0.22\textwidth]{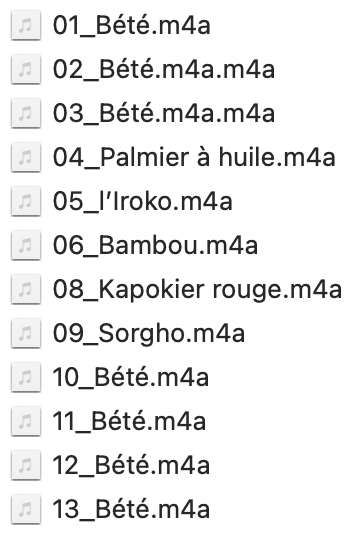} \\ \hline
\includegraphics[width=0.18\textwidth]{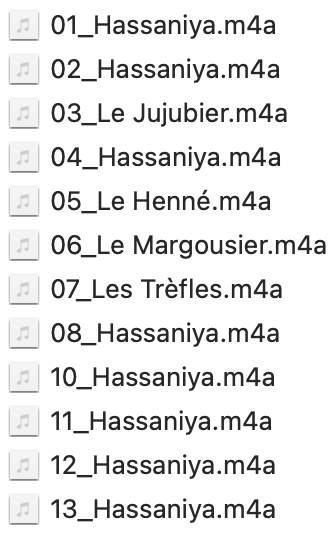} &
\includegraphics[width=0.14\textwidth]{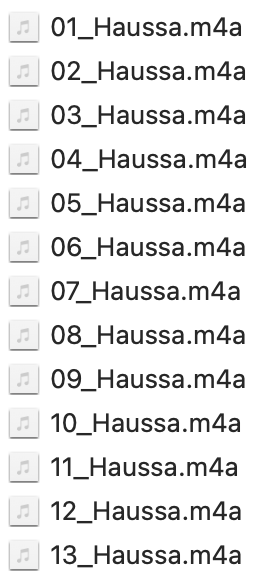} &
\includegraphics[width=0.18\textwidth]{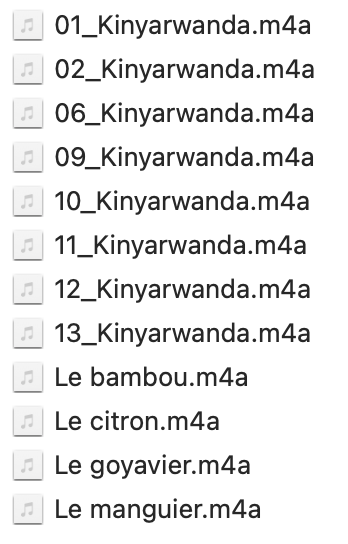} &
\includegraphics[width=0.18\textwidth]{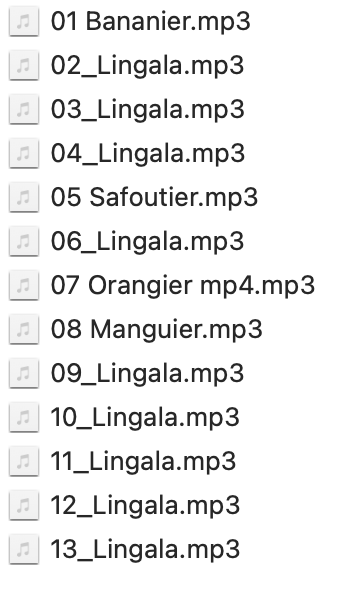} &
\includegraphics[width=0.18\textwidth]{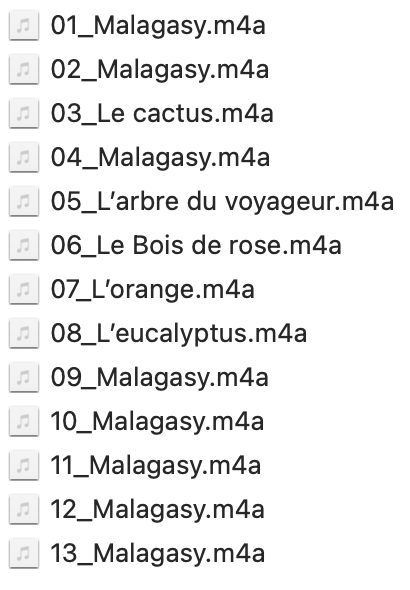} \\  \hline
\end{tabular}
\caption{An audio database matrix for Africa: part 1}
\label{audiomatrixlabel1}
\end{figure}
\begin{figure}
\begin{tabular}{|l|l|l|l|l|} \hline
\includegraphics[width=0.18\textwidth]{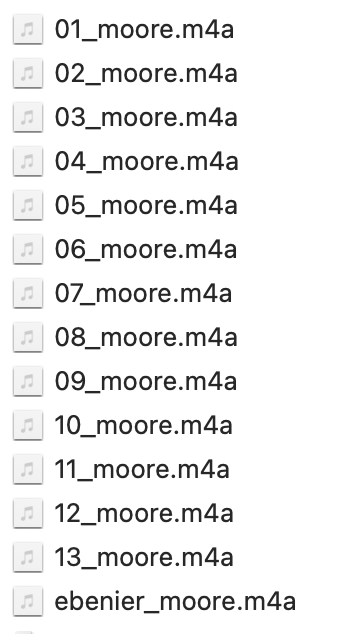} &
\includegraphics[width=0.18\textwidth]{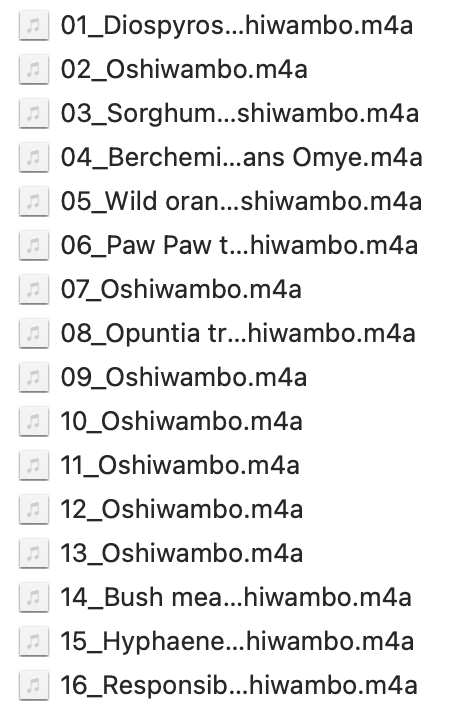} &
\includegraphics[width=0.18\textwidth]{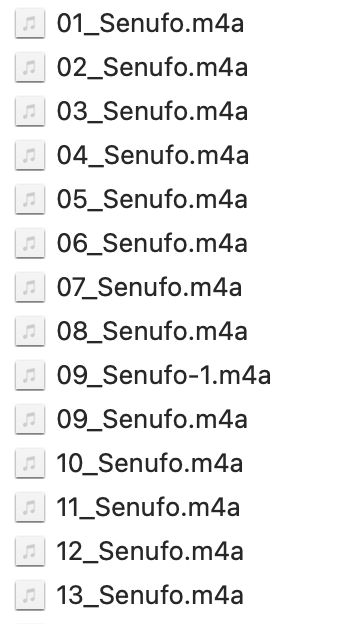} &
\includegraphics[width=0.18\textwidth]{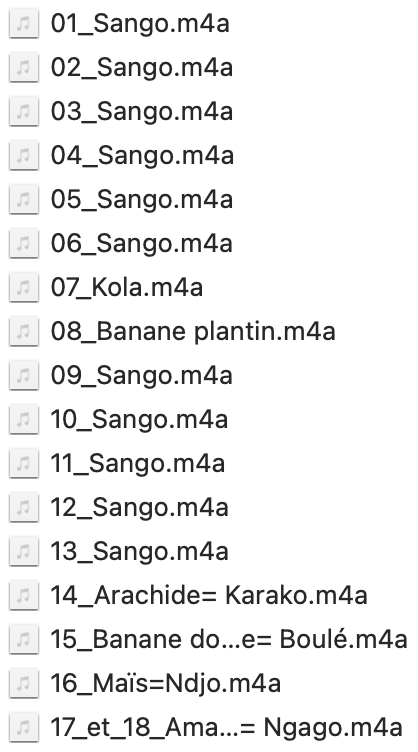} &
\includegraphics[width=0.18\textwidth]{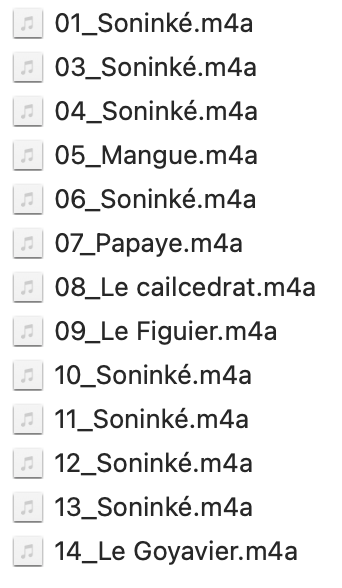} \\  \hline
\includegraphics[width=0.19\textwidth]{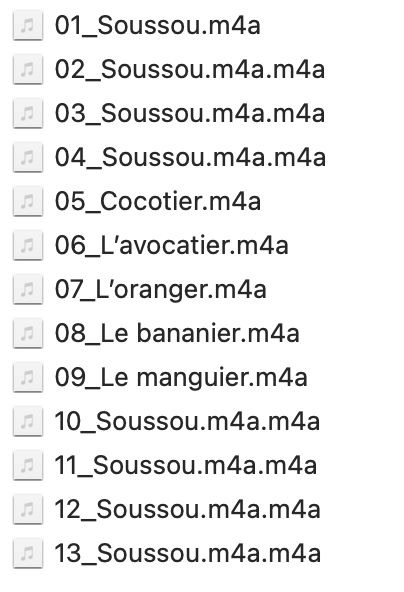} &
\includegraphics[width=0.23\textwidth]{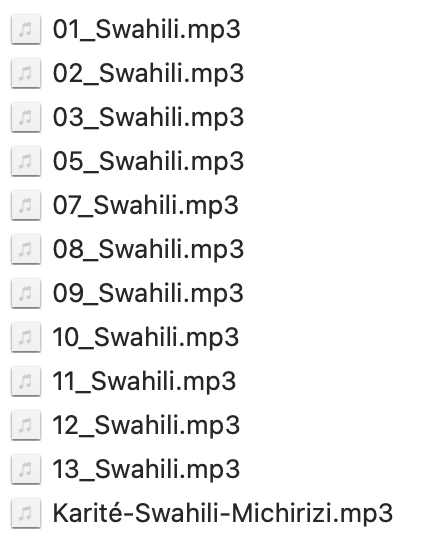} &
\includegraphics[width=0.18\textwidth]{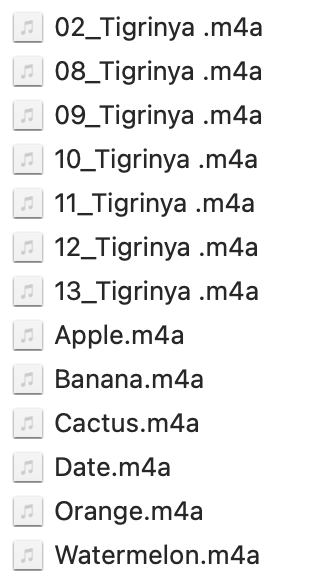} &
\includegraphics[width=0.18\textwidth]{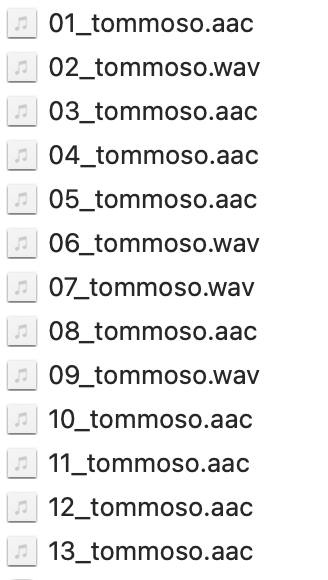} &
\includegraphics[width=0.18\textwidth]{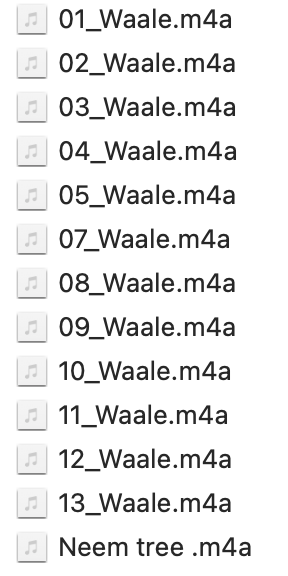} \\  \hline
\includegraphics[width=0.18\textwidth]{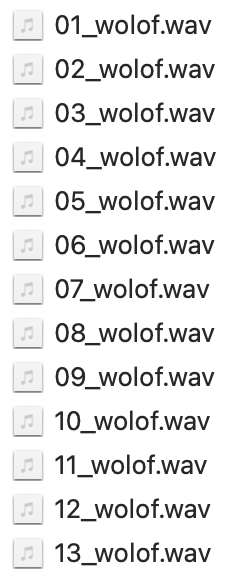} &
\includegraphics[width=0.19\textwidth]{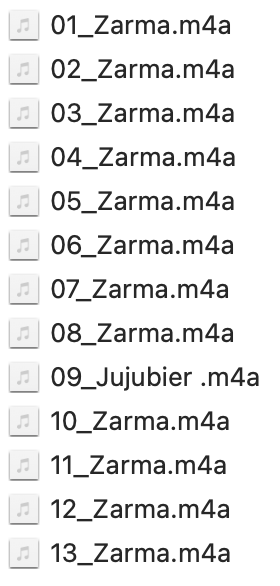} &
\includegraphics[width=0.17\textwidth]{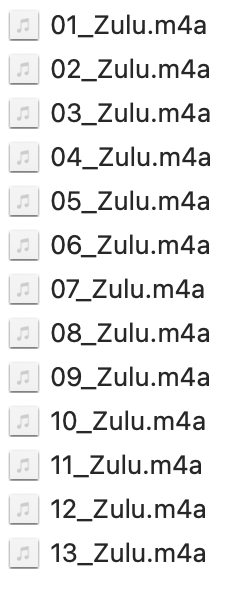} &
\includegraphics[width=0.19\textwidth]{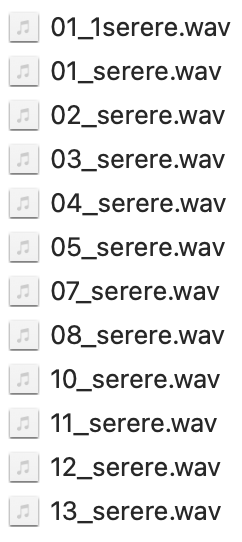} &
\includegraphics[width=0.11\textwidth]{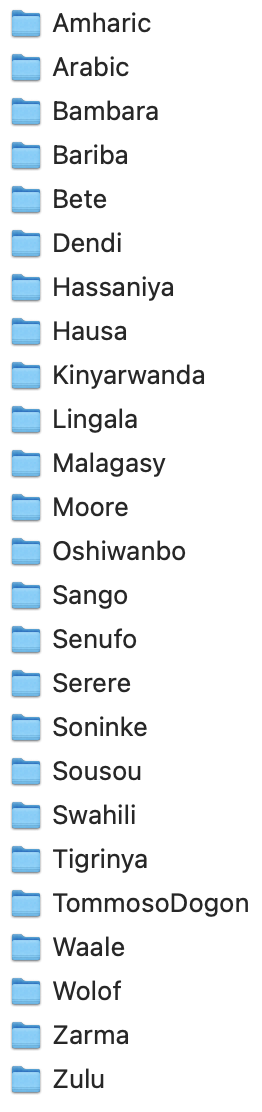}\\  \hline
\end{tabular}
\caption{An audio database matrix for Africa: part 2}
\label{audiomatrixlabel2}
\end{figure}
\begin{figure}
\includegraphics[width=0.9\textwidth]{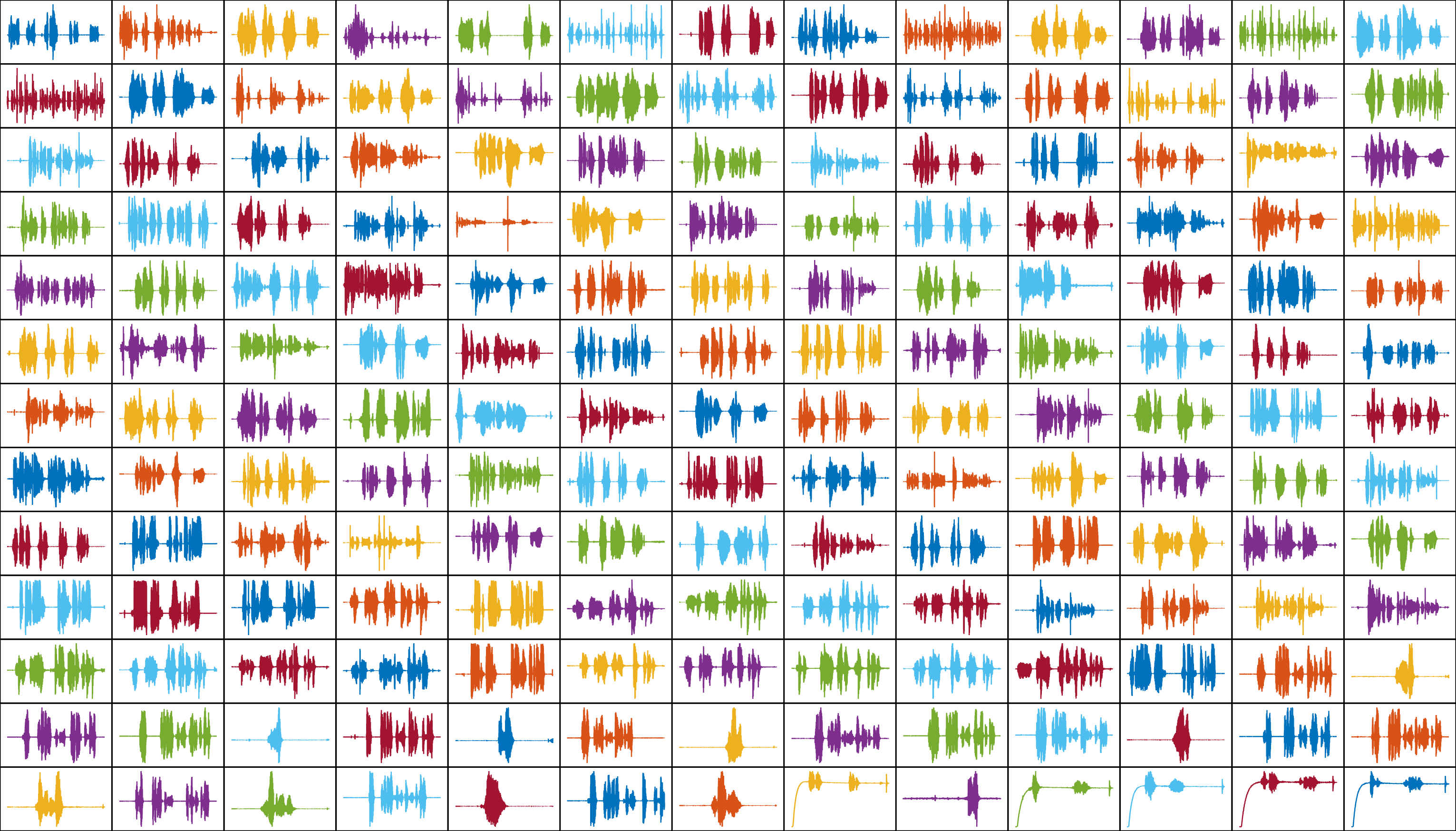} 
\caption{An example of blockaudio matrix.}
\label{blockaudiomatrixlabel}
\end{figure}

\subsubsection{Wavelet analysis for labeled  audio data }
We label high-quality data from native speakers (see sample audio data in Figure \ref{blockaudiomatrixlabel}) content-wise and assign a code/number/hash  to each file. For example, $yy\_$language$\_$name. The number $yy$ is the row number of the table. The column determines  a language. With this dictionary visualization approach, audio translation means changing a column while staying on the same row.  For each of the 13 words we can plot the original audio signal and save it as $yy\_$language$\_$name. Each column is a different  folder. We can now a wavelet analysis file per file. A sample wavelet analysis is provided in Figure \ref{waveletdogon}. We likinf the wavelet output to the raw audio data and keep the original source. During the reconstructed step, the reconstructed signal will be compared to the original signal, if they are closed enough, the original signal will be chosen as output making no non-native audio. By doing so the accent and dialects are preserved as they are from the sources.

\begin{figure}
\includegraphics[width=0.9\textwidth]{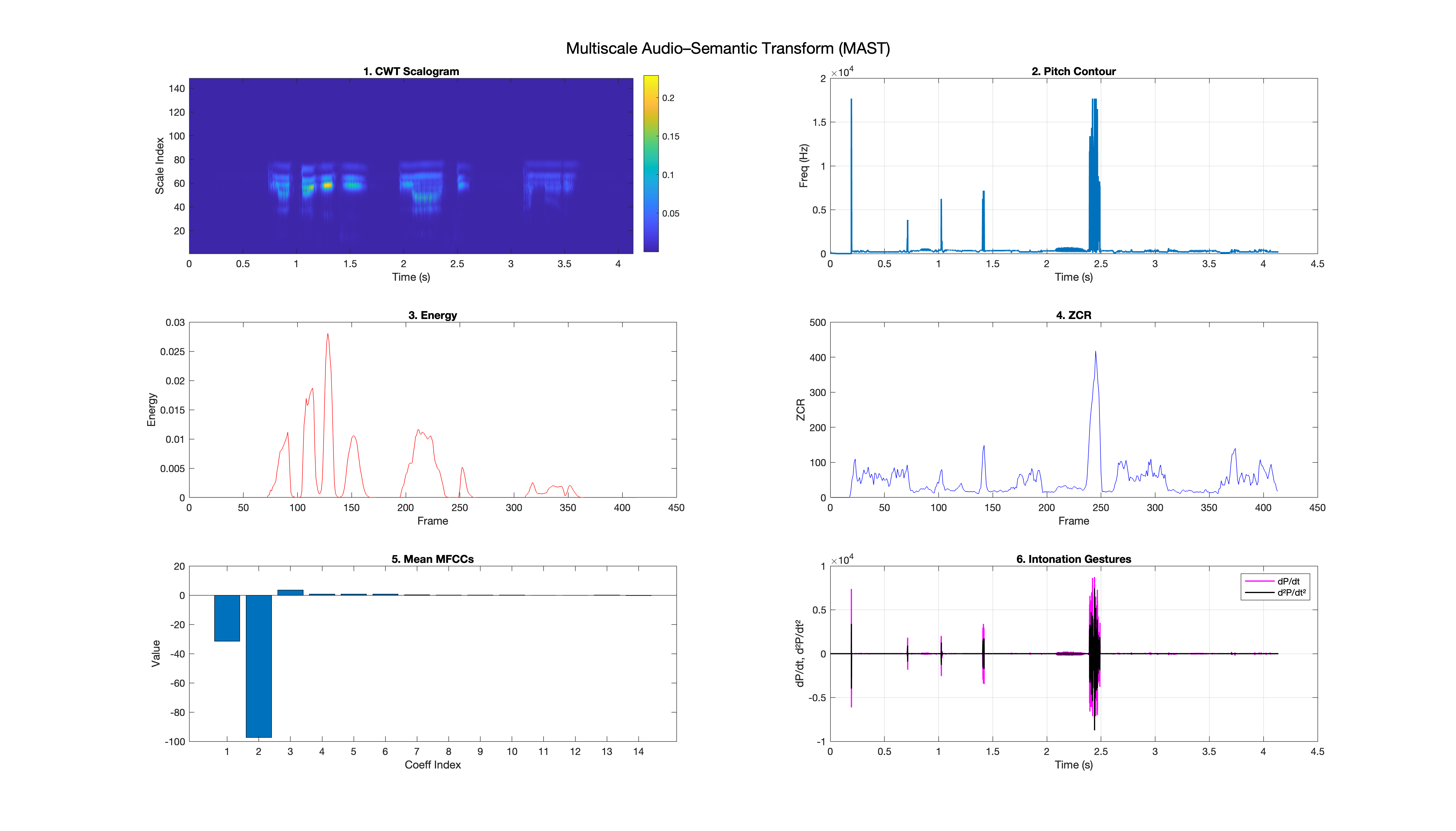} 
\caption{MAST wavelet analysis of an audio in Tommo-So Dogon.}
\label{waveletdogon}
\end{figure}
\subsubsection{Audio-Transformer }
As in Figure  \ref{waveletdogon}, the raw audio is transformed into a vector after the MAST wavelet analysis of the audio. A diffusion process followed by a transformer step can be applied to this mast vector. The resulting output of the transformer is then compared with the audio database using inverse wavelet transform.



\section{Conclusion}
We conclude with some limitations, ethical considerations, and future directions for self-supervised textless voice to voice  technology.  
This work presents a foundational shift in machine intelligence: a fully textless audio-to-audio framework designed to serve the 700 millions of audio-literate individuals, particularly in low-text-resource settings. 
By unifying multiscale acoustic representations, fractional diffusion dynamics, and mean-field-type interactions, we enable direct, semantically rich communication between humans and AI, and among AI agents 
themselves, entirely in the audio domain. Our architecture, grounded in the MAST, is robust to tonal, prosodic, and speaker variability. The integration of 
agentic behavior and risk-awareness further allows these systems to adaptively manage uncertainty in real-time conversational settings. 
To analyze and guide these interactions, we employ mean-field-type game theory, a mathematical formalism well-suited to modeling collective behaviors among  several agents under uncertainty. Each AI agent optimizes its audio output policy not in isolation, but relative to a distribution of the behaviors of  the environment and the shared resources by others. The interaction between human voices and AI-generated speech is governed by dynamic strategies conditioned on prosody, speaker identity, and intent, all encoded through the MAST. Each agent incorporates risk-awareness by minimizing audio interaction cost under expectile-based value-at-risk, capturing asymmetries in communicative errors (tone misinterpretation vs. content mismatch). The underlying fractional diffusion models allow agents to adaptively update these risk profiles as conversations evolve.  These interactions, whether between humans and AI, or among AI agents, are embedded within a fractional mean-field-type diffusion architecture.

While promising, our framework faces several practical, scientific, and ethical challenges that must be addressed for responsible deployment. High-quality, annotated audio corpora for many low-text-resource audio-rich languages poses a significant hurdle, particularly during initial training phases, despite our reliance on self-supervised learning. The computational demands of fractional diffusion models and wavelet-based transforms may limit real-time scalability, especially on low-power or edge devices.  The evaluation remains a persistent limitation, as universally accepted benchmarks for textless audio translation and dialogue are lacking; current assessments often rely on subjective human judgment, which can be inconsistent from one village to another. Ethically, the deployment of agentic audio-AI systems demands rigorous safeguards. Linguistic sovereignty must be upheld to ensure that audio-AI tools empower rather than homogenize linguistic communities, with co-development involving native speakers as a prerequisite. Given the potential misuse of continuous listening systems for passive surveillance or behavioral manipulation, all applications must respect consent, privacy, and cultural norms. There is also a risk of AI agents learning and reinforcing biases embedded in speech patterns linked to culture, gender, accent, or emotional expression, necessitating continuous auditing and fairness interventions. As these agents gain autonomy and operate in sensitive domains such as finance, healthcare, or legal consultation, robust accountability frameworks must be established to clarify responsibility for their actions and decisions.

\subsection*{Affiliations}
 Hamidou Tembine  is with  Department of Electrical Engineering and Computer Science, School of Engineering, UQTR, Quebec, Canada. He is also with AI Mali, Timadie, Grabal, Guinaga, MFTG, SK1 Sogoloton, WETE, CI4SI and TF.  Issa Bamia, Massa NDong, Oumar Issiaka Traor\'e, Moussa Sanogo are  with AI Mali. Bakary Coulibaly  is with School of Economics and Management, China University of Geosciences,  Moussa Traor\'e is with OSU THETA Kaboratory, Marie and Louis Pasteur University.   Mamadou Eric Sangar\'e is with Interdisciplinary Laboratory for Applied Research, Universiapolis, International University of Agadir, Morocco,  Salif Kant\'e is with Tomsk State University of Control Systems and Radioelectronics (TUSUR), Tomsk, Russia and AI Mali,  Daryl Noupa Yongueng and  Mamadou Lamine Doumbia are with UQTR,  Hafiz Tiomoko Ali is with Samsung Research UK, Malik Tiomoko is with Huawei Noah's Ark Lab, Fréjus Laleye is with Opscidia,  Boualem Djehiche is with Department of mathematics, KTH, Stockholm, Sweden,   Wesmanegda Elisee Dipama, Idris Baba Saje, Hammid Mohammed  Ibrahim, Marie Coursel Nini Nahazwe and Abdul-Latif  Siita, Mariam Serine Jeridi, Mutiyamuogo Parfait Mupenge, Lekoueiry Dehah,   Abdoul Aziz Bio Sidi D Bouko, Wilfried Franceslas Zokoue, Odette Richette Sambila, Alina RS Mbango, Mady Diagouraga , Oumarou  Moussa Sanoussi, Gizachew Dessalegn, Bintou Laeticia Audrey Coulibaly   are with China University of GeoSciences in Beijing,  Moumini Sanogo and Teddy Nelvy Dieu Merci Kouka, Mady Diagouraga are with Beijing Language and Culture University, Haine Mhlongo is with Beijing Film Academy. Mohamed Lamine Samoura is with Tianjin Normal University.


\bibliographystyle{plain} 
\bibliography{a2a} 

\end{document}